\documentclass[11pt,a4paper]{article}
\pdfoutput=1 

\usepackage[margin=1in]{geometry}

\usepackage{amsmath}
\usepackage{amsthm}
\usepackage{amsfonts}
\usepackage{amssymb}

\usepackage[pdfencoding=auto,psdextra,hidelinks]{hyperref}
\usepackage[svgnames]{xcolor}
\hypersetup{colorlinks={true},urlcolor={blue},linkcolor={DarkBlue},citecolor=[named]{DarkGreen}}
\usepackage[authoryear,square]{natbib}
\usepackage{doi}

\usepackage{microtype}
\usepackage[capitalise,nameinlink,noabbrev]{cleveref}
\crefname{ineq}{Inequality}{Inequalities}
\creflabelformat{ineq}{#2{\upshape(#1)}#3}

\usepackage{crossreftools}
\usepackage{mathptmx}
\usepackage[utf8]{inputenc}
\usepackage[T1]{fontenc}
\usepackage[english]{babel}
\usepackage{csquotes}
\usepackage{commath}
\usepackage{bm}
\usepackage{xfrac}
\usepackage{mathtools}
\usepackage{authblk}
\usepackage{bookmark}
\usepackage{float}
\usepackage{booktabs}
\usepackage{subcaption}
\usepackage{tikz}
\usepackage{enumerate}
\usetikzlibrary{calc,positioning,shapes,arrows}
\usepackage{array}

\usepackage{mdframed}

\usepackage{enumitem}

\newlist{conditions}{enumerate}{2}
\setlist[conditions]{label={\arabic*.}}
\crefname{conditionsi}{Condition}{Conditions}

\newlist{assumptions}{enumerate}{2}
\setlist[assumptions]{label={\arabic*.}}
\crefname{assumptionsi}{Assumption}{Assumptions}

\newcommand{\transpose}{\ensuremath{\mathsf{T}}}

\newcommand{\Simplex}{\ensuremath{\Delta}}

\newcommand{\eps}{\ensuremath{\varepsilon}}
\newcommand{\NN}{\ensuremath{\mathbb{N}}}

\newcommand{\RR}{\ensuremath{\mathbb{R}}}
\newcommand{\RRnn}{\ensuremath{\mathbb{R}_{\geq0}}}
\newcommand{\RRp}{\ensuremath{\mathbb{R}_{>0}}}

\DeclareMathOperator*{\argmax}{arg\,max}

\newcommand{\PPAD}{\ensuremath{\mathrm{PPAD}}}
\newcommand{\FIXP}{\ensuremath{\mathrm{FIXP}}}
\newcommand{\linearFIXP}{\ensuremath{\mathrm{Linear}\text{-}\mathrm{FIXP}}}
\newcommand{\FIXPa}{\ensuremath{\mathrm{FIXP}_a}}

\newcommand{\poly}{\ensuremath{\mathrm{poly}}}

\newcommand{\Heaviside}{\ensuremath{\operatorname{H}}}
\newcommand{\conv}{\ensuremath{\operatorname{Conv}}}

\newcommand{\calH}{\ensuremath{\mathcal{H}}}

\newcommand{\support}{\ensuremath{\operatorname{Supp}}}

\newcommand{\Sol}{\ensuremath{\operatorname{Sol}}}

\newcommand{\Exp}{\operatorname*{E}}

\newcommand{\Feas}{\ensuremath{\operatorname{Feas}}}
\newcommand{\Opt}{\ensuremath{\operatorname{Opt}}}
\newcommand{\size}{\ensuremath{\operatorname{size}}}
\newcommand{\LP}{\ensuremath{\operatorname{LP}}}
\newcommand{\CP}{\ensuremath{\operatorname{CP}}}
\newcommand{\fix}{\ensuremath{\operatorname{\mathsf{Fix}}}}

\theoremstyle{definition}
\newtheorem{definition}{Definition}[section]
\newtheorem{remark}{Remark}
\newtheorem{example}{Example}

\newtheorem{obstacle}{Obstacle}

\theoremstyle{plain}
\newtheorem{theorem}{Theorem}[section]
\newtheorem{proposition}{Proposition}
\newtheorem{lemma}{Lemma}
\newtheorem{corollary}{Corollary}
\newtheorem{claim}{Claim}

\title{FIXP-membership via Convex Optimization: \\Games, Cakes, and Markets}

\author[1]{Aris Filos-Ratsikas}
\author[2]{Kristoffer Arnsfelt Hansen}
\author[2]{Kasper H{\o}gh}
\author[3]{Alexandros Hollender}
\affil[1]{University of Edinburgh, UK}
\affil[2]{Aarhus University, Denmark}
\affil[3]{University of Oxford, UK}

\date{} 

\begin{document}

\maketitle 

\begin{abstract}
We introduce a new technique for proving membership of problems in FIXP -- the class capturing the complexity of computing a fixed-point of an algebraic circuit. Our technique constructs a ``pseudogate'' which can be used \emph{as a black box} when building FIXP circuits. This pseudogate, which we term the ``OPT-gate'', can solve most convex optimization problems. Using the OPT-gate, we prove \emph{new} FIXP-membership results, and we \emph{generalize} and \emph{simplify} several known results from the literature on fair division, game theory and competitive markets.

In particular, we prove complexity results for two classic problems: computing a market equilibrium in the Arrow-Debreu model with general concave utilities is in FIXP, and computing an envy-free division of a cake with very general valuations is FIXP-complete. We further showcase the wide applicability of our technique, by using it to obtain simplified proofs and extensions of known FIXP-membership results for equilibrium computation for various types of strategic games, as well as the pseudomarket mechanism of Hylland and Zeckhauser.
\end{abstract}

\clearpage

\section{Introduction}

Equilibria, i.e., stable states of some dynamic process or environment \citep{yannakakis2009equilibria}, appear in several classic applications in economics and computer science. Prominent examples include the Nash equilibrium \citep{Nash50}, which captures the stable outcome of deliberation between strategic agents, as well as the competitive equilibrium \citep{arrow1954existence}, which corresponds to a market-clearing outcome after the adjustment of prices based on demand and supply. These equilibria can most often be captured by fixed points of functions, i.e., points $x$ for which $f(x) = x$. For instance, Nash's existence theorem, i.e., that every strategic game has a mixed Nash equilibrium, was famously proven using Brouwer's fixed point theorem \citep{MA:Brouwer1911}. 

The computational class FIXP was defined by \citet{SICOMP:EtessamiY10} to capture the complexity of fixed point problems, and in particular those related to Brouwer's fixed point theorem. These problems are \emph{total search problems}, i.e., problems for which a solution is guaranteed to exist via Brouwer's (or some other) fixed point theorem, and for which we aim to find such a solution.
Indeed, the class has been successful in that regard, with interesting problems related to game theory \citep{SICOMP:EtessamiY10} and competitive markets \citep{SICOMP:EtessamiY10,garg2017settling,ChenPY17-non-monotone-markets} among others, being either members of FIXP, or complete for the class. 

At the heart of the definition of FIXP lies the notion of an algebraic circuit, used to represent a continuous function mapping a domain to itself. This representation effectively allows for the study of exact fixed points of the function, including irrational ones, and therefore can be used to capture the \emph{exact} complexity of these types of equilibrium problems. In contrast, in the usual Turing model of computation, sometimes the best one can hope for is approximate solutions (e.g., $\varepsilon$-Nash equilibria). The counterpart of FIXP in the Turing model is the class PPAD of \citet{JCSS:Papadimitriou1994} which famously captures the complexity of computing an $\varepsilon$-Nash equilibrium in strategic games \citep{SICOMP:DaskalakisGP2009,ChenDT09-Nash}. Indeed, for several of the aforementioned problems, computing approximate equilibria is in PPAD, whereas computing exact equilibria is in FIXP. Another interpretation of FIXP in the Turing model of computation is in terms of \emph{strong approximations}, i.e., computing points that are close in the sense of distance (e.g., in the max norm) to equilibrium points. In contrast, PPAD typically captures weak approximations, i.e., points that are approximately equilibrium points, but not necessarily close to an exact equilibrium point in the geometric sense.  

Contrary to the case of decision problems in NP, for which the membership in the class is often immediate, proving membership of a total search problem in the corresponding computational class is typically much more involved, and often requires ``transforming'' an existence proof into a computational reduction. This poses certain challenges, but it has been largely successful for problems in PPAD. For example, the PPAD-membership of Nash equilibrium computation incorporates Nash's existence proof (e.g., see \citep[Section 3.2]{goldberg2011survey}), and the PPAD-membership of the approximately envy-free cake cutting problem \citep{robertson1998cake,brams1996fair,OR:DengQS2012} is essentially a modification of an existence proof due to Simmons \citep{AMM:Su1999}. 

In the case of FIXP however, the aforementioned challenges are much more pronounced; for an existence proof to be used as a basis for a membership result, it has to display several characteristics. First, it has to go via Brouwer's fixed point theorem, and more importantly, it has to avoid using any ``discontinuous'' components, precluding the use of several types of discrete steps and limit arguments. For this reason, FIXP-membership results tend to be much more ad-hoc, using inventive but often rather involved techniques, which do not necessarily follow the known existence proofs. Even worse, for certain problems like the envy-free cake cutting problem for instance, the literature has not managed to produce any FIXP-membership result for the reasons mentioned above. 

A closer inspection into the several proofs of existence for versions of strategic games or competitive markets reveals that they often exhibit a common characteristic: they all include one or multiple optimization problems as subroutines. For example, at the heart of the Nash equilibrium notion is an agent's utility maximization problem, which can be expressed as a linear program (see \cref{eq:LP-Nash} in \cref{sec:OPT-gate-Nash}). Another example comes from competitive markets, where the market equilibrium notion includes convex optimization programs for maximizing the utilities of consumers and producers given a set of prices. This offers a possible explanation as to why the literature has fallen short of producing a systematic and unified approach for proving FIXP-membership results: Up until now, it was not known how to actually compute these optimization programs in FIXP, or more specifically, how to incorporate these programs as part of a FIXP circuit, as required for a membership result. 

Our paper remedies this situation: We show how to compute convex optimization programs, which can be used as black-box components of FIXP circuits. Simply put, under some mild assumptions, whenever such an optimization program is encountered in an existence proof, it can be effectively substituted by such a component in the FIXP-membership proof. Using our newly introduced technique, we manage to generalize and simplify several FIXP-membership proofs in the literature of game theory and competitive markets, as well as prove for the first time the seemingly elusive FIXP-completeness result for the envy-free cake cutting problem. We present our contributions in more detail below.

\subsection{Our Contribution}
Our main contribution is the introduction of the \emph{OPT-gate}, a new ``plug and play'' component which can be used as a black-box in FIXP-membership proofs for computing Linear Programs or more general convex optimization programs. The OPT-gate is a special kind of gate, with the following crucial property:
\begin{quote}
\emph{The OPT-gate is a ``pseudogate'', in the sense that its correct operation is only ensured at a fixed point of the function encoded by the algebraic circuit; with regards to a \FIXP-membership proof, it operates as a normal gate for all intents and purposes.}
\end{quote}
More specifically, the OPT-gate can solve any convex program with convex inequality constraints, explicit equality constraints and an explicit bound on its feasible region, as long as it satisfies a ``FIXP-appropriate'' variant of the well-known \emph{Slater condition} \citep{Slater50} for convex programs (see \cref{sec:OPT-gate-LP,sec:OPT-gate-convex}, and \cref{def:explicit-slater-LP,def:explicit-slater-convex}). Having programs of this form is in fact necessary (see the discussion in \cref{sec:OPT-gate}), but at the same time it is sufficient for capturing the rather general optimization problems that appear in the existence proofs mentioned above. 

To demonstrate the effectiveness of our technique, we present a host of different applications related to the \emph{envy-free cake cutting problem}, to computing different types of equilibria in various \emph{strategic games}, and to computing competitive equilibria in \emph{markets}. Our results advance the state of the art in three different ways: (a) we provide results for problems for which the complexity was previously entirely unknown, (b) we provide results that generalize known special cases in the literature to domains which are as general as possible, and (c) we provide proofs which are conceptually simpler and reminiscent of the known proofs of existence for those problems.

\subsubsection*{Applications to Game Theory} 

First, we discuss the application of our technique to the problem of computing exact equilibria in strategic games. Already in \cref{sec:OPT-gate-Nash}, we use the case of normal form games as a motivating example to demonstrate the strength of the OPT-gate. The FIXP-completeness of the problem was established by \citet{SICOMP:EtessamiY10}, in the same paper where they defined the class FIXP. We show that via the employment of our technique, the membership problem essentially boils down to simply writing the standard utility-maximization linear programs for the players and substituting them by the OPT-gate in the FIXP circuit, making the proof entirely straightforward. 

Then, in \cref{sec:games}, we move on to present more general classes of games and different equilibrium concepts, for which we also obtain FIXP-completeness or FIXP-membership results. In particular:
\begin{itemize}[leftmargin=*]
    \item[-] \textbf{FIXP-completeness of concave games.} In \cref{sec:concave-games}, we prove the FIXP-completeness of \emph{concave games} \citep{rosen1965concave}, a class of games which generalizes the class of normal-form games. \citet{rosen1965concave} showed via the employment of Kakutani's fixed point theorem \citep{kakutani1941generalization} that a Nash equilibrium of these games always exists. Our FIXP-membership proof defines a Brouwer function that uses the agent's utility-maximization program, now a convex program, as a subroutine, substituted by the OPT-gate. Similarly to the case of normal form games described above, our proof is very simple and intuitive.
    \item[-] \textbf{FIXP-membership of \eps-proper equilibria.} In \cref{sec:eps-proper}, we consider a Nash equilibrium refinement notion due to \citet{IJGT:Myerson78}, that of an \emph{$\eps$-proper equilibrium}.\footnote{We remark here that the \eps\ parameter is not the same type of approximation as in an \eps-Nash equilibrium mentioned earlier; see \cref{def:eps-proper} and the ``approximate'' vs ``almost'' discussion in \citep[Section 2]{GEB:Etessami2021}.} \citet{EC:HansenL18} showed that approximating a proper equilibrium (i.e., a limit point of \eps-proper equilibria) is complete for \FIXPa\  \citep{SICOMP:EtessamiY10}, the class of \emph{discrete} total search problems that reduce to (strong) approximate Brouwer fixed points. We show that computing an \eps-proper equilibrium is in FIXP. To obtain the result, we first develop a more general method based on solving \emph{systems of conditional convex constraints} (see \cref{sec:cond-conv-constraints}), making use of our OPT-gate, which might have applications beyond the \eps-proper equilibrium result. 
    \item[-] \textbf{FIXP-completeness of $\boldsymbol{n}$-player Stochastic Games}. In \cref{sec:stochastic-games}, we consider $n$-player stochastic games, which generalize the classic 2-player stochastic games of \citet{PNAS:Shapley53}. The existence of a \emph{stationary $\lambda$-discounted equilibrium} for any discount factor $\lambda$ was proven by \citet{JSHUA:Takahashi1964} and \citet{JSHUA:Fink1964} using a generalization of Kakutani's fixed point theorem. For 2-player zero-sum games, \citet{SICOMP:EtessamiY10} showed that computing a \emph{stationary $\lambda$-discounted equilibrium} is in FIXP. We generalize this membership result to $n$-player general stochastic games. Our proof is based on an enlarged domain of triples consisting of valuation profiles and pairs of stationary strategies, and constructs a Brouwer function from this domain to itself, for which the fixed points ``contain'' fixed points of the correspondence defined in Takahashi's proof on the original domain. The FIXP-hardness follows from \citep{SICOMP:EtessamiY10}, by noting that a normal form game may simply be viewed as a stochastic game with a single state.
\end{itemize}

\subsubsection*{Applications to Cake Cutting}

Next, in \cref{sec:cakes}, we prove our main result for the well-known envy-free cake cutting problem \citep{gamow1958puzzle} (see also \citep{brams1996fair,robertson1998cake,procaccia2013cake}). In this problem, the cake serves as a metaphor for a divisible resource, which needs to be divided fairly among a set of agents. The agents have different preferences over how to divide the resource, and an envy-free division is one which guarantees that each agent would rather have their own piece than any other agent's piece.  The existence of an envy-free division was proven by \citet{AMM:Stromquist1980}, even for the case where each agent receives a single piece (known as the \emph{contiguous} version or the version with \emph{connected pieces}). An alternative proof was provided by Simmons (cited in \citep{AMM:Su1999}). Both proofs employ a discretization of the space of possible divisions and then apply some topological lemma (either a variant of the K-K-M lemma \citep{FM:KnasterKM1929} or Sperner's lemma \citep{sperner1928neuer}), together with a limit argument. 

In terms of the complexity of the problem, results were only known for the approximate version of the problem: \citet{OR:DengQS2012} proved that for agents with very general valuations, computing a contiguous envy-free division of the cake is PPAD-complete. \citeauthor{OR:DengQS2012}'s proof closely follows Simmons' proof \citep{AMM:Su1999}, which, without the limit argument, obtains the existence of an approximately envy-free division. However, before our paper, the complexity of the \emph{exact} envy-free cake cutting problem was not known. To this end, we provide the following result.

\begin{quote}
\emph{The (contiguous) envy-free cake cutting problem with very general valuations is $\FIXP$-complete}. 
\end{quote}
By ``very general valuations'' we mean valuations that are not necessarily additive measures or even monotone over subsets of the cake, and which can assign different values to different divisions for an agent, even if the agent receives the same piece in all of those. The aforementioned existence proofs apply to this very general case as well, and therefore our FIXP-membership result is as strong as possible. We discuss this in more detail in \cref{sec:cakes} (see \cref{rem:general-valuations}).

In order to obtain the FIXP-membership result, we develop a new proof of existence for envy-free cake cutting, one which is not based on discretizations and limit arguments. Our proof constructs a bipartite graph between agents and preferred pieces and computes a maximum flow on this graph. This computation can be immediately substituted by our OPT-gate, effectively turning this new existence proof into a FIXP-membership result. This proof is somehow reminiscent of another existence proof by \citet{woodall1980dividing}, but as we explain in \cref{sec:cakes}, \citeauthor{woodall1980dividing}'s proof uses discontinuous steps and therefore cannot conceivably be ``turned'' into a FIXP-membership proof. 

For the FIXP-hardness, we construct a very simple reduction from a generalization of Brouwer's fixed point problem due to \citet{MP:Bapat1989}. The very same reduction also shows that Bapat's Brouwer fixed point problem is in FIXP. This in turn has implications for the \emph{rainbow} K-K-M problem \citep{IJGT:Gale1984}, a generalization of the K-K-M problem \citep{FM:KnasterKM1929}, which we show to be FIXP-complete via reductions from and to Brouwer's fixed point problem. These results, which are included in \cref{sec:kkm}, develop a potentially useful machinery for proving FIXP-completeness results for more general cake cutting and fair division problems. For example \citet{aharoni2020fractionally} establish the relation between K-K-M-type theorems and envy-free divisions of multiple cakes; whether these can yield FIXP-membership results for those problems as well is something to be explored in the future. 

\subsubsection*{Applications to Markets}

Our last application domain is that of competitive markets. Here we provide results for general \emph{Arrow-Debreu markets} \citep{arrow1954existence}, as well as for the \emph{pseudomarket mechanism} of \citet{hylland1979efficient}.

\begin{itemize}[leftmargin=*]
\item[-] \textbf{Arrow-Debreu markets.} In \cref{sec:markets-arrow-debreu} we prove a very general result, namely that computing competitive equilibria in Arrow-Debreu markets with concave utilities is in FIXP. The Arrow-Debreu market is the most fundamental market model, proposed and studied by \citet{arrow1954existence}. It consists of a set of consumers with utilities, consumption sets and endowments, and a set of producers or firms with production sets. A \emph{competitive} or \emph{market} equilibrium is a stable state in which supply equals demand, and all participants maximize their utilities or profits at the current set of prices. \citet{arrow1954existence} proved that under mild assumptions, every market has a competitive equilibrium.

FIXP-membership results were only previously known for special cases of Arrow-Debreu markets. \citet{SICOMP:EtessamiY10} in their original paper already proved the FIXP-membership of a setting where there are no explicit utilities, and the aggregate demand is a given function, rather than a correspondence which is typically the case in these markets. \citet{garg2016dichotomies} proved a FIXP-membership result for markets with \emph{Piecewise Linear Concave (PLC)} utilities, straightforward consumption sets (i.e., where consumption is only constrained to be non-negative), and production sets that are also given by PLC functions.

Our result for Arrow-Debreu markets generalizes\footnote{To be precise, our result applies to any class of concave utility functions, as long as we have access to the supergradients of those functions or when we can compute them given access to the functions. This is possible for the PLC utilities of \citep{garg2016dichotomies} as we explain in \cref{app:PDC}, but not for the CES utilities of \citep{ChenPY17-non-monotone-markets}, since these are non-superdifferentiable at $0$ coordinates. See \cref{rem:CES} in \cref{sec:markets-arrow-debreu} for more details.} the aforementioned results as it considers (a) more general utility functions (i.e., general concave functions) and (b) more general consumption and production sets (i.e., general convex sets).
Additionally, compared to the proofs in these papers, our membership proof is arguably simpler and follows rather easily from the original existence proof of \citeauthor{arrow1954existence}. Essentially, the only difference is that we ``organically'' devise a Brouwer function rather than a fixed point correspondence, and we substitute the various convex optimization programs that appear in the proof (for the consumers' and producers' optimality) by our OPT-gate.

\item[-] \textbf{The pseudomarket mechanism of \citet{hylland1979efficient}.} In \cref{sec:markets-hz} we consider the problem of computing equilibria of the pseudomarket mechanism of \citet{hylland1979efficient}. This mechanism solves the \emph{random assignment problem} (e.g., see \citep{bogomolnaia2001new}) by allocating to each agent a unit of artificial currency, and by then setting up a ``pseudomarket'' where agents buy probability shares of the different items. The \citeauthor{hylland1979efficient} pseudomarket is not a special case of the Arrow-Debreu market, because of additional allocation constraints that ensure that each agent receives exactly one item in expectation. \citeauthor{hylland1979efficient} employed Kakutani's fixed point theorem to prove that an equilibrium of this market is always guaranteed to exist.

The complexity of computing a \citeauthor{hylland1979efficient} equilibrium was an open problem since the definition of the mechanism in 1979 and certainly since the introduction of the relevant complexity classes for equilibrium computation problems. Very recently, \citet{vazirani_et_al:LIPIcs.ITCS.2021.59} showed that the problem lies in FIXP, leaving the FIXP-hardness as an open question. We employ our OPT-gate to obtain the same membership result, via, what we believe to be, an easier proof. Again, like most of our results, the proof resembles strongly the existence proof of \citet{hylland1979efficient}, except that it constructs a Brouwer fixed point function (rather than a Kakutani fixed point correspondence) and substitutes the agents' utility maximization Linear Programs by instances of the OPT-gate.
\end{itemize}

\subsection{Related Work}\label{sec:related}

Below we present some further related work related to our applications, as well as to fixed point computation problems.

\paragraph{Strategic games.} The field of game theory was developed in the late 1920s by the works of \citeauthor{neumann1928theorie} \citep{neumann1928theorie,von1944theory} and then notably in the 1950s with the concept of Nash equilibrium, guaranteed to exist by Nash's theorem \citep{Nash50}. The theorem can be proven by either using Brouwer's fixed point theorem \citep{nash1951non} or Kakutani's fixed point theorem \citep{Nash50}. The complexity of Nash equilibrium computation was firstly considered by \citet{JCSS:Papadimitriou1994}, who actually defined the class PPAD with this problem as the central consideration. More than a decade later, the celebrated results of \citet{SICOMP:DaskalakisGP2009} and \citet{ChenDT09-Nash} showed the PPAD-completeness of the approximate version of the problem, followed by the definition of FIXP and the FIXP-completeness result of \citet{SICOMP:EtessamiY10} for exact equilibria. Since then, several variants of the main normal form game setting and several refinements of the standard equilibrium notions have been considered, with corresponding complexity results being obtained (e.g., see \citep{deligkas2016inapproximability,deligkas2017computing,hansen2010computational,rubinstein2016settling}). 
Out of these refinements, the most relevant to us is the notion of proper equilibria defined by \citep{IJGT:Myerson78}. These equilibria were studied by \citet{EC:HansenL18} as we explained above, and it was shown that approximating them is complete for the class \FIXPa, a discrete variant of FIXP also defined by \citet{SICOMP:EtessamiY10}. 

Stochastic games were defined by \citet{PNAS:Shapley53} in the early 1950s, and they constitute one of the most fundamental models of repeated games in the literature, which can capture very general scenarios; we refer the reader to \citep{mertens1981stochastic,neyman2003stochastic} for more details on different types of stochastic games and their definitions. Our FIXP-membership result establishes the FIXP-completeness of the problem for $n$-player games; for 2-player games, besides the FIXP-membership shown in \citep{SICOMP:EtessamiY10}, the authors also show that the problem is at least as hard as the Square Root Sum problem, defined therein; whether the 2-player problem is FIXP-complete is still an open question. Exponential or superexponential time algorithms for the 2-player problem were developed by \citet{hansen2011exact} and \citet{oliu2021new}.

\paragraph{Cake cutting.} The cake cutting problem was introduced by \citet{Steinhaus1949} in the late 1940s and has since been studied extensively in the literature of mathematics, economics and computer science. The problem of finding an envy-free division was introduced by \citet{gamow1958puzzle} about a decade later. The existence of an envy-free division was shown in several proofs, but perhaps the most famous are those by \citet{AMM:Stromquist1980} and Simmons \citep{AMM:Su1999} that we mentioned earlier, which in fact guarantee the existence of contiguous divisions. The computational complexity of the approximate problem was considered by \citet{OR:DengQS2012} who proved a PPAD-completeness result for the case of general continuous preferences, which is equivalent to the setting of very general valuations we consider here. As we explained earlier, our paper provides the first computational complexity results for the problem of finding an exact envy-free division. For the usual case of additive valuation functions (e.g., see \citep{brams1996fair,robertson1998cake}), the FIXP-hardness for exact equilibria, or even the PPAD-hardness for approximate equilibria is still a major open problem.

In a related, but, in a sense, orthogonal line of work, several discrete protocols for finding an envy-free solution were proposed over the years, starting from the cut-and-choose protocol for 2 agents and the Selfridge-Conway protocol for 3 agents (e.g., see \citep{robertson1998cake}), leading to recent breakthrough results from the literature of computer science \citep{aziz2016discrete}. These protocols interact with the agents via a set of queries, in the so-called Robertson-Webb (RW) model (see \citep{woeginger2007complexity}). The RW model is not inherently a computational model, and RW queries can in fact return irrational points as answers. In that regime, the goal is to come up with a protocol that finds an envy-free solution using the smallest number of such queries possible. Even for the non-contiguous version, the discrepancy between the lower bound of \citet{procaccia2009thou} and the upper bound of \citet{aziz2016discrete} is astronomical.

\paragraph{Competitive Markets.} The fundamental principles of competitive markets and equilibrium theory date back to the 1870s and the works of \citet{walras1874elements}. \citeauthor{walras1874elements} described a process of adjusting the market prices based on supply and demand, the so-called ``t\^atonnement process'', which would eventually lead to the stable outcome that was later known as the competitive equilibrium. Foundational in the establishment of the associated equilibrium theory were the contributions of \citet{arrow1954existence} and \citet{mckenzie1954equilibrium},\footnote{In fact, sometimes the fundamental market model is referred to as the ``Arrow-Debreu-McKenzie'' market.} who proved the existence of an equilibrium. The proof of \citeauthor{arrow1954existence} uses a fixed point theorem due to \citet{debreu1952social}, whereas \citeauthor{mckenzie1954equilibrium} used Kakutani's fixed point theorem to obtain the result. An alternative proof via Brouwer's fixed point theorem was given by \citet{geanakoplos2003nash}. 

In computer science, much work has been devoted to the question of computing exact or approximate equilibria of different markets, which are special cases of the general Arrow-Debreu market that we study. There are several works that developed polynomial-time algorithms for finding or approximating equilibria for some classes of utility functions, e.g., see \citep{devanur2008market,jain2003approximating,jain2007polynomial,duan2015combinatorial,duan2016improved,garg2019strongly, garg2004auction,garg2015complementary}. For more complex utility functions, besides the results that we mentioned earlier, the approximate equilibrium computation problem for additively separable piecewise linear concave (SPLC) functions was shown to be PPAD-complete by \citep{vazirani2011market,ChenDDT09-Arrow-Debreu}, where the approximation notion is a ``weak approximation'' in the market clearing and utility-maximization conditions, see \citep{scarf1967approximation}. For exact equilibria, \citet{garg2014equilibrium} showed the PPAD-completeness of Arrow-Debreu markets with linear utility functions and SPLC production sets; in this case, it turns out that there always exist rational exact equilibria, and they can be computed in PPAD. An interesting class of utility functions is that of Leontief utilities, which are simultaneously subcases of the PLC utilities studied in \citep{garg2016dichotomies,garg2017settling} and limit cases of the CES utilities studied in \citep{ChenPY17-non-monotone-markets}. For this class, \citet{CodenottiSVY08-economies-games} showed a PPAD-hardness result. \citet{garg2017settling} showed hardness results for a market model with Leontief utilities, but their FIXP-hardness does not quite yield a FIXP-completeness result together with our membership proof, or even the membership proof of \citet{garg2016dichotomies}, because it is obtained for markets with different sufficiency conditions for equilibrium existence, and not for the market model as presented by \citet{arrow1954existence} that we study in \cref{sec:markets}.

\paragraph{Fixed point computation.} Besides the applications above, the class \FIXP\ also captures the complexity of other problems, such as branching process and context-free grammars \citep{SICOMP:EtessamiY10}, equilibrium refinements \citep{SAGT:EtessamiHMS14,GEB:Etessami2021}, and more recently the complexity of computing a Bayes-Nash equilibrium in the first-price auction with subjective priors \citep{fghlp2021_ec}. Besides \FIXP, there are some other computational classes that capture the complexity of different fixed point problems, namely the classes BU \citep{deligkas2021computing} and BBU \citep{batziou2021strong} which correspond to the Borsuk-Ulam theorem \citep{Borsuk1933}, and the class HB \citep{GoldbergH19-HairyBall}, which corresponds to the Hairy Ball theorem \citep{Poincare85-courbes}.

\section{Preliminaries}

In this section, we provide some definitions and theorems that we will use or reference throughout the paper, as well as the formal definition of the class FIXP. 

\subsection{Fixed Point Theorems}

We start with the definition of \emph{Brouwer's fixed point theorem} \citep{MA:Brouwer1911}, one of the most widely used fixed point theorems in economic applications, such as game theory or market theory.  

\begin{theorem}[Brouwer's Fixed Point Theorem \citep{MA:Brouwer1911}] \label{thm:brouwer}
  Let $A \subseteq \RR^n$ be a nonempty, compact, and convex set. Let
  $f \colon A \rightarrow A$ be continuous. Then there is $x \in A$
  such that $f(x)=x$.
\end{theorem}

\noindent The next fixed point theorem that we will present is \emph{Kakutani's fixed point theorem} \citep{kakutani1941generalization}, a generalization of Brouwer's fixed point theorem. Importantly, this fixed point theorem applies to \emph{correspondences} rather than functions; we provide the definition of a correspondence below.

\begin{definition}[Correspondence]
  A correspondence $f$ (or multi-valued function) between sets $A$ and
  $B$ is a function $f \colon A \rightarrow \mathcal{P}(B)$, where $\mathcal{P}(B)$ denotes the powerset of $B$. We denote
  this by $f \colon A \rightrightarrows B$. In case $f(a)=\{b\}$ we
  use the function notation $f(a)=b$ for notational simplicity. In a
  similar way, if for all $a \in A$ we have $\abs{f(a)}=1$, we may think of
    $f$ simply as a function $f \colon A \rightarrow B$.
\end{definition}

\noindent For the statement of the theorem, we need the definitions of \emph{upper} and \emph{lower hemicontinuous} correspondences.

\begin{definition}[Upper and lower hemicontinuous correspondence]
  Let $A \subseteq \RR^n$, $B \subseteq \RR^m$, and
  $f\colon A \rightrightarrows B$.
  \begin{enumerate}   
  \item $f$ is upper hemicontinuous (uhc) at $a \in A$ if and only if
    for all open sets $V \subseteq B$ for which $f(a) \subseteq V$
    there is an open set $U \subseteq A$ with $a \in U$ such that
    $f(x) \subseteq V$ for all $x \in U$.
  \item $f$ is lower hemicontinuous (lhc) at $a \in A$ if and only if
    for all open sets $V \subseteq B$ for which
    $f(a) \cap V \neq \emptyset$ there is an open set $U \subseteq A$
    with $a \in U$ such that $f(x) \cap V \neq \emptyset$ for all
    $x \in U$.
  \end{enumerate}
  We say that $f$ is uhc (lhc) if $f$ is uhc (lhc) at every $a \in
  A$. If $f$ is both uhc and lhc we simply say that $f$ is continuous.
\end{definition}

\noindent We are now ready to state the fixed point theorem.

\begin{theorem}[Kakutani's Fixed Point Theorem \citep{kakutani1941generalization}] \label{thm:kakutani}
  Let $A \subseteq \RR^n$ be a nonempty, compact, and convex set. Let
  $f \colon A \rightrightarrows A$ be uhc as well as nonempty-, compact- and convex-valued. Then there is $x \in A$
  such that $x \in f(x)$.
\end{theorem}

\noindent 
Kakutani's fixed point theorem is often used in conjunction with the following
theorem, called \emph{the maximum theorem}, proven in 1963 by \citet{berge1997topological}.
For this to be possible, it is additionally needed that the
maximizer-correspondence $g^*$ is convex-valued. This is in
particular ensured if $f$ is quasi-concave in its second variable and
$g$ is convex-valued.

\begin{theorem}[Berge's Maximum Theorem \citep{berge1997topological}]\label{thm:berge}
  Let $A \subseteq \RR^n$ and $B \subseteq \RR^m$. Let
  $f \colon A \times B \rightarrow \RR$ be continuous and
  $g \colon A \rightrightarrows B$ continuous as well as nonempty- and
  compact-valued. Define $f^* \colon A \rightarrow \RR$ and
  $g^* \colon A \rightrightarrows B$ by
  $f^*(a) = \max_{b \in g(a)} f(a,b)$ and
  $g^*(a) = argmax_{b \in g(a)} f(a,b)$. Then $f^*$ is continuous and
  $g^*$ is uhc as well as nonempty- and compact-valued.
\end{theorem}

\subsection{The Class FIXP}

As we said in the introduction, the class FIXP captures the complexity of real-valued search problems associated with Brouwer's fixed point theorem. We provide the formal definition of the class below.\medskip

\noindent A search problem $\Pi$ with real-valued search space is defined by associating to
any input instance $I$ (encoded as a string over a finite alphabet $\Sigma$)
a search space $D_I\subseteq\RR^{d_I}$ and a set of solutions $\Sol(I)$. 
We assume there is a polynomial time algorithm that given $I$ computes a description of $D_I$.
In order to define $\FIXP$, we first introduce a set of basic $\FIXP$-problems
corresponding to the formulation of Brouwer's fixed point theorem. Afterwards,
we explain how the class is closed with respect to a certain type of reductions. We start with the definition of an algebraic circuit.

\begin{definition}[Algebraic Circuit]\label{def:algebraic-circuit}
An algebraic circuit $C$ is a circuit using gates in $\{+,-,\ast,\div,\max,\min\}$ as well as rational constants. We let $\size(C)$ denote the size of the circuit, including the description of the rational constants.
\end{definition}

\noindent Next, we define the notion of a \emph{basic \FIXP\ problem}.

\begin{definition}[Basic \FIXP\ problem]\label{def:basic-fixp}
A search problem $\Pi$ is a basic $\FIXP$ problem
if every instance $I$ describes a nonempty compact convex domain
$D_I$ described by a set of linear inequalities with rational coefficients and a continuous map $F_I\colon D_I\rightarrow D_I$ given
by an algebraic circuit $C_I$,\footnote{Note that given an algebraic circuit, it is not clear how to check that it is indeed well-defined (i.e., does not divide by zero), and that it indeed represents a function $F_I$ with $F_I(D_I) \subseteq D_I$. For that reason, we assume that it is promised that the algebraic circuit $C_I$ indeed satisfies these two properties (in other words, we only consider instances where this is the case). Furthermore, note that as long as the circuit is well-defined, the function will be continuous, since all gates perform continuous operations.} and the solution set is $\Sol(I) = \{x\in D_I\mid
F_I(x)=x\}$. 
\end{definition} 

\noindent We now discuss reductions between search problems. Let $\Pi$ and $\Gamma$ be search problems with real-valued search space. A $\emph{many-one reduction}$ 
from $\Pi$ to $\Gamma$ is a pair of maps $(f,g)$. The instance mapping $f$ maps
instances $I$ of $\Pi$ to instances $f(I)$ of $\Gamma$, and for any solution 
$y\in\Sol(f(I))$ the solution mapping $g$ maps the pair $(I,y)$ to a solution $g(I,y)\in\Sol(I)$
of $\Pi$. In order to avoid meaningless reductions, it is required that $\Sol(f(I))\neq\emptyset$
if $\Sol(I)\neq\emptyset$. 
We require that the instance mapping $f$ is computable in polynomial time.
\citet{SICOMP:EtessamiY10} defined the notion of \emph{SL}-reductions where
the solution mapping $g$ is \emph{separable linear}. This means there 
exists a map $\pi\colon\{1,\dots, d_I\}\rightarrow \{1,\dots, d_{f(I)}\}$
and rational constants $a_i,b_i$, $i=1,\dots, d_I,$ such that for $y\in\Sol(f(I))$
one has that $x=g(I,y)$ is given by $x_i = a_iy_{\pi(i)}+b_i$ for all~$i$.\footnote{
This allows one, for example, to project away auxiliary variables.} The map $\pi$ and the constants $a_i,b_i$ should be computable from $I$ in polynomial time.\medskip

\noindent We are now ready to define the class \FIXP.

\begin{definition}[\FIXP]\label{def:FIXP}
The class $\FIXP$ consists of all search problems with real-valued search space that SL-reduce
to a basic $\FIXP$ problem for which the domain $D_I$ is a convex polytope
described by a set of linear inequalities with rational coefficients and the function
$F_I$ is defined by an algebraic circuit $C_I$. 
\end{definition}

\noindent In our definition of basic $\FIXP$ problems, we have assumed that the domains of the functions are polytopes. However, the definition of the class $\FIXP$ is robust to modifications of the definition of basic $\FIXP$ problems \citep{SICOMP:EtessamiY10}. For example, the basic $\FIXP$ problems could have been defined as having for instance unit-balls
$B_p^d$ as domains, $p\in [0,\infty]$, and one could have allowed the functions $F_I$
to be computed by circuits over $\{+,-,\ast,\div,\max,\min,\sqrt[k]{}\}$. 

Note that if we only allow the gates $\{+,-,\ast c,\max,\min\}$ (where ``$\ast c$'' denotes multiplication by a constant) in the definition of $\FIXP$, then we instead obtain the class $\linearFIXP$. It is known that $\linearFIXP = \PPAD$ \citep{SICOMP:EtessamiY10}. In other words, the difference between $\FIXP$ and $\PPAD$ comes from the added power of the general multiplication gate.

\section{The OPT-gate}\label{sec:OPT-gate}

In this section, we present a new technique for proving membership in \FIXP. Namely, we introduce the OPT-gate: a gate that can essentially solve Linear Programs, under some minor conditions. This gate can be used like any other gate for the purpose of proving \FIXP-membership. We begin with a motivating example in \cref{sec:OPT-gate-Nash} and then present some obstacles to constructing the OPT-gate, which lead us to define the notion of a pseudogate in \cref{sec:OPT-gate-pseudogate}. Using this notion we show how to construct an OPT-gate that solves LPs (\cref{sec:OPT-gate-LP}) and even more general convex programs (\cref{sec:OPT-gate-convex}).

\subsection{A motivating example: Nash equilibrium computation}\label{sec:OPT-gate-Nash}

To provide an example of how the OPT-gate could be used, we consider the classical problem of computing a Nash equilibrium in a normal form game.

\paragraph*{Normal form game.}
There are $n$ players and every player $i \in [n] := \{1, 2, \dots, n\}$ has a finite set of pure strategies $S_i=[m_i]$ and a payoff function $u_i \colon S \to \RR$, where $S := S_1 \times \dots \times S_n$. A \emph{mixed strategy} of player $i$ is a probability distribution on $S_i$. We let $\Sigma_i := \Delta(S_i)$ denote the set of all such distributions, i.e., $\Sigma_i := \{y \in \RRnn^{m_i}: \sum_j y_j = 1\}$. In other words, $\Sigma_i$ is the $(m_i-1)$-dimensional unit simplex. A \emph{mixed strategy profile} is a vector $x \in \Sigma := \Sigma_1 \times \dots \times \Sigma_n$, where $x_i \in \Sigma_i$ is the mixed strategy played by player $i$ in the strategy profile $x$. For $j \in S_i = [m_i]$, the $j$th coordinate of $x_i$ is denoted $x_{i,j}$ and it corresponds to the probability that player $i$ plays its pure strategy $j$. The payoff function $u_i$ can be extended to all mixed strategy profiles to obtain the \emph{expected} payoff function $\tilde{u}_i \colon \Sigma \to \RR$ where
\[\tilde{u}_i (x) = \Exp_{j_i \sim x_i} [u_i(j_1,\dots,j_n)] = \sum_{(j_1,\dots,j_n) \in S} x_{1,j_1} \cdots x_{n,j_n} u_i(j_1,\dots,j_n).\]

For a mixed strategy profile $x \in \Sigma$ and a mixed strategy $x_i' \in \Sigma_i$ for player $i$, we let $(x_i',x_{-i}) \in \Sigma$ denote the mixed strategy profile where $x_i$ has been replaced by $x_i'$.
A mixed strategy profile $x \in \Sigma$ is a \emph{Nash equilibrium} if for every player $i$ and every mixed strategy $x_i' \in \Sigma_i$ it holds that $\tilde{u}_i(x) \geq \tilde{u}_i(x_i',x_{-i})$. In other words, no player can improve its expected utility by unilaterally modifying its strategy.

\paragraph*{Computational problem.}
We consider the problem of computing a Nash equilibrium of a normal form game, where the payoff functions $u_i$ are given explicitly, i.e., for every $i \in [n]$ and every $(j_1,\dots,j_n) \in S$, the value $u_i(j_1,\dots,j_n)$ is provided as a rational number. By Nash's theorem such an equilibrium always exists \citep{Nash50}.

The simplest proof of Nash's theorem uses Kakutani's fixed point theorem (\cref{thm:kakutani}). However, prior to our work, it was not known how to use this proof to prove membership of the computational problem in \FIXP. Instead, the membership in \FIXP\ was shown by relying on an alternative proof of Nash's theorem that uses Brouwer's fixed point theorem \citep{SICOMP:EtessamiY10}.

Now assume for an instant that we allow an extra gate in the definition of \FIXP, namely the OPT-gate: a gate that can solve a Linear Program (LP). To be more precise, assume that the gate takes as input the description of an LP of the form
\begin{equation*}
\begin{aligned}
\mbox{maximize}\quad & c^\transpose x\\
\mbox{subject to}\quad & A x \leq b
\end{aligned}
\end{equation*}
namely, it takes as input $c,A,b$ and it outputs an optimal solution of the LP.

If such a gate was allowed in the construction of a \FIXP-circuit, then the \FIXP-membership of the Nash problem would essentially follow immediately. Indeed, note that at a Nash equilibrium $x$, every player maximizes its utility given the mixed strategies chosen by the other players. In other words, the mixed strategy $x_i$ played by player $i$ is an optimal solution of the following LP, where the variables are $y \in \RR^{m_i}$:
\begin{equation}\label{eq:LP-Nash}
\begin{aligned}
\mbox{maximize}\quad & \tilde{u}_i(y,x_{-i})\\
\mbox{subject to}\quad & \sum_{j=1}^{m_i} y_j = 1\\
& y_j \geq 0 \qquad j=1,\dots,m_i
\end{aligned}
\end{equation}
Note that this is indeed an LP, since $\tilde{u}_i(y,x_{-i})$ is linear in $y$.

In more detail, the \FIXP-circuit $F$ for this problem would be constructed as follows. On input $x=(x_1,\dots,x_n)$, it outputs $F(x)=(F_1(x),\dots,F_n(x))$, where $F_i(x) \in \RR^{m_i}$ is an optimal solution of the corresponding LP \eqref{eq:LP-Nash}. Since the LP \eqref{eq:LP-Nash} can be put in the form needed for the hypothetical OPT-gate above by replacing the equality constraint by two inequality constraints, we can use the hypothetical gate for solving an LP for this. In more detail, the inputs to the gate will be $A$ and $b$ that encode the inequality constraints, and the vector $c$ for the objective function, which will depend on the values of the other inputs $x_{-i}$. Clearly, any fixed point of $F$ is a Nash equilibrium of the game.

This simple example already shows how such a gate could make some \FIXP-membership results very easy to prove. Importantly, the technique also feels very ``natural'', because it can be applied almost immediately given the description of the problem, without the need to reformulate the problem in any way. Indeed, in this example, the \FIXP-membership is essentially immediately obtained from the simple proof of Nash's theorem based on Kakutani's fixed point theorem.

Unfortunately, this ``ideal'' gate described above is in fact too good to be true. Indeed, there are some fundamental obstacles to constructing such a gate using the standard gates allowed in \FIXP-circuits.

\subsection{Pseudogates: Circumventing obstacles to the construction of an OPT-gate}\label{sec:OPT-gate-pseudogate}

We consider the task of constructing a gate that solves LPs. To be more precise, we would like to use the standard algebraic gates allowed in a \FIXP-circuit to construct, for any $n \in \NN$ and $m \in \NN_0$ (where $\NN_0 = \NN \cup \{0\}$), a new gate $G_{n,m}$ that takes as input $c \in \RR^n$, $A \in \RR^{m \times n}$ and $b \in \RR^m$, and outputs an optimal solution to the following LP:
\begin{equation}\label{eq:LP-initial-form}
\begin{aligned}
\mbox{maximize}\quad & c^\transpose x\\
\mbox{subject to}\quad & A x \leq b
\end{aligned}
\end{equation}

\begin{obstacle}
Any function mapping the description of an LP to an optimal solution of the LP cannot be continuous everywhere. This holds even for very simple LPs.
\end{obstacle}

As an example for this obstacle, consider the following very simple LP:
\begin{equation}\label{eq:LP-discontinuous}
\begin{aligned}
\mbox{maximize}\quad & x_1 v_1 + x_2 v_2\\
\mbox{subject to}\quad & x_1 + x_2 = 1\\
& x_1,x_2 \geq 0
\end{aligned}
\end{equation}
where the variables are $x_1,x_2$, and $v_1,v_2 \in \RR$ are external parameters. Clearly, our gate should be able to solve the following task: given $v_1,v_2$ as input, output any optimal solution of the LP. However, note that this function is \emph{not} continuous: when $v_1>v_2$ it outputs $(x_1,x_2)=(1,0)$, but when $v_1<v_2$ it outputs $(x_1,x_2)=(0,1)$. Thus, there is no hope of implementing a gate computing this function by using the gates allowed in a \FIXP-circuit, which are all continuous.

\paragraph*{Pseudogates: The power of fixed point computation.}
The crucial observation that allows us to go beyond this impossibility result is the following: when the gate is used inside a \FIXP-circuit $F$, it does not have to work correctly for all inputs $x$ to $F$; it suffices if it works correctly whenever the input to $F$ is a fixed-point $x^*$ of $F$. Indeed, in order to prove the membership of some problem in \FIXP\ using $F$, we have to show that any fixed point $x^*$ of $F$ yields a solution to the problem. Thus, we only care about the behavior of the gate when the input to $F$ is some fixed point $x^*$. Of course, the gate should remain well-defined for all inputs $x$, namely not divide by zero, etc.

This observation essentially allows us to use an additional---very powerful---tool in the construction of the gate: fixed point computation. In order to illustrate this point, we show how this tool can be used to construct a ``gate'' that computes the so-called \emph{Heaviside step function}. For our purposes, we define the Heaviside function as the correspondence $\Heaviside \colon \RR \rightrightarrows [0,1]$ with
\[
\Heaviside(x) = \begin{cases}
1 & \text{ if } x>0 \\
[0,1] & \text{ if } x=0 \\
0 & \text{ if } x<0 
\end{cases} \enspace .
\]
We would like a gate that on input $x$, outputs any $y \in \Heaviside(x)$. Clearly, a gate that is constructed using only the standard \FIXP-gates cannot compute $\Heaviside$, which is discontinuous at $x=0$. Indeed, note that the Heaviside function is closely related to the example LP \eqref{eq:LP-discontinuous} above. If we had a gate computing the Heaviside function, then by computing $y \in H(v_1-v_2)$ and then outputting $(y,1-y)$, we would simulate a gate solving \eqref{eq:LP-discontinuous}. Similarly, if we had a gate solving the LP \eqref{eq:LP-discontinuous}, then by computing a solution $(x_1,x_2)$ to the LP \eqref{eq:LP-discontinuous} with parameters $(v_1,v_2) = (x,0)$, and outputting $x_1$, we would simulate a gate for $\Heaviside$.

Let us now see how we can construct a ``gate'' computing the Heaviside function $\Heaviside$. Consider the function $G \colon \RR \times [0,1] \to [0,1]$ given by $G(x,y) = \min(1,\max(0,x+y))$. Let us examine the fixed points of $G$, where we think of $x$ as being fixed or an external parameter. If $x > 0$, then the fixed point condition $G(x,y) = y$ implies that $y$ must be equal to $1$. If $x < 0$, then the fixed point condition implies that $y=0$. Finally, when $x = 0$, the fixed point condition implies that $y$ can take any value in $[0,1]$. In particular, note that we always have $y \in \Heaviside(x)$.

How can we use this to prove membership in \FIXP? Imagine that we can reduce our problem of interest to the problem of finding a fixed point of a correspondence $F \colon D \rightrightarrows D$, i.e., a point $x \in D$ with $x \in F(x)$. Imagine, further, that we can construct a circuit computing $F$ that uses the standard gates, but also a gate computing $\Heaviside$. Then, we can construct a \FIXP-circuit $\tilde{F}$ for this problem by replacing the gate for $\Heaviside$ in $F$ by the function $G$ defined above. In more detail, if we want to use a gate computing $\Heaviside$ with some input $x$, we instead compute $G(x,y)$, where $y$ is an additional input to $\tilde{F}$. We also add this value $G(x,y)$ as an additional output to $\tilde{F}$ (namely, the output corresponding to the new input $y$). As a result, we obtain a FIXP-circuit $\tilde{F} \colon D \times [0,1] \to D \times [0,1]$ that only uses the standard gates and is such that any fixed point $(z,y)$ of $\tilde{F}$ satisfies $z \in F(z)$. In other words, finding a fixed point of the correspondence $F$ reduces to finding a fixed point of the function $\tilde{F}$, which is a standard \FIXP-circuit. In the case where $F$ makes use of multiple gates computing $\Heaviside$, every occurrence of the gate will be replaced by the construction above using $G$. In particular, if the gate for $\Heaviside$ is used $\ell$ times, then we will obtain $\tilde{F} \colon D \times [0,1]^\ell \to D \times [0,1]^\ell$ such that $\tilde{F}(z,y_1,\dots,y_\ell) = (z,y_1,\dots,y_\ell) \implies z \in F(z)$.

As a result, when constructing a \FIXP-circuit for some problem, we can assume that we also have access to a gate computing $\Heaviside$. However, one should keep in mind that the gate is only guaranteed to work correctly at a fixed point of the circuit. In order to stress this limitation, we say that we have a \emph{pseudogate} computing $\Heaviside$. Note that for the purpose of proving membership in \FIXP, a pseudogate is just as good as a normal gate. We now present these ideas more formally.

\begin{definition}\label{def:fixed-point-repr}
Let $A \subseteq \RR^n$, and let $B \subset \RR^\ell$ be a nonempty, compact, and convex set. For any continuous function $G \colon A \times B \to \RR^m \times B$ we let $\fix_B[G]$ denote the correspondence induced by $G$ with fixed-point constraints on $B$. Formally, the correspondence $\fix_B[G] \colon A \rightrightarrows \RR^m$ is defined as
\[x \mapsto \{z \in \RR^m : \exists y \in B \quad G(x,y)=(z,y)\}.\]
When $f = \fix_B[G]$ we will say that $G$ is a fixed-point representation of $f$. We will often have $B=[0,1]^\ell$ for some $\ell \in \NN$, in which case we will use $\fix_\ell$ as an abbreviation for $\fix_{[0,1]^\ell}$.
\end{definition}

\begin{example}
The function $G_{\Heaviside} \colon \RR \times [0,1] \to \RR \times [0,1]$, $(x,y) \mapsto (y,\min(1,\max(0,x+y)))$, is a fixed-point representation of the Heaviside function $\Heaviside$, i.e., $\fix_1[G_{\Heaviside}] = \Heaviside$. Now let us consider a function similar to the Heaviside function, but which will require us to have the first output of $G$ be something other than just $y$ itself. Let $f \colon \RR \rightrightarrows \RR$ be defined by
\[
f(x) = \begin{cases}
x+1 & \text{ if } x>0 \\
[0,1] & \text{ if } x=0 \\
x & \text{ if } x<0 
\end{cases} \enspace .
\]
Then a fixed-point representation of $f$ is given by
\[G_f \colon \RR \times [0,1] \to \RR \times [0,1], \quad (x,y) \mapsto (x+y,\min(1,\max(0,x+y)))\]
i.e., $\fix_1[G_f] = f$.
\end{example}

\begin{definition}\label{def:pseudogate}
Let $A \subseteq \RR^n$ and let $f \colon A \rightrightarrows \RR^m$ be a correspondence. We say that there is a pseudogate computing $f$ if there exists $\ell \in \NN_0$ and an algebraic circuit computing $G \colon A \times [0,1]^\ell \to \RR^m \times [0,1]^\ell$ such that for all $x \in A$, $\fix_\ell[G](x) \subseteq f(x)$.
\end{definition}
The algebraic circuit computing $G$ can use any of the standard gates allowed in \FIXP-circuits, and should be well-defined, in the sense that it never divides by zero, never takes the square root of a negative number, etc.
Note that by Brouwer's fixed point theorem, $\fix_\ell[G](x)$ is never empty. Thus, if there is a pseudogate computing some correspondence $f$, then $f$ must be nonempty-valued. For a discussion about why \cref{def:pseudogate} uses ``$\fix_\ell[G](x) \subseteq f(x)$'' instead of ``$\fix_\ell[G](x) = f(x)$'' see \cref{rem:pseudogate-def} at the end of the section.

Using this terminology we can now formally state:

\begin{lemma}\label{lem:Heaviside-gate}
There exists a pseudogate computing the Heaviside function $\Heaviside \colon \RR \rightrightarrows [0,1]$.
\end{lemma}

\begin{proof}
The function $G_{\Heaviside} \colon \RR \times [0,1] \to \RR \times [0,1]$, $(x,y) \mapsto (y,\min(1,\max(0,x+y)))$, can be represented by an algebraic circuit using the gates $+, \max, \min$ and rational constants $0$ and $1$. Furthermore, it is easy to see that $\fix_1[G_{\Heaviside}]=H$.
\end{proof}

The pseudogate for the Heaviside function will be a crucial building block for the construction of the OPT-gate. In fact, as mentioned above, the pseudogate for $\Heaviside$ essentially immediately yields a pseudogate solving the simple LP \eqref{eq:LP-discontinuous}.

Note that a (potential) pseudogate for our general LP \eqref{eq:LP-initial-form} will necessarily depend on $n$ and $m$, namely the number of variables and constraints, respectively. As a result, we say that a pseudogate solves such an LP, if, given $n,m$ we can, in polynomial time in $n$ and $m$, construct a pseudogate solving the LP for fixed $n$ and $m$.

\bigskip

Can we construct a pseudogate for the general LP \eqref{eq:LP-initial-form}? Unfortunately, there are a few more obstacles.

\begin{obstacle}
A pseudogate cannot solve a general LP without some explicit bound on the feasible region.
\end{obstacle}

As an example, consider the following LP, which corresponds to letting $n=1$, $m=2$, $c=1$, $A = (a, -1)^\transpose$ and $b = (1, 0)$ in \eqref{eq:LP-initial-form}:
\begin{equation*}
\begin{aligned}
\mbox{maximize}\quad & x\\
\mbox{subject to}\quad & a x \leq 1\\
& x \geq 0
\end{aligned}
\end{equation*}
where the variable is $x \in \RR$. The solution to this LP is $x = 1/a$ when $a > 0$, and the LP is unbounded when $a \leq 0$.

Assume that we have a pseudogate solving the LP \eqref{eq:LP-initial-form} and we use it to solve the LP above.
It is reasonable to only demand that the pseudogate solve the LP correctly when $a > 0$. However, the pseudogate---or, to be more precise, the continuous function $G$ implementing it---should be well-defined for \emph{all} $a \in \RR$. In particular, it should never divide by zero or take a square root of a negative number. This is to ensure that the pseudogate can really be used like a normal gate without a second thought.

Unfortunately, this means that no pseudogate can be constructed for LP \eqref{eq:LP-initial-form}. Indeed, by \cref{def:pseudogate}, the existence of such a pseudogate would imply the existence of an algebraic circuit $G \colon \RR \times [0,1]^\ell \to \RR \times [0,1]^\ell$ such that $\fix_\ell[G](a) = \{1/a\}$ for all $a > 0$. In particular, $\fix_\ell[G](a)$ would be unbounded when $a$ tends to $0$ from above. However, this is a contradiction to the continuity of $G$, which says that $G([0,1] \times [0,1]^\ell)$ must be compact and thus, in particular, bounded.

\paragraph*{Explicitly bounded domain.}
This issue can be resolved by introducing an explicit bound on the feasible region, namely by replacing \eqref{eq:LP-initial-form} by:
\begin{equation}\label{eq:LP-bounded}
\begin{aligned}
\mbox{maximize}\quad & c^\transpose x\\
\mbox{subject to}\quad & A x \leq b\\
& x \in [-R,R]^n
\end{aligned}
\end{equation}
where $R \in \RRp$. Note that the notation ``\,$x \in [-R,R]^n$\,'' is used for convenience here; this constraint can equivalently be rewritten as ``\,$-R \leq x_i \leq R, \forall i$\,''.

Importantly, the parameter $R$ is \emph{not} fixed, but is just another input to the gate, like $c$, $A$ and $b$. As a result, this explicit bound is not a significant limitation, since in most applications it is straightforward---or even trivial---to provide such a bound. For example, in the problem of computing a Nash equilibrium, the LP \eqref{eq:LP-Nash} that we used is clearly bounded with $R=1$.

If we are only interested in finding a feasible point of the LP \eqref{eq:LP-bounded}, or equivalently in solving the LP when $c=0$, then indeed there exists a pseudogate for that! However, there is still one last obstacle to constructing a pseudogate that \emph{solves} \eqref{eq:LP-bounded}.

\begin{obstacle}
A pseudogate cannot solve a general LP without some constraint qualification.
\end{obstacle}

A constraint qualification is some property that the constraints must satisfy. Importantly, it is a property of the constraints and not of the feasible region. In other words, when a feasible region can be represented by various different sets of constraints, some of them may satisfy the constraint qualification, and others not.

As an example for this obstacle, consider the following LP:
\begin{equation*}
\begin{aligned}
\mbox{maximize}\quad & x_2\\
\mbox{subject to}\quad & x_1 + a x_2 \leq 0\\
& x_1 \geq 0\\
& x \in [-1,1]^n
\end{aligned}
\end{equation*}
Note that for $a=0$ the optimal solution is $(x_1,x_2)=(0,1)$, while for $a>0$ it is $(x_1,x_2)=(0,0)$.
Clearly, this LP can be expressed in the form \eqref{eq:LP-bounded} by letting $n=2$, $m=2$, and
\[c=\begin{pmatrix}
0 \\
1
\end{pmatrix}, \quad
A=\begin{pmatrix}
1 & a \\
-1 & 0
\end{pmatrix}, \quad
b=\begin{pmatrix}
0 \\
0
\end{pmatrix}, \quad
R=1.\]
Thus, a pseudogate for \eqref{eq:LP-bounded} should in particular correctly solve this LP for any $a \in [0,1]$. According to \cref{def:pseudogate}, this would mean that there exists an algebraic circuit $G \colon [0,1] \times [0,1]^\ell \to \RR^2 \times [0,1]^\ell$ such that $\fix_\ell[G](a) = \{(0,0)\}$ for all $a \in (0,1]$, and $\fix_\ell[G](0) = \{(0,1)\}$. However, this contradicts the continuity of $G$.

Indeed, consider the sequence $(a_n)_n$ where $a_n=1/n > 0$. For any $n \in \NN$, let $y_n \in [0,1]^\ell$ be a fixed point of the function $h \colon [0,1]^\ell \to [0,1]^\ell$, $y \mapsto G_2(a_n,y)$, where $G_2(a_n,y) \in [0,1]^\ell$ denotes the second output of $G$ on input $(a_n,y)$. Recall that such a fixed point must exist by Brouwer's fixed point theorem. Since $(y_n)_n$ is a sequence in the compact set $[0,1]^\ell$, it has a subsequence $(y_{n_k})_k$ that converges to some $y \in [0,1]^\ell$. Note that for all $k \in \NN$ we have $G(a_{n_k},y_{n_k})=((0,0),y_{n_k})$, because $\fix_\ell[G](a_{n_k}) = \{(0,0)\}$. Now, since $a_{n_k} \to 0$, $y_{n_k} \to y$, and by the continuity of $G$, it follows that $G(0,y) = ((0,0),y)$. However, this implies that $(0,0) \in \fix_\ell[G](0)$, a contradiction to $\fix_\ell[G](0) = \{(0,1)\}$.

The issue in this example essentially stems from the fact that, when $a=0$, the two inequality constraints are equivalent to a single equality constraint. In fact, it is possible to construct a pseudogate for \eqref{eq:LP-bounded} that works as long as this does not happen, i.e., as long as the constraint qualification $\{x \in (-R,R)^n: Ax < b\} \neq \emptyset$ holds (where $<$ is componentwise). However, this rules out equality constraints, which we would clearly like our pseudogate to be able to handle, in particular for the Nash problem. To address this issue, in the next section we consider a modified formulation of our LP that allows explicit equality constraints and we show that we can construct a pseudogate that solves it as long as a well-known constraint qualification holds.

In particular, our constraint qualification will require the equality constraints to be linearly independent.
As an example for why the linear independence of the equality constraints is needed, consider the following LP:
\begin{equation*}
\begin{aligned}
\mbox{maximize}\quad & x_2\\
\mbox{subject to}\quad & x_1 + a x_2 = 0\\
& x_1 = 0\\
& x \in [-1,1]^n
\end{aligned}
\end{equation*}
By the same arguments as above, it can be shown that a pseudogate cannot solve this LP correctly for all $a \in [0,1]$.

\begin{remark}\label{rem:pseudogate-def}
The attentive reader might look at the definition of a pseudogate (\cref{def:pseudogate}) and wonder why the condition ``$\fix_\ell[G](x) \subseteq f(x)$'' is not simply replaced by ``$\fix_\ell[G](x) = f(x)$''. Indeed, the pseudogate presented above for the Heaviside function does satisfy the condition with equality. In fact, using Brouwer's fixed point theorem, it is not too hard to show that \emph{any} pseudogate computing $\Heaviside$ will satisfy the condition with equality. However, consider now the following modification of the Heaviside function:
\[
\widehat{\Heaviside}(x) = \begin{cases}
1 & \text{ if } x>1 \\
[0,1] & \text{ if } x \in [0,1] \\
0 & \text{ if } x<0 
\end{cases} \enspace .
\]
Note that the pseudogate $G_{\Heaviside}$ we provided above for $\Heaviside$ is also a pseudogate for $\widehat{\Heaviside}$, but we now have $\fix_1[G_{\Heaviside}](1) \subsetneq \widehat{\Heaviside}(1)$. This begs the question of whether we lose anything by allowing the pseudogate to only compute a subset of the output of the initial correspondence. For the purpose of proving membership in \FIXP\, the answer is no. Ultimately, we only want to show that any fixed point of the circuit that we construct satisfies some conditions (e.g., is a Nash equilibrium). Using a pseudogate that enforces a stronger condition than actually intended will not make this any harder. The important thing to note is that the constructed circuit will always have a fixed point, and thus, even if we use pseudogates that enforce stronger conditions than intended, there is no risk of the conditions being ``too strong.'' Thus, in this context, there is no reason to require equality in \cref{def:pseudogate}, since we only really care about the containment in one direction. Is there any setting where we would care about having equality? The only setting that comes to mind is if one is not only interested in proving \FIXP-membership, but wants to construct a circuit such that there is a one-to-one correspondence between its fixed points (perhaps after projecting away some coordinates) and the solutions of the problem that is studied.
\end{remark}

\subsection{The OPT-gate for Linear Programming}\label{sec:OPT-gate-LP}

We consider the following LP formulation, which includes explicit equality constraints and an explicit bound on the feasible region:
\begin{equation}\label{eq:standard-LP}
\begin{aligned}
\mbox{maximize}\quad & c^\transpose x\\
\mbox{subject to}\quad & A x = b\\
& Cx \leq d\\
& x \in [-R,R]^n
\end{aligned}
\end{equation}
where $x \in \RR^n$ is the vector of unknown variables, $c \in \RR^n$ defines the objective function, and the constraints are given by $A \in \RR^{m \times n}$, $b \in \RR^m$, $C \in \RR^{k \times n}$, $d \in \RR^k$, and $R \in \RRp$.

We introduce the following constraint qualification for our LP formulation.
\begin{definition}\label{def:explicit-slater-LP}
We say that the \emph{explicit Slater condition} is satisfied by the LP \eqref{eq:standard-LP} if the following two conditions hold:
\begin{enumerate}
	\item \emph{non-empty interior:} there exists $x \in (-R,R)^n$ with $Ax=b$ and $Cx < d$ (componentwise),
	\item \emph{linear independence:} the rows of $A$ are linearly independent.
\end{enumerate}
\end{definition}
The Slater condition \citep{Slater50} is very popular in convex optimization, where it is usually defined using only the first condition. This is without loss, because the second condition can always be enforced with some additional preprocessing (namely, eliminating redundant equality constraints). For our purpose, however, the second condition is required because we cannot perform the usual preprocessing inside an algebraic circuit. To avoid any confusion, we thus refer to the two conditions above as the \emph{explicit} Slater condition.

\bigskip

\noindent The main result of this section can now informally be stated as:
\begin{mdframed}[linewidth=1pt,frametitle={The OPT-gate for Linear Programs},frametitlealignment=\centering]
There exists a pseudogate for the LP formulation \eqref{eq:standard-LP}. This pseudogate has the following guarantees:
\begin{itemize}
	\item when the feasible region of the LP is non-empty, it outputs a feasible point.
	\item when the LP satisfies the explicit Slater condition, it outputs an optimal solution.
\end{itemize}
This pseudogate can be used like any other algebraic gate for the purpose of proving membership in \FIXP.
\end{mdframed}
The informal statement above is formally stated in \cref{thm:OPT-gate-LP} below.
Note that with the OPT-gate we can in particular directly prove the \FIXP-membership of the Nash problem (\cref{sec:OPT-gate-Nash}), since the explicit Slater condition is trivially satisfied by all the LPs in question.

\bigskip

To formalize the statement, we think of the LP \eqref{eq:standard-LP} as being parameterized by the tuple $(c,A,b,C,d,R)$. Thus, after fixing $n \in \NN$ and $m,k \in \NN_0$, we can define the parameter space
\[P_{n,m,k} = \RR^n \times \RR^{m \times n} \times \RR^m \times \RR^{k \times n} \times \RR^k \times \RRp.\]
For any choice of parameters $p=(c,A,b,C,d,R) \in P_{n,m,k}$ we let $\LP(p)$ denote the corresponding LP formulated in \eqref{eq:standard-LP}. We will use $\Feas(\LP(p))$ and $\Opt(\LP(p))$ to denote its set of feasible and optimal solutions, respectively. The main result of this section can be stated formally as follows.

\begin{theorem}\label{thm:OPT-gate-LP}
Given $n \in \NN$ and $m,k \in \NN_0$ we can construct an algebraic circuit $G \colon P_{n,m,k} \times [0,1]^\ell \rightarrow \RR^n \times [0,1]^\ell$ in time $\poly(n,m,k)$ such that for any parameters $p = (c,A,b,C,d,R) \in P_{n,m,k}$ it holds:
\begin{itemize}
	\item if the feasible region of $\LP(p)$ is non-empty, i.e., $\Feas(\LP(p)) \neq \emptyset$, then
		\[\fix_\ell[G](p) \subseteq \Feas(\LP(p)).\]
	\item if $\LP(p)$ satisfies the explicit Slater condition (\cref{def:explicit-slater-LP}), then
		\[\fix_\ell[G](p) \subseteq \Opt(\LP(p)).\]
\end{itemize}
\end{theorem}

\noindent \cref{thm:OPT-gate-LP} follows from the more general \cref{thm:OPT-gate-convex}, which is stated and proved in the next section.

\subsection{The OPT-gate for Convex Optimization}\label{sec:OPT-gate-convex}

We can apply the approach presented above to the more general setting of convex optimization. Consider a Convex Program (CP) of the following form:
\begin{equation*}
\begin{aligned}
\mbox{minimize}\quad & f(x)\\
\mbox{subject to}\quad & A x = b\\
& g_i(x) \leq 0 \qquad i=1,\dots,k\\
& x \in [-R,R]^n
\end{aligned}
\end{equation*}
where $f \colon \RR^n \rightarrow \RR$ and $g_i \colon \RR^n \rightarrow \RR$, $i=1,\dots,k$, are convex functions, and, as before, the remaining constraints are given by $A \in \RR^{m \times n}$, $b \in \RR^m$, and $R \in \RRp$.

For this setting we can again define the appropriate explicit Slater condition.
\begin{definition}\label{def:explicit-slater-convex}
We say that the \emph{explicit Slater condition} is satisfied by the Convex Program \eqref{eq:standard-CP} if the following two conditions hold:
\begin{enumerate}
	\item \emph{non-empty interior:} there exists $x \in (-R,R)^n$ with $Ax=b$ and $g_i(x) < 0$ for $i=1,\dots,k$,
	\item \emph{linear independence:} the rows of $A$ are linearly independent.
\end{enumerate}
\end{definition}

\bigskip

\noindent The main result of this section can now informally be stated as follows:
\begin{mdframed}[linewidth=1pt,frametitle={The OPT-gate for Convex Optimization},frametitlealignment=\centering]
There exists a pseudogate for the Convex Program (CP) \eqref{eq:standard-LP}. This pseudogate has the following guarantees:
\begin{itemize}
	\item when the feasible region of the CP is non-empty, it outputs a feasible point.
	\item when the CP satisfies the explicit Slater condition, it outputs an optimal solution.
\end{itemize}
This pseudogate can be used like any other algebraic gate for the purpose of proving membership in \FIXP.
\end{mdframed}
For a formal statement, see \cref{thm:OPT-gate-convex} below.

\paragraph*{Parameters.}
On the way to making this statement formal, we need to allow various parts of the optimization problem to depend on a set of parameters (which are going to be the inputs to our pseudogate). Clearly, $A$, $b$ and $R$ are such parameters---as before---but we would also like the objective function $f$ and the inequality constraints $g_i$ to be parameterized (even just to be able to encode the LP \eqref{eq:standard-LP} from the previous section). To address this, we introduce an additional parameter $w \in \RR^s$ and reformulate the optimization problem as follows:
\begin{equation}\label{eq:standard-CP}
\begin{aligned}
\mbox{minimize}\quad & f(x\,;w)\\
\mbox{subject to}\quad & A x = b\\
& g_i(x\,;w) \leq 0 \qquad i=1,\dots,k\\
& x \in [-R,R]^n
\end{aligned}
\end{equation}
where $f \colon \RR^n \times \RR^s \rightarrow \RR$ and $g_i \colon \RR^n \times \RR^s \rightarrow \RR$, $i=1,\dots,k$, are continuous functions such that $f(\cdot\,;w)$ and $g_i(\cdot\,;w)$ are convex functions for any $w \in \mathbb{R}^s$. Note that $x \in \RR^n$ is still the vector of unknown variables and $w$ is simply an additional \emph{external} parameter, just like $A$, $b$ and $R$, and is thus treated as completely fixed when optimizing.

After fixing $n \in \NN$ and $m, k, s \in \NN_0$, as well as the functions $f$ and $g_i$, $i = 1,\dots,k$, we can define the parameter space $P_{n,m,k,s,f,g} = \RR^s \times \RR^{m \times n} \times \RR^m \times \RRp$. To simplify notation we write $P_{n,m,f,g}$ to mean $P_{n,m,k,s,f,g}$, since $k$ and $s$ are, in a certain sense, also implicitly given by $f$ and $g=(g_1,\dots, g_k)$. For any choice of parameters $p=(w,A,b,R) \in P_{n,m,f,g}$ we let $\CP(p)$ denote the corresponding CP formulated in \eqref{eq:standard-CP}. As before, we will use $\Feas(\CP(p))$ and $\Opt(\CP(p))$ to denote its set of feasible and optimal solutions, respectively. As above, the parameters $p=(w,A,b,R) \in P_{n,m,f,g}$ will be the inputs of the pseudogate we construct.

\paragraph*{Representation of functions and subgradients.}
For computational purposes, we assume that the functions $f$ and $g_i$ are given as algebraic circuits. However, we will also need access to \emph{subgradients} of these functions.

\begin{definition}\label{def:subgradient}
Let $A \subseteq \RR^n$ be a convex set and let $f \colon A \to \RR$ be a convex function. A vector $v \in \RR^n$ is a \emph{subgradient} of $f$ at the point $x \in A$ if, for all $y \in A$,
\[f(y) - f(x) \geq v \cdot (y-x).\]
We let $\partial f(x)$ denote the \emph{subdifferential} of $f$ at $x$, namely the set of all subgradients of $f$ at $x$.

If $f$ is a concave function instead, then the \emph{superdifferential} of $f$ at $x$ is given by $\partial f(x) := -\partial (-f)(x)$. In that case, the elements of $\partial f(x)$ are called \emph{supergradients}.
\end{definition}

The subdifferential has the following well-known properties (see, e.g., \citep{Rockafellar70-convex-analysis}).
\begin{lemma}
Let $A \subseteq \RR^n$ be a convex set and let $f \colon A \to \RR$ be a convex function. Then it holds that:
\begin{itemize}
	\item $\partial f(x)$ is a closed convex set for all $x \in A$,
	\item $\partial f(x)$ is nonempty for all $x \in \textup{rel int } A$,
	\item if $f$ is differentiable at $x$, then $\partial f(x) = \{\nabla f(x)\}$,
	\item $x^\star \in A$ is a global minimum of $f$ on $A$, if and only if $0 \in \partial f(x^\star)$.
\end{itemize}
\end{lemma}

Note that since $f(\cdot\,;w)$ and $g_i(\cdot\,;w)$, $i=1,\dots, k$, are convex functions defined over $\RR^n$, the subdifferentials $\partial f(\cdot\,;w)$, $\partial g_i(\cdot\,;w) \colon \RR^n \rightrightarrows \RR^n$, $i=1,\dots,k$, are guaranteed to exist and be nonempty. For our purposes we will assume that we are given pseudogates computing these subgradients. In other words, we assume that we are given algebraic circuits $G_{\partial f}, G_{\partial g_i} \colon \RR^n \times \RR^s \times [0,1]^\ell \rightarrow \RR^n \times [0,1]^\ell$, $i=1,\dots,k$, such that $\fix_\ell[G_{\partial f}](x,w) \subseteq \partial f(x\,;w)$ and $\fix_\ell[G_{\partial g_i}](x,w) \subseteq \partial g_i(x\,;w)$, $i=1,\dots,k$, for all $x,w \in \RR^n \times \RR^s$. See \cref{app:PDC} for an example of how such pseudogates can be constructed.

\begin{example}
As an example, let us see why our convex optimization setting \eqref{eq:standard-CP} is indeed a generalization of our LP setting \eqref{eq:standard-LP}. To go from \eqref{eq:standard-LP} to \eqref{eq:standard-CP} we set $s = n + kn + k$ and decompose $w = (c,C,d) \in \RR^n \times \RR^{k \times n} \times \RR^k$ accordingly. Then we let $f(x\,;w) = - c^\transpose x$ and $g_i(x\,;w) = C_i^\transpose x - d_i$ for $i=1,\dots,k$, where $C_i$ denotes the $i$th row of $C$. As a result, it is easy to see that the subdifferentials are in fact gradients, namely $\nabla f(x\,;w) = -c$ and $\nabla g_i(x\,;w) = C_i$. Clearly, the functions are convex for any fixed value of $w$, and both the functions and their subdifferentials can easily be expressed as algebraic circuits. In particular, this means that \cref{thm:OPT-gate-convex} below implies \cref{thm:OPT-gate-LP}.
\end{example}

\noindent We can now formally state the main result of this section.
\begin{theorem}\label{thm:OPT-gate-convex}
Given $n \in \NN$, $m, k, s \in \NN_0$ and $G_{\partial f}$, $g_i$, $G_{\partial g_i}$, $i=1,\dots,k$, we can construct an algebraic circuit $G \colon P_{n,m,f,g} \times [0,1]^\ell \rightarrow \RR^n \times [0,1]^\ell$ in time $\poly(n,m,k,s,\size(G_{\partial f}),\size(g),\size(G_{\partial g_i}))$ such that for any parameters $p = (w,A,b,R) \in P_{n,m,f,g}$ it holds:
\begin{itemize}
	\item if the feasible region of $\CP(p)$ is non-empty, i.e., $\Feas(\CP(p)) \neq \emptyset$, then
	\[\fix_\ell[G](p) \subseteq \Feas(\CP(p)).\]
	\item if $\CP(p)$ satisfies the explicit Slater condition, then
	\[\fix_\ell[G](p) \subseteq \Opt(\CP(p)).\]
\end{itemize}
\end{theorem}

\begin{remark}
First of all, note that according to the statement of \cref{thm:OPT-gate-convex}, the construction does not actually need access to an algebraic circuit computing $f$, but only to a pseudogate for $\partial f$. However, it requires both a circuit for $g_i$ and a pseudogate for $\partial g_i$.
Furthermore, a careful examination of the proof of \cref{thm:OPT-gate-convex} below reveals that the construction also works if the functions $f$ and $g_i$ are pseudoconvex, instead of convex. A differentiable function $f \colon A \to \RR$, where $A \subseteq \RR^n$ is an open convex set, is said to be \emph{pseudoconvex} if for all $x,y \in A$ it holds that
\[f(y) < f(x) \implies \nabla f(x) \cdot (y-x) < 0.\]
Any differentiable convex function is pseudoconvex, but every pseudoconvex function is not necessarily convex. Furthermore, every convex function is not necessarily pseudoconvex, because it might not be differentiable. It is possible to define a notion that generalizes both convexity and pseudoconvexity, and to state \cref{thm:OPT-gate-convex} with this notion, but for simplicity we have stated it only for convex functions above.
\end{remark}

\subsection{Proof of \cref{thm:OPT-gate-convex}}

\paragraph*{High-level idea.}
We want to construct a circuit $G$ that takes as input $(w,A,b,R) \in P_{n,m,f,g}$ and $y \in [0,1]^n$, and outputs $z \in \RR^n$ and $\overline{y} \in [0,1]^n$, such that if $y = \overline{y}$, then $z$ is an optimal solution of $\CP(w,A,b,R)$ (when $\CP(w,A,b,R)$ satisfies the explicit Slater condition). The circuit $G$ will roughly perform the following computations:
\begin{enumerate}
    \item Compute $x := 2Ry - R$, i.e., scale $y \in [0,1]^n$ into a point $x$ in $[-R,R]^n$.
    \item Compute $\mu_i \in \Heaviside(g_i(x\,;w))$ and $\lambda_j \in 2 \Heaviside(a_j \cdot x - b_j) - 1$ for $i = 1, \dots, k$, and $j = 1, \dots, m$, where $\Heaviside$ denotes the Heaviside function.
    \item Compute $\mu_0 := 1 - \max(\mu_1,\dots,\mu_k,|\lambda_1|,\dots,|\lambda_m|)$.
    \item Compute $v_0 \in \partial f(x\,;w)$ and $v_i \in \partial g_i(x\,;w)$ for $i = 1, \dots, k$.
    \item Compute
    \[z := \Pi_R \Bigg( x - \mu_0 v_0 - \sum_{i=1}^k \mu_i v_i - \sum_{j=1}^m \lambda_j a_j \Bigg)\]
    where $\Pi_R$ denotes projection to $[-R,R]^n$.
    \item Compute $\overline{y} := (z+R)/2R$, i.e., scale $z \in [-R,R]^n$ back into a point $\overline{y}$ in $[0,1]^n$.
    \item Output $(z,\overline{y})$.
\end{enumerate}
Note that steps 2 and 4 compute correspondences and will actually be implemented by pseudogates, which will require us to add more auxiliary variables (which we denote by $y'$ in the full construction).

Now assume that $y = \overline{y}$. By construction of the circuit, it follows that $x = z$. For simplicity, assume that $x \in (-R,R)^n$ (the case where $x$ lies on the boundary is handled in the full proof below). Then, from $x=z$ it follows that
\[\mu_0 v_0 - \sum_{i=1}^k \mu_i v_i - \sum_{j=1}^m \lambda_j a_j = 0.\]
Using this equation, we show that by construction of $\mu$ and $\lambda$, it must be that $x \in \Feas(\CP((w,A,b,R)))$, if this feasible region is non-empty.
Furthermore, again by construction, we additionally have that $(\mu, \lambda) \neq (0,0)$, $\mu \geq 0$ (componentwise), and $g_i(x\,;w) < 0 \implies \mu_i = 0$ for all $i = 1, \dots, k$. Taking all these conditions together, it follows that $x$ satisfies the so-called Fritz John conditions \citep{John48-Fritz,MangarasianF67-Fritz-John}, which are necessary conditions for optimality of $x$. Now, by using the explicit Slater condition, we can show that $\mu_0 = 0$. In that case, the Fritz John conditions become the well-known Karush-Kuhn-Tucker (KKT) conditions \citep{Karush39-thesis,KuhnT51-nonlinear-programming}. Given that we have a convex program, the KKT conditions are also sufficient for optimality. Thus, it follows that $x$ is an optimal solution, i.e., $x \in \Opt(\CP((w,A,b,R)))$. The full proof that we present below does not assume any knowledge of the various optimality conditions mentioned here.

\paragraph*{Notation.}
Let $G_{\Heaviside} \colon \RR \times [0,1] \to \RR \times [0,1]$ denote the algebraic circuit which implements the pseudogate computing the Heaviside function $\Heaviside$, as given by \cref{lem:Heaviside-gate}, i.e., such that $\fix_1[G_{\Heaviside}](x) \subseteq \Heaviside(x)$ for all $x \in \RR$.

For $i=1,\dots,k$ let $g_i \colon \RR^n \times \RR^s \to \RR$ denote the algebraic circuits computing the functions for the inequality constraints. For $i=1,\dots,k$ let $G_{\partial g_i} \colon \RR^n \times \RR^s \times [0,1]^t \to \RR^n \times [0,1]^t$ denote an algebraic circuit which implements a pseudogate computing the subdifferential $\partial g_i$, i.e., such that $\fix_t[G_{\partial g_i}](x,w) \subseteq \partial g_i(x\,;w)$ for all $(x,w) \in \RR^n \times \RR^s$. Similarly, let $G_{\partial f} \colon \RR^n \times \RR^s \times [0,1]^t \to \RR^n \times [0,1]^t$ denote an algebraic circuit which implements a pseudogate computing the subdifferential $\partial f$, i.e., such that $\fix_t[G_{\partial f}](x,w) \subseteq \partial f(x\,;w)$ for all $(x,w) \in \RR^n \times \RR^s$. Note that we have assumed that $G_{\partial f}$, $G_{\partial g_1}, \dots, G_{\partial g_k}$ all use the same number $t$ of auxiliary inputs/outputs. This is without loss of generality, because additional auxiliary inputs/outputs can be added to such a circuit without altering the represented correspondence.

\paragraph*{Construction of the algebraic circuit $\boldsymbol{G}$.}
We now describe in detail the construction of the algebraic circuit $G \colon P_{n,m,f,g} \times [0,1]^\ell \rightarrow \RR^n \times [0,1]^\ell$. 
The circuit $G$ has exactly $\ell = n + k + m + t(k+1)$ auxiliary inputs/outputs. We denote the input to circuit $G$ by $(w,A,b,R,y,y')$ where $(w,A,b,R) \in P_{n,m,f,g}$ are the parameters for the convex program (i.e., the inputs to the pseudogate we are constructing), $y \in [0,1]^n$ are the first $n$ auxiliary inputs, and $y' \in [0,1]^{\ell-n}$ are the remaining $\ell-n = k + m + t(k+1)$ auxiliary inputs. The output of the circuit is denoted by $(z,\overline{y},\overline{y}')$, where $z \in \RR^n$ is the primary output (i.e., the actual output of the pseudogate we are constructing), $\overline{y} \in [0,1]^n$ are the first $n$ auxiliary outputs, and $\overline{y}' \in [0,1]^{\ell-n}$ are the remaining $n-\ell$ auxiliary outputs. We now describe how the outputs of the circuit are computed using the inputs and standard algebraic gates.

The circuit $G$ begins by computing the vector $x \in [-R,R]^n$ as $x := 2Ry-R$. This simply corresponds to scaling $y \in [0,1]^n$ to a vector in $[-R,R]^n$, and can clearly be computed using the standard algebraic gates. Next, $G$ uses the given algebraic circuits $g_i$ to compute $g_1(x\,;w), \dots, g_k(x\,;w)$. Then, for each $i=1,\dots,k$, the circuit computes $\mu_i \in \Heaviside(g_i(x\,;w))$ by using the pseudogate computing $\Heaviside$. To be more precise, the circuit computes $(\mu_i,\overline{y}_i') := G_{\Heaviside}(g_i(x\,;w),y_i')$, using the algebraic circuit $G_{\Heaviside}$. Note that when $y_i' = \overline{y}_i'$, we indeed have $\mu_i \in \Heaviside(g_i(x\,;w))$, as desired.

For $j=1, \dots, m$ let $a_j \in \RR^n$ denote the $j$th row of the matrix $A$. In particular, the $j$th equality constraint can be written as $a_j \cdot x = b_j$. The next step is to compute $\lambda_j \in 2 \Heaviside(a_j \cdot x - b_j) - 1$ for each $j=1, \dots, m$, again by using the pseudogate computing $\Heaviside$. Formally, this means that the circuit sets $(\lambda_j',\overline{y}_{k+j}') := G_{\Heaviside}(a_j \cdot x - b_j,y_{k+j}')$ and then $\lambda_j := 2\lambda_j' - 1$. Note that the computation of the $\mu_i$'s and the $\lambda_j$'s has used up exactly $k + m$ coordinates of the auxiliary inputs/outputs $y',\overline{y}'$, which means that $t(k+1)$ are still available at this point.

Next, the circuit computes $v_0 \in \partial f(x\,;w)$ and $v_i \in \partial g_i(x\,;w)$ for $i=1, \dots, k$. Formally, this is achieved by setting $(v_0,\overline{y}_{(0)}') := G_{\partial f}(x,w,y_{(0)}')$ and $(v_i,\overline{y}_{(i)}') := G_{\partial g_i}(x,w,y_{(i)}')$ for $i=1, \dots, k$, where $y_{(i)}' =$ $(y_{k+m +it+1}',$ $\dots,$ $y_{k+m +it+t}')$ for $i=0,1,\dots,k$, and $\overline{y}_{(i)}'$ is defined analogously.

We summarize some properties of the construction up to this point in the following claim.

\begin{claim}\label{clm:opt-gate-1}
If $y' = \overline{y}'$, then we have:
\begin{itemize}
	\item $\mu_i \in \Heaviside(g_i(x\,;w))$ for $i=1, \dots, k$,
	\item $\lambda_j \in 2 \Heaviside(a_j \cdot x - b_j) - 1$ for $j=1, \dots, m$,
	\item $v_0 \in \partial f(x\,;w)$,
	\item $v_i \in \partial g_i(x\,;w)$ for $i=1, \dots, k$.
\end{itemize}
\end{claim}

\noindent We are now ready to finish the construction of $G$.
The circuit computes
\[\mu_0 := 1 - \max(\mu_1,\dots,\mu_k,|\lambda_1|,\dots,|\lambda_m|).\]
Note that $|\lambda_j|$ can simply be computed as $\max(\lambda_j,-\lambda_j)$. We let $\Pi_R: \RR^n \to [-R,R]^n$ denote the projection onto $[-R,R]^n$. The function $\Pi_R$ can easily be computed using algebraic gates, since it suffices to apply the function $\alpha \mapsto \max(-R, \min(R,\alpha))$ to each coordinate separately. The primary output $z$ of $G$ is computed as
\begin{equation}\label{eq:opt-gate-output}
z := \Pi_R \Bigg( x - \mu_0 v_0 - \sum_{i=1}^k \mu_i v_i - \sum_{j=1}^m \lambda_j a_j \Bigg)
\end{equation}
and the auxiliary output $\overline{y} \in [0,1]^n$ of $G$ is then computed as
\[\overline{y} := \frac{z + R}{2R}\]
which, in particular, implies that $\overline{y} \in [0,1]^n$. Note that here it is important that we always have $R > 0$.

This completes the construction of the circuit $G$. Clearly, the construction can be performed in time $\poly(n,m,k,s,\size(G_{\partial f}),\size(g_1), \dots, \size(g_k),\size(G_{\partial g_1}), \dots, \size(G_{\partial g_k}))$. Note that, in particular, we have not used a circuit computing $f$ at any point in the construction.

\paragraph*{Fixed-point properties.}
In \cref{clm:opt-gate-1} we have already noted some properties that must hold when $y'=\overline{y}'$. Now we consider the implications of $y = \overline{y}$. First of all, when $y = \overline{y}$, it follows that $x = z$, since $x = 2Ry-R$ and $z = 2R \overline{y}-R$. From \eqref{eq:opt-gate-output} we then obtain:

\begin{claim}\label{clm:opt-gate-2}
If $y = \overline{y}$, then $x = z$ and the vector
\begin{equation}\label{eq:opt-gate-nu}
\nu := \mu_0 v_0 + \sum_{i=1}^k \mu_i v_i + \sum_{j=1}^m \lambda_j a_j
\end{equation}
satisfies, for $r=1,\dots,n$,
\[\nu_r > 0 \implies x_r=-R\]
and
\[\nu_r < 0 \implies x_r=R.\]
\end{claim}

\noindent Next, we prove the following technical result, which will be useful for the remainder of the proof.

\begin{claim}\label{clm:opt-gate-3}
If $(y,y') = (\overline{y},\overline{y}')$ and $u$ is a feasible point, i.e., $u \in \Feas(\CP(w,A,b,R))$, then
\begin{itemize}
	\item $\nu_r(u_r-x_r) \geq 0$ for $r=1,\dots,n$,
	\item $\mu_i v_i \cdot (u-x) \leq 0$ for $i=1,\dots,k$,
	\item $\lambda_j a_j \cdot (u-x) \leq 0$ for $j=1,\dots,m$.
\end{itemize}
Furthermore, if $\mu_0=0$, then all these terms are equal to zero.
\end{claim}

\begin{proof}
Since $u \in \Feas(\CP(w,A,b,R))$, it holds that $u \in [-R,R]^n$, $Au=b$ and $g_i(u\,;w) \leq 0$ for $i=1,\dots,k$. It follows that $\nu_r \cdot (u_r-x_r) \geq 0$ for $r=1,\dots,n$, because
\[\nu_r > 0 \implies x_r = -R \implies u_r \geq x_r\]
and
\[\nu_r < 0 \implies x_r = R \implies u_r \leq x_r\]
where we used \cref{clm:opt-gate-2} and the fact that $u \in [-R,R]^n$.

By \cref{clm:opt-gate-1} we know that $\mu_i \geq 0$ for $i=1,\dots,k$. Now, if $\mu_i > 0$ for some $i$, then it must be that $g_i(x\,;w) \geq 0$. But this means that $g_i(x\,;w) \geq g_i(u\,;w)$, because $u$ is feasible. Since $v_i \in \partial g_i(x\,;w)$ (\cref{clm:opt-gate-1}), and by the definition of subgradients (\cref{def:subgradient}), it follows that $v_i \cdot (u-x) \leq g_i(u\,;w) - g_i(x\,;w) \leq 0$. As a result, we obtain that $\mu_i v_i \cdot (u-x) \leq 0$ for all $i=1,\dots,k$.

If $\lambda_j > 0$ for some $j$, then by \cref{clm:opt-gate-1} we have $a_j \cdot x -b_j \geq 0$. Since $u$ is feasible, we have $a_j \cdot u -b_j = 0$ and thus $a_j \cdot (u-x) \leq 0$. Similarly, if $\lambda_j < 0$ for some $j$, then by \cref{clm:opt-gate-1} we have $a_j \cdot x -b_j \leq 0$, which by feasibility of $u$ yields $a_j \cdot (u-x) \geq 0$. As a result, we obtain that $\lambda_j a_j \cdot (u-x) \leq 0$ for all $j=1,\dots,m$.

Finally, consider the case where $\mu_0=0$. Taking the inner product of \eqref{eq:opt-gate-nu} with $(u-x)$, we obtain
\[\sum_{r=1}^n \nu_r (u_r-x_r) = \sum_{i=1}^k \mu_i v_i \cdot (u-x) + \sum_{j=1}^m \lambda_j a_j \cdot (u-x)\]
which, together with the above, implies that all the terms must be zero.
\end{proof}

We are now ready to prove the desired properties of $G$ in the following two claims. Recall that $z$ is the primary output of the circuit $G$, i.e., the output of the pseudogate it computes.

\begin{claim}\label{clm:opt-gate-feasible}
If $(y,y') = (\overline{y},\overline{y}')$ and $\Feas(\CP(w,A,b,R)) \neq \emptyset$, then $z \in \Feas(\CP(w,A,b,R))$.
\end{claim}

\begin{proof}
We will show that $x \in \Feas(\CP(w,A,b,R))$, which suffices to prove the claim since $x=z$ by \cref{clm:opt-gate-2}. 
Since $\Feas(\CP(w,A,b,R)) \neq \emptyset$, there exists a feasible vector $u$, i.e., $u \in [-R,R]^n$ such that $Au=b$ and $g_i(u\,;w) \leq 0$ for $i=1,\dots,k$.

Now, towards a contradiction, let us assume that $x \notin \Feas(\CP(w,A,b,R))$. Since $x \in [-R,R]^n$, this means that there exists $i^\star$ with $g_{i^\star}(x\,;w) > 0$, or $j^\star$ with $a_{j^\star} \cdot x \neq b_{j^\star}$. In both cases, it follows that $\mu_0 = 0$, since by \cref{clm:opt-gate-1}, $\mu_{i^\star}=1$ or $\lambda_{j^\star} \in \{-1,1\}$, respectively. By \cref{clm:opt-gate-3}, it follows that $\mu_i v_i \cdot (u-x) = 0$ for all $i$, and $\lambda_j a_j \cdot (u-x) = 0$ for all $j$.

If there exists $i^\star$ with $g_{i^\star}(x\,;w) > 0$, then by \cref{clm:opt-gate-1} we have that $\mu_{i^\star} = 1 > 0$. Furthermore, since $v_{i^\star} \in \partial g_{i^\star}(x\,;w)$ (\cref{clm:opt-gate-1}), it follows by the definition of subgradients (\cref{def:subgradient}) that $v_{i^\star} \cdot (u-x) \leq g_{i^\star}(u\,;w) - g_{i^\star}(x\,;w) < 0$, since $u$ is feasible. But this means that $\mu_{i^\star} v_{i^\star} \cdot (u-x) < 0$, a contradiction.

It remains to consider the case where there exists $j^\star$ with $a_{j^\star} \cdot x \neq b_{j^\star}$. If $a_{j^\star} \cdot x > b_{j^\star}$, then $\lambda_{j^\star} = 1 > 0$ (\cref{clm:opt-gate-1}), and $a_{j^\star} \cdot (u-x) < 0$, since $u$ is feasible. On the other hand, if $a_{j^\star} \cdot x < b_{j^\star}$, then $\lambda_{j^\star} = -1 < 0$ (\cref{clm:opt-gate-1}), and $a_{j^\star} \cdot (u-x) > 0$, since $u$ is feasible. As a result, in both cases we obtain that $\lambda_{j^\star} a_{j^\star} \cdot (u-x) < 0$, a contradiction.

Since we have obtained a contradiction in all cases, it must be that $x \in \Feas(\CP(w,A,b,R))$.
\end{proof}

\begin{claim}
If $(y,y') = (\overline{y},\overline{y}')$ and $\CP(w,A,b,R)$ satisfies the explicit Slater condition, then we have $z \in \Opt(\CP(w,A,b,R))$.
\end{claim}

\begin{proof}
By \cref{clm:opt-gate-2}, $x=z$, and thus it suffices to show that $x \in \Opt(\CP(w,A,b,R))$. By \cref{clm:opt-gate-feasible}, we already know that $x$ is feasible for $\CP(w,A,b,R)$.

By \cref{clm:opt-gate-1}, we always have $\mu_0 \geq 0$. Let us first consider the case where $\mu_0 > 0$. Let $u \in [-R,R]^n$ be any feasible point. Taking the inner product of \eqref{eq:opt-gate-nu} with $(u-x)$ we obtain
\[\mu_0 v_0 \cdot (u-x) = \sum_{r=1}^n \nu_r (u_r-x_r) - \sum_{i=1}^k \mu_i v_i \cdot (u-x) - \sum_{j=1}^m \lambda_j a_j \cdot (u-x).\]
By \cref{clm:opt-gate-3}, all the terms on the right hand side are non-negative. This implies that $\mu_0 v_0 \cdot (u-x) \geq 0$ and thus $v_0 \cdot (u-x) \geq 0$. Since $v_0 \in \partial f(x\,;w)$ (\cref{clm:opt-gate-1}), by the definition of subgradients (\cref{def:subgradient}), it follows that $f(u)-f(x) \geq v_0 \cdot (u-x) \geq 0$. Since this holds for any feasible point $u$, this means that $x$ is an optimal solution, i.e., $x \in \Opt(\CP(w,A,b,R))$.

It remains to handle the case where $\mu_0 = 0$. We will show that this case cannot occur. Towards a contradiction, assume that indeed $\mu_0 = 0$. Since $\CP(w,A,b,R)$ satisfies the explicit Slater condition, there exists $u \in (-R,R)^n$ with $Au=b$ and $g_i(u\,;w) < 0$ for $i=1, \dots, k$. In particular, $u$ is feasible and since $\mu_0=0$, by \cref{clm:opt-gate-3} we obtain that $\nu_r (u_r-x_r) = 0$ for all $r$, $\mu_i v_i \cdot (u-x) = 0$ for all $i$, and $\lambda_j a_j \cdot (u-x) = 0$ for all $j$.

If $\nu_r > 0$ for some $r$, then, by \cref{clm:opt-gate-2}, we have $x_r=-R$. Since $u_r \in (-R,R)$, it follows that $u_r-x_r > 0$, and thus $\nu_r (u_r-x_r) > 0$, a contradiction. If $\nu_r < 0$ for some $r$, then, by \cref{clm:opt-gate-2}, we have $x_r=R$, and thus again $\nu_r (u_r-x_r) > 0$, a contradiction. As a result, we obtain that $\nu_r = 0$ for all $r=1,\dots,n$.

If $\mu_i > 0$ for some $i$, then by \cref{clm:opt-gate-1} it must be that $g_i(x\,;w) \geq 0$. Since $v_i \in \partial g_i(x\,;w)$ (\cref{clm:opt-gate-1}), it follows by the definition of subgradients (\cref{def:subgradient}) that $v_i \cdot (u-x) \leq g_i(u\,;w) - g_i(x\,;w) < 0$, because $g_i(u\,;w) < 0$. Thus, we obtain that $\mu_i v_i \cdot (u-x) < 0$, a contradiction. As a result, we have $\mu_i = 0$ for all $i=1,\dots,k$.

Now, since $\mu_0=\mu_i=\nu_r=0$, the equation in \cref{clm:opt-gate-2} just yields $\sum_{j=1}^m \lambda_j a_j = 0$. But $\CP(w,A,b,R)$ satisfies the explicit Slater condition, so the vectors $a_j$, $j=1,\dots,m$, are linearly independent. It follows that $\lambda_j=0$ for all $j=1,\dots,m$. However, note that this is a contradiction, because according to the construction of $\mu_0$, if $\mu_i=0$, for all $i=1,\dots,k$, and $\lambda_j=0$, for all $j=1,\dots,m$, then $\mu_0=1$.
\end{proof}

The proof of \cref{thm:OPT-gate-convex} is thus completed.

\section{Applications to Game Theory}\label{sec:games}

In this section, we discuss further applications of our technique to equilibrium computation in strategic games. In \cref{sec:OPT-gate-Nash}, we already demonstrated how the employment of our OPT-gate can make the FIXP-membership proof of normal form games essentially straightforward. In this section, we provide further FIXP-membership results, namely:
\begin{itemize}[leftmargin=*]
\item[-] Computing equilibria in \emph{concave $n$-player games} \citep{rosen1965concave}. In these games, which generalize the normal form games mentioned above, the players have continuous strategy spaces and continuous payoff functions. Again via a relatively simple proof based on convex programs rather than Linear Programs, we show that computing equilibria of these games is in \FIXP; the \FIXP-completeness follows from the \FIXP-hardness of normal form games due to \citet{SICOMP:EtessamiY10}.
\item[-] Computing \emph{\eps-proper equilibria} in normal form games, an equilibrium refinement due to \citet{IJGT:Myerson78}. We show that the corresponding problem is in \FIXP. Our proof first shows how to compute solutions to \emph{systems of conditional convex constraints} (\cref{sec:cond-conv-constraints}) using our OPT-gate, and then employs this to show the FIXP-membership result for \eps-proper equilibria.  
\item[-] Computing \emph{stationary $\lambda$-discounted equilibria} in $n$-player stochastic games \citep{PNAS:Shapley53}. We show that this problem is \FIXP-complete. The FIXP-membership could technically already be achieved via the machinery used by \citet{SICOMP:EtessamiY10} to achieve the FIXP-membership of the 2-player problem, but our proof uses the OPT-gate here as well.
The \FIXP-hardness follows from \citep{SICOMP:EtessamiY10} by viewing a normal form game as a stochastic game consisting of a single state.
\end{itemize}

\subsection{Concave \texorpdfstring{$\boldsymbol{n}$}{n}-player games}\label{sec:concave-games}

In this section, we generalize the \FIXP-membership result from normal form games to concave games, a class of games studied by \citet{rosen1965concave}, which we define below. Together with the \FIXP-hardness result for normal form games \citep{SICOMP:EtessamiY10}, we obtain the following result.

\begin{theorem}\label{thm:concave-games}
Computing an equilibrium of a concave $n$-player game is \FIXP-complete.
\end{theorem}

\noindent An $n$-player game $G$ consists of $n$ players, with each player $i\in [n]$ having a
compact and convex strategy space $\Sigma_i\subseteq\RR^{m_i}$ and a continuous 
payoff function $u_i\colon \Sigma\rightarrow\RR$, where $\Sigma:= \Sigma_1\times
\cdots\times\Sigma_n$. For a strategy profile $x\in \Sigma$, let $(y,x_{-i})=(x_1,
\dots, x_{i-1},y,x_{i+1},\dots, x_n)$ denote the strategy profile where player $i$ 
unilaterally changes strategy. The game $G$ is \emph{concave} if for any player $i$
and fixed $x\in\Sigma$, the function $u_i(y,x_{-i})$ is concave in $y$. 

A strategy profile $x\in\Sigma$ is an \emph{equilibrium} if $u_i(x)=\max_{y\in\Sigma_i}
u_i(y,x_{-i})$ for every player $i$, that is, no player can increase its payoff by a 
unilateral change of strategy. \citet{rosen1965concave} provides a proof that every concave game
admits an equilibrium point by constructing an upper hemicontinuous correspondence 
$F\colon\Sigma\rightarrow 2^{\Sigma}$ mapping any $x\in\Sigma$ to the set $\{y\in\Sigma\mid \forall i\colon 
y_i\mbox{ is a best response to }x_{-i}\}$. One then applies Kakutani's fixed point theorem
to show the existence of some strategy profile $x\in\Sigma$ with $x\in F(x)$. As $x_i$ is
a best response to $x_{-i}$ for all players by construction, $x$ is an equilibrium. 

\paragraph{Computational problem.}
In the computational problem, we assume that each strategy space $\Sigma_i$ is given
as the set of all $x\in [-R_i,R_i]^{m_i}$ satisfying equality constraints $A_i x=b_i$ and inequality 
constraints $g_{ij}(x)\leq 0$ for some $R_{i}>0$, matrix $A_i\in\RR^{d_i\times m_i}$, 
vector $b_i\in\RR^{d_i}$, and convex functions $g_{ij},$ $j=1,\dots, k_i$. As $A_i$
is given as input, we may apply preprocessing to eliminate linear dependence among the rows,
so we just have to assume that the constraints satisfy the general Slater condition. We also assume that we are given algebraic circuits for $g_{ij}$ and pseudogates computing
their subdifferentials $\partial g_{ij}$, as well as the superdifferentials $\partial u_i$.

In order to prove \FIXP-membership, we construct a circuit $F\colon D \rightarrow D$, where $D = \bigtimes_i [-R_i,R_i]^{m_i}$. On input $x$, the $i$th output of the circuit $F$ will be a best response of player $i$ to $x_{-i}$. Namely, the $i$-th output of $F$ is simply set as the output of the OPT-gate for the following convex program:
\begin{equation*}
\begin{aligned}
\mbox{minimize}\quad & -u_i(y,x_{-i})\\
\mbox{subject to}\quad & A_i y = b_i\\
& g_{ij}(y) \leq 0 \qquad j=1,\dots, k_i\\
& y \in [-R_i,R_i]^{m_i}
\end{aligned}
\end{equation*}
Note that the explicit Slater condition is satisfied by assumption. Thus, the OPT-gate correctly solves the convex program. As a result, if $x$ is a fixed point of $F$, then $x \in \Sigma$ and $x_i$ is a best response to $x_{-i}$ for every player
$i$, meaning that $x$ is indeed an equilibrium of the game.

\begin{remark}
\citet{rosen1965concave} actually considers a more general setting where the space of strategy
profiles $\Sigma$ is not assumed to be equal to the product space of the players' individual strategy spaces. 
Rather, he just assumes that the space of strategies is a compact and convex subset 
$\Sigma\subseteq\Sigma_1\times\cdots\times\Sigma_n$. Let us assume that $\Sigma$ is given
as the set of all $x\in [-R,R]^m$ that satisfy $Ax=b$ and $g_j(x)\leq 0$, $j=1,\dots, k$, where, as per usual, $A$ is a matrix, $b$ 
a vector and the $g_j$ are convex functions. We may write $A = (A_1\mid\cdots\mid A_n)$ as a concatenation
of block matrices. For fixed $x\in\Sigma$, player $i$ would then maximize its utility $u_i(y,x_{-i})$ subject to
the constraints $A_i y = b_i(x_{-i}) := b-\sum_{j\neq i}A_j x_j$ and $g_j(y,x_{-i})\leq 0$ for all $j$. As $A_i$ is given in the input
we can apply preprocessing to it and remember the linear combinations required to eliminate potential
linear dependence in the rows of $A_i$. Applying these same linear combinations to the $A_i$ and $b_i$, we
obtain constraints $\tilde{A}_i y = \tilde{b}_i$. It suffices to require that the constraints
$A_i y=b_i(x_{-i})$ and $g_j(y,x_{-i})\leq 0$, $j=1,\dots, k$, satisfy the general Slater condition
for all $i\in [n]$ and $x\in\Sigma$. 
\end{remark}

\subsection{Computing an \eps-proper equilibrium via systems of conditional convex constraints}\label{sec:eps-proper}

In this section we consider the Nash equilibrium refinement of proper
equilibrium due to \citet{IJGT:Myerson78}. First, we define the
notion of $\eps$-proper equilibrium and then define a proper
equilibrium as a limit point of $\eps$-proper equilibria.
\begin{definition}[\citep{IJGT:Myerson78}] \label{def:eps-proper}
  Let $\Gamma$ be a finite $n$-player game in strategic form. Given
  $\eps>0$, a mixed strategy profile $x$ is an \emph{$\eps$-proper
  equilibrium} in $\Gamma$ if it is fully mixed and satisfies
  $x_{ik} \leq \eps x_{i\ell}$ whenever
  $u_i(k,x_{-i}) < u_i(\ell,x_{-i})$ for all players~$i$ and all pairs
  of actions $k,\ell$ of player~$i$.\medskip

  \noindent A mixed strategy profile $x$ is a \emph{proper equilibrium} if and only if it is a
  limit point of a sequence of $\eps$-proper
  equilibria with $\eps \rightarrow 0^+$.
\end{definition}

\noindent It was proved recently by \citet{EC:HansenL18} that the task of approximating a proper equilibrium is complete for the class $\FIXPa$ of \citep{SICOMP:EtessamiY10}. This work follows a line of similar results~\citep{SAGT:EtessamiHMS14,GEB:Etessami2021} for approximating other notions of equilibrium refinements, e.g.\ \citeauthor{selten1975reexamination}'s trembling hand perfect equilibrium \citep{selten1975reexamination}, that are, like proper equilibria, defined as limit points of certain $\eps$-equilibria. These previous results were proved by showing that the problems of computing the $\eps$-equilibria are in \FIXP. In fact they can be computed by a reduction to a basic \FIXP-problem (\cref{def:basic-fixp}) where $\eps$ is an input variable of the algebraic circuit. This additional property is exploited to prove that approximating the equilibrium refinement notions is in \FIXPa\ by using the ability of an algebraic circuit to compute a ``virtual infinitesimal'' by means of repeated squaring that then takes the place of \eps. 

\citet{EC:HansenL18} did not prove that computing an \eps-proper equilibrium is in \FIXP, but instead proved that computing a so-called $\delta$-almost \eps-proper equilibrium is in \FIXP. These equilibria can in fact be computed by reducing to a basic \FIXP-problem where $\delta$ and $\eps$ are inputs of the algebraic circuit. It is then shown that approximating a proper equilibrium is in \FIXPa\ by substituting ``virtual infinitesimals'' for both $\delta$ and $\eps$. The question of whether computing an \eps-proper equilibrium is in \FIXP\ was left as an open problem. Using our technique we resolve this question, and thereby significantly simplify the proof of \citet{EC:HansenL18} that approximating a proper equilibrium is \FIXPa-complete.

\begin{theorem}\label{thm:eps-proper}
  The problem of computing an $\eps$-proper equilibrium of a given finite game $n$-player game in normal form is in $\FIXP$.
\end{theorem}

\noindent To establish existence of \eps-proper equilibria, \citet{IJGT:Myerson78} made use of Kakutani's fixed point theorem (\cref{thm:kakutani}). Suppose that $\Gamma$ is a given $n$-player game in strategic form and $\eps>0$. Let $S_i=[m_i]$ and $u_i$ be the set of strategies and utility function of Player~$i$. Define $\eta_i(\eps)=\eps^{m_i}/{m_i}$ and let

\[\Sigma^{\eta_i}_i = \Big\{y \in \RR^{m_i} \Big| \sum_j y_j=1 ; \forall j:
y_j \geq \eta_i\Big.\Big\}
\]
be the set of \emph{$\eta_i$-perturbed} mixed strategies
for Player~$i$. Let $\eta=(\eta_1,\dots,\eta_n)$ and define $\Sigma^\eta = \prod_{i=1}^n \Sigma^{\eta_i}_{i}$ to be the set of
all $\eta$-perturbed mixed strategy profiles for $\Gamma$. Define the correspondence $F \colon \Sigma^\eta \rightrightarrows \Sigma^\eta$ by $F(x) = \prod_{i=1}^n F_i(x)$, where
\[
    F_i(x) = \{ y \in \Sigma^\eta \mid \forall k,\ell \in S_i : u_i(k,x_{-i}) < u_i(\ell,x_{-i}) \Rightarrow y_{ik} \leq \eps y_{i\ell}\} \enspace .
\]
Clearly any fixed point of $F$ is an $\eps$-proper equilibrium of $\Gamma$.
\citeauthor{IJGT:Myerson78} concluded his proof by showing that $F$ satisfies the conditions of the Kakutani fixed point theorem.
In particular, $F_i$ is nonempty since we have $y_i \in F_i(x)$ where
\[
    y_{ik} = \eps^{\rho_i(k)}/ \sum_{\ell \in S_i} \eps^{\rho_i(\ell)}
\]
and $\rho_i(k) = \abs{\{\ell \in S_i \mid u_i(k,x_{-i}) < u_i(\ell,x_{-i})\}}$. \medskip

\noindent Computing a fixed point of $F$ is a special case of the result of the following subsection (see \cref{thm:cond-conv-constraints}).

\subsubsection{Solving Systems of Conditional Convex Constraints}\label{sec:cond-conv-constraints}

In this section we consider the task of solving systems of what we
refer to as conditional convex constraints by finding fixed points. We make
use of the main result of the section (\cref{thm:cond-conv-constraints}) to prove
\cref{thm:eps-proper}, but it could be applicable to other problems as well, and therefore
it could be of independent interest.

\begin{definition}
  A conditional convex constraint on $n$ variables is a pair $(f,g)$
  of a continuous function $f: \RR^n \rightarrow \RR$ and a convex
  function $g: \RR^n \rightarrow \RR$. A point $x \in \RR^n$ satisfies $(f,g)$ if
  $f(x) > 0 \Rightarrow g(x) \leq 0$.
\end{definition}
A system of conditional convex constraints naturally defines a search
problem, where the task is to find a point $x$ that satisfies all
constraints of the system. A system of conditional convex constraints
also defines a correspondence in a natural way. We shall further
restrict our attention to correspondences with nonempty, compact, and
convex domain.
\begin{definition}
  Let $D \subseteq \RR^n$ be a non-empty, compact and convex set. Let
  $(f_1,g_1),\dots,(f_m,g_m)$ be conditional convex constraints on $n$
  variables.  The correspondence $F \colon D \rightrightarrows D$
  defined by $D$ and $(f_1,g_1),\dots,(f_m,g_m)$ is given by
  \[
    F(x) = \{y \in D \mid \forall i: f_i(x)>0 \Rightarrow g_i(y)\leq 0\} \enspace .
  \]
\end{definition}
Note that there is a one-to-one correspondence between fixed points of
$F$ and solutions of the system of constraints contained in $D$.

Except for the property of nonempty-valued, such correspondences
satisfy the conditions of Kakutani's fixed point theorem.
\begin{proposition}
  Let $F$ be a correspondence defined by a non-empty, compact and
  convex set $D \subseteq \RR^n$ and conditional convex constraints
  $(f_1,g_1),\dots,(f_m,g_m)$.  Then $F$ is uhc as well as compact and
  convex-valued.
\end{proposition}
\begin{proof}
  Let $x \in D$ and let $V \subseteq D$ be an open set such that
  $F(x) \subseteq V$. By continuity of the functions $f_i$ we may find
  an open set $U$ containing $x$ such that if $x' \in U$ and
  $f_i(x')>0$ then we have $f_i(x) >0$ as well. It follows that
  $F(x') \subseteq F(x) \subseteq V$, which means that $F$ is uhc.  We
  also have that $F(x)$ is an intersection of closed and convex sets,
  and $F(x)$ is thus closed and convex as well.
\end{proof}
Thus if we had a guarantee that $F$ was nonempty-valued as well, a
fixed point would be guaranteed by Kakutani's fixed point theorem. We
can associate a total search problem with $F$ where the task is to
find $x \in D$ such that either $x \in F(x)$ or $F(x)=\emptyset$.
For the computational problem, we assume that $D$ is given as a set of
linear constraints $x\in [-R,R]^n$ and $Ax=b$, and convex constraints
$h_i(x)\leq 0$, $i=1,\dots, k$, that satisfy the explicit Slater condition. We also assume that we are given algebraic circuits computing
$f_i,g_i$ and $h_i$, and pseudogates computing the subgradients of $g_i$ and $h_i$. 

The idea is that the function $G$ in the proof below is derived from a system of convex constraints in variables $y$ that are parameterized by variables $x$. We consider the constraints given by $\max(0,f_i(x))g_i(y)\leq 0$. When $f_i(x)>0$ this is equivalent to the constraint $g_i(y)\leq 0$. When $f_i(x)\leq 0$, the constraint becomes trivial.

\begin{theorem}\label{thm:cond-conv-constraints}
The problem of solving systems of conditional convex constraints is in $\FIXP$.
\end{theorem}
\begin{proof}
Define a circuit $G\colon [-R,R]^n \times [-R,R]^n \to [-R,R]^n \times [-R,R]^n$, by letting $G(x,y) := (\overline{x},\overline{y})$. Here $\overline{y}$ is computed as the output of our convex OPT-gate for the following feasibility
problem (parameterized by $x$)
\begin{equation*}
\begin{aligned}
\mbox{maximize}\quad & 0\\
\mbox{subject to}\quad & Az=b\\
 & h_i(z)\leq 0\qquad i = 1,\dots, k\\
 & z \in [-R,R]^n\\
 & \max(0,f_i(x))g_i(z)\leq 0\qquad i=1,\dots, m
\end{aligned}
\end{equation*}
and $\overline{x}$ is the projection of $y$ onto $D$, which can be obtained by using the OPT-gate to solve
\begin{equation*}
\begin{aligned}
\mbox{minimize}\quad & \|y-z\|_2^2\\
\mbox{subject to}\quad & Az=b\\
 & h_i(z)\leq 0\qquad i = 1,\dots, k\\
 & z \in [-R,R]^n
\end{aligned}
\end{equation*}
Suppose that $(x,y)$ is a fixed point of $G$. We argue that $x$ is a solution to
the search problem described above. First of all, note that $x \in D$, because $x = \overline{x}$ is the projection of some point (namely, $y$) onto $D$. Now there are two cases. If the set of feasible solutions to the first convex program is
empty, then it follows that $F(x)=\emptyset$, and so $x \in D$ is indeed a solution to our search problem. 
If, on the other hand, the set of feasible solutions
is non-empty, then it follows from the first part of \cref{thm:OPT-gate-convex} that
$\overline{y} = y$ is a solution to the first convex program. The first three constraints show that $y\in D$, and
from the inequalities $\max(0,f_i(x))g_i(y)\leq 0$ it follows that $f_i(x)>0\Rightarrow g_i(y)\leq 0$, i.e., $y \in F(x)$. But since $x$ is the projection of $y$ onto $D$, and $y \in D$, it must be that $x = y$, and thus $x \in F(x)$.
\end{proof}

\subsection{\texorpdfstring{$\boldsymbol{n}$}{n}-player Stochastic Games}\label{sec:stochastic-games}

Stochastic games, as first introduced by \citeauthor{PNAS:Shapley53} in his seminal
work~\citep{PNAS:Shapley53}, model dynamic interaction between players
in an environment whose state is changing according to a stochastic
process influenced by the actions of the players. We shall here
consider discrete-time finite games, where players receive immediate
payoffs in each round of play and discount future payoffs. \citeauthor{PNAS:Shapley53}'s
model then corresponds to the special case of two-player zero-sum games.\medskip

\noindent The main result of this section is the following.

\begin{theorem}\label{thm:stochastic-games}
Computing a \emph{stationary $\lambda$-discounted equilibrium} of an $n$-player stochastic game is \FIXP-complete.
\end{theorem}

\noindent Next, we first define $n$-player stochastic games formally, as well as the equilibrium notion which appears in the statement of \cref{thm:stochastic-games} above. Then we present the proof of \FIXP-membership for the problem; the \FIXP-hardness follows from \citep{SICOMP:EtessamiY10}, by considering a single-state stochastic game, i.e., a normal form game. \\

\noindent An $n$-player finite stochastic game $\Gamma$ is given as
follows. The game is played on a finite set of states $S$. Every
player~$i$ has a finite set of actions $A_i$. Let
$A = A_1\times \dots \times A_n$ denote the set of action profiles and
$P = \{(s,a) : s \in [n], a \in A\}$ the pairs of states and action
profiles. The immediate payoffs to player~$i$ are then given by a
function $u_i \colon P \rightarrow \RR$ and the state transitions are given
by a function $q \colon P \rightarrow \Simplex(S)$. Let $M = \max_{i,(s,a)\in P} \abs{u_i(s,a)}$.

A play of $\Gamma$ is an infinite sequence $h \in P^\infty$. A finite
play up to stage~$t$ is a sequence $h_t \in P^{t-1} \times S$. Let
$\calH = \mathbin{\dot{\cup}}_{t=1}^\infty \left(P^{t-1} \times
  S\right)$ denote the set of all finite plays. A behavioral strategy
for player~$i$ is a function
$\sigma_i \colon \calH \rightarrow \Simplex(A_i)$. A stationary
strategy is a behavioral strategy that depends only on the last state
of a finite play. A stationary strategy $x_i$ may thus be viewed as a
function $x_i \colon S \rightarrow \Simplex(A_i)$. Behavioral
strategies $\sigma_i$ for each player~$i$ form a behavioral strategy
profile $\sigma=(\sigma_1,\dots,\sigma_n)$. In the same way,
stationary strategies for each player form a stationary strategy
profile. A behavioral strategy profile $\sigma$ and an initial state
$s^1 \in S$ define by Kolmogorov's extension theorem a unique
probability distribution $\Pr_{s^1,\sigma}$ on plays
$(s^1,a^1,s^2,a^2,\dots)$, where the conditional probability of
$a^t=a$ given the play up to stage~$t$, $h_t=(s^1,a^1,\dots,s^t)$, is
equal to $\prod_{i=1}^n \Pr[\sigma_i(h_t)=a_i]$, and the conditional
probability of $s^{t+1}$ given $s^t$ and $a^t$ is equal to
$q(s^t,a^t)$. We denote by $\Exp_{s^1,\sigma}$ the expectation with
respect to $\Pr_{s_1,\sigma}$.\medskip

\noindent For every \emph{discount factor} $0<\lambda\leq 1$, the
$\lambda$-discounted payoff to player~$i$ is defined to be
\begin{equation}
  \gamma_i^\lambda(s^1,\sigma) = \Exp_{s^1,\sigma} \left[ \lambda \sum_{t=1}^\infty (1-\lambda)^{t-1}u_i(s^t,a^t)\right] \enspace .
\end{equation}
A behavioral strategy profile $\sigma$ is a $\lambda$-discounted equilibrium if
\begin{equation}
  \gamma_i^\lambda(s^1,\sigma) \geq \gamma_i^\lambda(s^1,(\sigma'_i,\sigma_{-i})) \enspace ,
\end{equation}
for all states $s_1 \in S$, all players $i \in [n]$, and all
behavioral strategies $\sigma'_i$ for player~$i$. It was proved by
\citet{JSHUA:Fink1964} and \citet{JSHUA:Takahashi1964}
that any finite discounted stochastic game has a $\lambda$-discounted
equilibrium in stationary strategies for any discount
factor~$\lambda$. In the case of two-player zero-sum games, \citeauthor{PNAS:Shapley53}
proved existence of the $\lambda$-discounted value
$v_\lambda \in \RR^S$, as well as optimal stationary strategies, for the
$\lambda$-discounted payoff. \medskip

\noindent The results of \citeauthor{PNAS:Shapley53}, \citeauthor{JSHUA:Fink1964}, and \citeauthor{JSHUA:Takahashi1964} lead to a natural real-valued total search problem. \citet{SICOMP:EtessamiY10} proved the $\FIXP$-membership of the
problem of finding the $\lambda$-discounted values and optimal
stationary strategies in two-player zero-sum stochastic games with
discounted payoffs.

Now let $\Gamma$ be an $n$-player stochastic game. The proofs by \citeauthor{JSHUA:Fink1964}
and \citeauthor{JSHUA:Takahashi1964} of existence of $\lambda$-discounted equilibrium in
stationary strategies both make use of Kakutani's fixed point
theorem (\cref{thm:kakutani}).\footnote{More accurately, the proof by \citet{JSHUA:Takahashi1964} applies to stochastic games with
infinite action spaces, and as a consequence uses a generalization of
Kakutani's fixed point to locally convex spaces due to \citet{PNAS:Fan52}
and \citet{PAMS:Glicksberg1952}.} Let us now consider the
approach of \citeauthor{JSHUA:Takahashi1964}, specialized to finite stochastic games. \medskip

\noindent Valuations of states $v_i \colon S \rightarrow \RR^n$ by every
player~$i$ now form a \emph{valuation profile} $v=(v_1,\dots,v_n)$.
Given a discount factor $\lambda$ and a valuation profile $v$, we can
form associated $n$-player normal form games $\Gamma_{s,\lambda}(v)$,
generalizing the case of two players. For every state $s \in S$, the
utility function $u_i^{s,\lambda,v} \colon A \rightarrow \RR$ of
player~$i$ in $\Gamma_{s,\lambda}(v)$ is given by
\[
  u_{i}^{s,\lambda,v} (a) = \lambda u_i(s,a) + (1-\lambda)\sum_{s' \in S} q(s' \mid s,a) v_i(s')
\]
A stationary strategy profile $x=(x_1,\dots,x_n)$ in $\Gamma_\lambda$
induces strategy profiles $x(s)=(x_1(s),\dots,x_n(s))$ in the games
$\Gamma_{s,\lambda}(v)$, and corresponding valuations of states
$u_i^{s,\lambda,v}(x(s))$. Let
\[D = ([-M,M]^S)^n \times (\Delta(A_1)^S \times \dots \times
\Delta(A_n)^S)\]
be the set of valuation profiles $v$ and stationary strategy
profiles $x$. \citeauthor{JSHUA:Takahashi1964} defines a correspondence
$F : D \rightrightarrows D$ whose fixed points are pairs of valuation
profiles and stationary strategy profiles such that the stationary
strategy profiles are $\lambda$-discounted equilibria in
$\Gamma_\lambda$. Letting $F(v,x)=(G(v,x),H(v,x))$, \citeauthor{JSHUA:Takahashi1964} defines
\[
  G(v,x)_{i,s} = \max_{y(s)_i \in \Delta(A_i)} u_i^{s,\lambda,v}(y(s)_i; x(s)_{-i})
\]
and
\[
  H(v,x)_{i,s} = \argmax_{y(s)_i \in \Delta(A_i)} u_i^{s,\lambda,v}(y(s)_i; x(s)_{-i})
\]
By Berge's maximum theorem (\cref{thm:berge}), the correspondence $F$ satisfies the
requirements of Kakutani's fixed point theorem (\cref{thm:kakutani}) which in turn yields
existence of a fixed point $(v,x)$. \citeauthor{JSHUA:Takahashi1964} then proves that the
stationary strategy profile $x$ is a $\lambda$-discounted equilibrium
in $\Gamma_\lambda$.\medskip

\noindent We obtain \FIXP-membership by replacing these correspondences by OPT gates, obtaining a circuit computing a function $F': D' \to D'$, where $D' := ([-M,M]^S)^n \times ([0,1]^{|A_1|})^S \times \dots \times
([0,1]^{|A_n|})^S$. On input $(v,x)$, simply consider for every player~$i$ and every state~$s$ the linear program computing best replies for player~$i$ in the game $\Gamma_{s,\lambda}$:
\begin{equation*}
\begin{aligned}
\mbox{maximize}\quad & u_i^{s,\lambda,v}(z,x(s)_{-i})\\
\mbox{subject to}\quad & \sum_{j=1}^{|A_i|} z_j = 1\\
& z \in [0,1]^{|A_i|}
\end{aligned}
\end{equation*}
Denote by $\overline{x}(s)_i$ the output of the OPT gate. Then $\overline{v}(s)_i=u_i^{s,\lambda,v}(\overline{x}(s)_i,x(s)_{-i})$ is computed and we let $F'(v,x)=(\overline{v},\overline{x})$. The set of fixed points of $F$ and $F'$ clearly coincide and the result follows.

\section{Applications to Cake Cutting}\label{sec:cakes}

In this section, we discuss the applications of our main technique to the area of fair division, and in particular to the complexity of the well-known \emph{envy-free cake cutting problem} \citep{gamow1958puzzle} (see also \citep{robertson1998cake,brams1996fair,procaccia2013cake}). In this problem, the cake serves as a metaphor for a divisible resource---represented without loss of generality by the interval $[0,1]$---which needs to be divided among a set of $n$ agents, such that every agent receives a piece of cake she prefers compared to any piece assigned to any other agent. In the general formulation of the problem, a ``piece'' can be a collection of possibly disconnected subintervals. In the \emph{contiguous} version, each piece is a single interval (and hence the cake is divided using $n-1$ cuts), and in that case any division of the cake can be represented as a point $x$ in the simplex $\Delta^{n-1}$, with $x_j$ denoting the $j$-th coordinate.

Formally, the valuation of agent $i$ for the $j$-th piece is given by a map $u_{ij}\colon\Delta^{n-1}\rightarrow \RRnn$. Given a division $x$, we say that agent $i$ prefers the $j$-th piece
if $u_{ij}(x)\geq u_{ik}(x)$ for every $k$. We say that a 
division $x$ is \textit{envy-free} if there exists a permutation $\pi$ of $\{1,\dots, n\}$
such that for every $i$, agent $i$ prefers piece $\pi(i)$. For the computational version of the problem, we will assume that the valuations are given by means of an algebraic circuit as defined in \cref{def:algebraic-circuit}.

\citet{AMM:Stromquist1980} proved that an envy-free division of the cake is guaranteed to exist, even for the contiguous version, as long as the valuations $u_{ij}$ are continuous functions and the agents are \textit{hungry}, that is no agent prefers an empty piece of cake.

\begin{theorem}[\citet{AMM:Stromquist1980}]
When the valuations are continuous and the agents are hungry, an envy-free division of the cake always exists.
\end{theorem}

\noindent Stromquist's proof considers the simplex of divisions as described above, and applies a variant of the K-K-M lemma of \citet{FM:KnasterKM1929} (see \cref{lem:kkm} in \cref{sec:kkm}), which regards the covering of the simplex by sets. While this lemma is defined on a continuous domain (the unit-simplex), \citet{AMM:Stromquist1980} actually applied it after constructing a subdivision of the simplex into ``cells'', first obtaining the existence of an \emph{approximately} envy-free division, and then invoked a limit argument to prove the existence for \emph{exact} envy-freeness. Another proof of existence was developed independently by Simmons in 1980 (published for the first time in \citep{AMM:Su1999} and attributed to Simmons as ``private communication to Michael Starbird''). This proof made use of the well-known Sperner's lemma from topology \citep{sperner1928neuer}, and therefore also works on a subdivision (a ``triangulation'') of the unit-simplex. Similarly to the proof of \citet{AMM:Stromquist1980}, the proof first establishes the existence of an approximately envy-free division, and then applies a limit argument to obtain the existence for the exact version.

Given their nature as described above, these existence proofs cannot be turned into a membership in FIXP. Indeed, the related literature has only gone as far as proving the membership of the approximate version of the problem in the class PPAD of \citet{JCSS:Papadimitriou1994}, a result due to \citet{OR:DengQS2012}. PPAD is fundamentally related to an approximate computational version of Brouwer's fixed point theorem, and in that sense can be seen as a discrete, approximate analogue of FIXP. Before our paper, the complexity of the \emph{exact} envy-free cake-cutting problem was not known. To this end, we provide the following theorem.

\begin{theorem}\label{thm:cake-cutting}
The envy-free cake cutting problem with very general valuations is $\FIXP$-complete.
\end{theorem}

\noindent In order to obtain the FIXP-membership result above, in the process we (implicitly) develop a new existence proof for the envy-free cake cutting problem, one which is not based on discrete subdivisions of the unit-simplex and limit arguments. Our proof is based on maximum flow computation on a network given by a bipartite graph with the agents on one side and the pieces of cake on the other side, as given by a division $x$. Using our OPT-gate from \cref{sec:OPT-gate}, we can turn this existence proof into a FIXP-membership result rather easily. Our approach is somewhat reminiscent of the only markedly different proof for envy-free cake cutting that we know of, that of \citet{woodall1980dividing}. This proof constructs a similar bipartite graph and uses Brouwer's fixed point theorem to prove the result, but crucially, it uses a discontinuous step (see Stage 2 of the construction in \citep{woodall1980dividing}) which impedes its applicability as a potential argument for a FIXP-membership proof. 

\begin{remark}[Very general valuations]\label{rem:general-valuations}
The statement of \cref{thm:cake-cutting} regards the case of \emph{very general} agents' valuations. Note that according to the definition of the problem in the beginning of the section, an agent has a possibly different valuation for each possible division of the cake. In this generality, the setting captures scenarios in which an agents' value for a piece is not necessarily the sum of her values for the subpieces that comprise the piece (i.e., the valuations are not necessarily \emph{additive}), or even where an agent's value for a piece could be smaller for her value for a subset of that piece (i.e., the valuations are not necessarily \emph{monotone}), and also scenarios that exhibit \emph{externalities}, as an agent might value differently two allocations that assign her the same piece. In several fair division textbooks (e.g., see \citep{robertson1998cake,brams1996fair}), the problem is typically presented for the case where the valuations are simply additive measures. On the other hand, the aforementioned existence proofs apply to the case of very general valuations.

Ideally of course, we would like to obtain a FIXP-hardness result for the version of the problem with additive valuations, and this is in fact a major open problem in the literature of computational fair division. That being said, our FIXP-hardness result is in fact in line with the PPAD-hardness result of \citet{OR:DengQS2012} for the approximate version of the problem, which only holds for very general valuations, in the same generality as we have defined them here. Besides, our focus in this paper is primarily establishing the FIXP-membership of interesting problems via our newly introduced technique in \cref{sec:OPT-gate}, and clearly the FIXP-membership result is stronger for the case that we consider, i.e., that of very general valuations.
\end{remark}

\noindent Before we proceed with the proof of \cref{thm:cake-cutting}, we remark that in \cref{sec:kkm}, we prove that the computational versions of several well-known topological lemmas and theorems, very much related to the cake cutting problem and its generalizations are also FIXP-complete. In particular, we show that the computational versions of the K-K-M lemma of \citet{FM:KnasterKM1929} that we mentioned above, its ``rainbow'' generalization due to \citet{IJGT:Gale1984}, as well as a ``rainbow'' generalization of Brouwer's fixed point theorem due to \citet{MP:Bapat1989} are also FIXP-complete. These results could be useful for showing FIXP-completeness for more general versions of the envy-free cake cutting problem, or other ``multi-label'' problems in fair division (e.g., see \citep{aharoni2020fractionally}).

\subsection{Envy-free cake cutting is FIXP-complete -- The proof of \cref{thm:cake-cutting}}

\subsubsection{FIXP-membership}\label{sec:efck-membership}
We start with showing the membership of the problem in FIXP. 
Given a division $x\in\Delta^{n-1}$ of the cake, consider 
the bipartite graph with agents on the left and pieces of cake on the right
and an edge between agent $i$ and piece $j$ if and only if agent $i$ 
prefers piece $j$. It is clear that $x$ is an envy-free division if and only if
this bipartite graph admits a perfect matching. 

\paragraph{Construction of the circuit.}

Using the above idea we construct a circuit whose fixed points are envy-free
divisions. The domain of the circuit is $\Delta^{n-1}$, namely the set of all possible divisions of the cake into $n$ pieces.

\medskip

\noindent We now construct $F\colon \Delta^{n-1}\rightarrow \Delta^{n-1}$. On input $x \in \Delta^{n-1}$, the circuit $F$ outputs $\overline{x} \in \Delta^{n-1}$, which is computed as follows. First, for each agent $i$, we compute capacities $c_i \in \Delta^{n-1}$ for the edges incident on $i$ in the bipartite graph by solving the following LP using the OPT-gate:
\begin{equation*}
\begin{aligned}
\mbox{maximize}\quad & \sum_{j = 1}^{n}z_{j}\cdot u_{ij}(x)\\
\mbox{subject to}\quad & \sum_{j=1}^n z_{j}=1\\
& z \geq 0
\end{aligned}
\end{equation*}
Next, we compute a maximum flow $y \in ([0,1]^n)^n$ in the bipartite graph with edge-capacities given by the previously computed $c_i$'s by solving the following LP using the OPT-gate:
\begin{equation*}
\begin{aligned}
\mbox{maximize}\quad & \sum_{1\leq i,j\leq n}z_{ij}'\\
\mbox{subject to}\quad & 0\leq z_{ij}'\leq c_{ij}+\frac{1}{n^3}, \, \forall i,j\\
 & \sum_{i=1}^{n}z_{ij}'\leq 1,\,\forall j\\
 & \sum_{j=1}^{n}z_{ij}'\leq 1,\,\forall i
\end{aligned}
\end{equation*}
Finally, the circuit computes $r_k = \max(0,1-\sum_{i=1}^{n}y_{ik})$ for every $k$. Then the output $\overline{x} \in \Delta^{n-1}$ of the circuit is computed as follows
\begin{align*}
\overline{x} := \Big(\frac{x_j+r_j}{1+\sum_{k=1}^{n}r_k}\Big)_{1\leq j\leq n}
\end{align*}
We note that both linear programs satisfy the explicit Slater condition (\cref{def:explicit-slater-LP}) that is required
for the OPT-gate to output an optimal solution by \cref{thm:OPT-gate-LP}. In particular, for the second LP this is ensured by the nonzero term ``$1/n^3$'' in the first constraint. Furthermore, the term $1/n^3$ is picked to be sufficiently small so that the arguments in the proof go through (namely, \cref{lem:lemma1}). Any sufficiently small nonzero quantity would also work.

\paragraph{Fixed points.} 

Suppose that $x$ is a fixed point of $F$, i.e., $x = \overline{x}$, where $\overline{x}$ is computed as described above. 
In the following we argue that $x$ is an envy-free division. First, we
argue that it is sufficient to show that the total flow $y$ is $n$.

\begin{lemma}\label{lem:lemma1}
If the total flow $(y_{ij})_{i,j}$ is $n$, then $x$ is an envy-free division. 
\end{lemma}
\begin{proof}
We argue by contradiction, so suppose that $x$ is not an envy-free division. 
This means that there is no perfect matching in the bipartite graph described at
the beginning of this subsection. For a set of agents $A$, let $N(A)$ denote the
set of pieces that is preferred by some agent $i\in A$. By Hall's theorem \citep{hall1935representatives}, there
must exist some set of agents $A$ with $|N(A)|<|A|$. 
In particular the $y$-flow from $A$ to $N(A)$ is at most $|A|-1$.

Now note that the first linear program implies that $c_{ij}=0$ for any $i\in A$ and $j\notin N(A)$, because piece $j$ is not preferred by agent $i$ (i.e., there exists piece $j'$ such that $u_{ij'}(x) > u_{ij}(x)$), and hence we always have $z_{j} = 0$ at any optimal solution of this LP.
As a result, by the first constraint of the second linear program, $y_{ij}\leq 1/n^3$ for any $i\in A$ and $j\notin N(A)$.
We conclude that the total flow from $A$ to $\{1,\dots, n\}\setminus N(A)$ is 
bounded by $\tfrac{1}{n^3}\cdot |A|\cdot |\{1,\dots, n\}\setminus N(A)|
\leq \tfrac{n^2}{n^3}<1$. Since, as shown above, the flow from $A$ to $N(A)$ is at most $|A|-1$, it follows that the total flow out of $A$ is less than $|A|$. 
We conclude that the total flow is strictly less than $n$.
\end{proof}

\noindent By \cref{lem:lemma1}, it now suffices to show that the total flow is $n$. This
is the same as saying that $r_j = 0$ for all $j$. Since $x = \overline{x}$,
we have that $x_j = (x_j+r_j)/(1+\sum_{k = 1}^{n}r_k)$ for all $j$. This
implies that
\begin{align}{\label{ligning}}
x_j \cdot\sum_{k = 1}^{n}r_k = r_j\mbox{ for all }j. 
\end{align}

\noindent We have the following lemma.

\begin{lemma}\label{lem:lemma2}
There exists some $j$ with $r_j = 0$ and $x_j >0$.
\end{lemma}
\begin{proof}
Suppose towards a contradiction that for all $j$ we have that $r_j >0$ or $x_j = 0$. Because
$x_j >0$ for some $j$, this means that the total flow is less than $n$. Therefore, there
exists some agent $i$ whose out-flow is less than $1$. We claim that there exists some $j\in N(i)$ with $r_j = 0$. Suppose towards a contradiction that $r_j >0$ for all $j\in N(i)$. Then we can send
more flow along the edge $(i,j)$ for some $j\in N(i)$, because $\sum_{k = 1}^{n}y_{ik}<1$
and $\sum_{k\in N(i)}c_{ik}=1$. This contradicts the fact that $y$ is an optimal solution 
to the second second linear program. Hence, there exists some $j\in N(i)$ with $r_j = 0$. This concludes the proof, since $x_j >0$ holds by the hungriness-condition.
\end{proof}

\noindent Combining \cref{lem:lemma2} and \cref{ligning}, we conclude that $r_k = 0$
for every $k$. Hence, the total flow is $n$, and it follows from \cref{lem:lemma1}
that $x$ is an envy-free division. Hence, we have shown the following.

\begin{proposition}\label{prop:efmembership}
The envy-free cake cutting problem with very general valuations is in $\FIXP$.
\end{proposition}

\subsubsection{FIXP-hardness}\label{sec:efck-hardness}
\noindent We now turn our attention to showing FIXP-hardness. Since we are considering very general
valuations, the FIXP-hardness result is rather straightforward. For the proof, we will use the following generalization of Brouwer's fixed point theorem due to \citet{MP:Bapat1989}, which we refer
to as \emph{Bapat's Brouwer fixed point theorem}.

\begin{theorem}[Bapat's Brouwer fixed point theorem \citep{MP:Bapat1989}] \label{thm:bapat}
Let $f_{i}\colon\Delta^{n-1}\rightarrow \Delta^{n-1}$ be 
continuous functions for $i=1,\dots, n$. There exists 
some $z\in\Delta^{n-1}$ and a permutation $\pi$ of 
$\{1,\dots, n\}$ such that $z_{\pi(i)}\geq f_{i,\pi(i)}(z)$ for all $i$.
\end{theorem}

\noindent \cref{thm:bapat} implies Brouwer's fixed point theorem
by choosing all the $f_i$'s to be equal. The computational version of 
the problem is defined similarly to the computational version of Brouwer's
fixed point theorem (see \cref{def:FIXP}). Since Bapat's version 
is a generalization of Brouwer's fixed point problem, its FIXP-hardness is immediate.

Let continuous functions $f_1,\dots, f_n\colon \Delta^{n-1}\rightarrow\Delta^{n-1}$
be given, represented by algebraic circuits. To show $\FIXP$-hardness
of the cake cutting problem, we define an instance with $n$ players such that an envy-free division corresponds to a solution for the Bapat problem given by $f_1,\dots, f_n$.
For our definition of the valuations of the agents to satisfy the hungriness-condition, we need to make an assumption about the $f_i,$ namely that $f_{ij}(x)\geq\tfrac{1}{2n}$ for all $i,j\in [n]$ and $x\in\Delta^{n-1}$. We first show that we may assume this without loss of generality.
\paragraph{Preprocessing.}

We now show how to make new functions $g_1,\dots, g_n\colon\Delta^{n-1}\rightarrow\Delta^{n-1}$ such that 1) a solution to the Bapat problem for $g_1,\dots,g_n$ easily translates into a solution for $f_1,\dots, f_n$, and 2) the functions satisfy $g_{ij}(x)\geq\tfrac{1}{2n}$ for all $i,j\in [n]$ and $x\in\Delta^{n-1}$. Let $x\in\Delta^{n-1}$ be given. As in \citep{SICOMP:EtessamiY10}, we may use a sorting network to compute a value $t\geq 0$ such that $\sum_{i=1}^{n}\max(x_i-t,\tfrac{1}{2n})=1$. We now define $\pi(x)$ by $\pi_j(x)=\max(x_j-t,\tfrac{1}{2n})$ and $g_{ij}(x)=\tfrac{1}{2}f_{ij}(2(\pi(x)-\tfrac{1}{2n}))+\tfrac{1}{2n}$. Note the following:
\begin{enumerate}[label=(\roman*)]
    \item For all $i,j\in [n]$, $g_{ij}(x)\geq \tfrac{1}{2n}$. 
    \item If $x$ is a Bapat solution for $g_1,\dots, g_n$, then $2(x-\tfrac{1}{2n})$ is a solution for $f_1,\dots, f_n.$ Suppose that $x$ is a solution for $g_1,\dots, g_n$. Then there is a permutation $\pi$ such that $x_{\pi(i)}\geq g_{i,\pi(i)}(x)$. In particular, we get from (i) that $x_{\pi(i)}\geq\tfrac{1}{2n}$ for all $i$. This implies that the value $t$ computed by the sorting network is $t=0$, which implies that $\pi(x)=x$. Now the inequalities
    \begin{align*}
        x_{\pi(i)}\geq g_{i,\pi(i)}=\frac{1}{2}f_{i,\pi(i)}\left(2\left(\pi(x)-\frac{1}{2n}\right)\right)+\frac{1}{2n} = \frac{1}{2}f_{i,\pi(i)}\left(2\left(x-\frac{1}{2n}\right)\right)+\frac{1}{2n}
    \end{align*}
    for all $i$, imply that $2(x_{\pi(i)}-\tfrac{1}{2n})\geq f_{i,\pi(i)}(2(x-\tfrac{1}{2n}))$ for all $i$. This says that $2(x-\tfrac{1}{2n})$ is a solution to the Bapat problem for $f_1,\dots, f_n$.
\end{enumerate}
Now (i) says that the $g_i$ satisfy the required inequalities and (ii) says that given any solution for $g_1,\dots, g_n$, we may (very easily) compute a solution for $f_1,\dots, f_n$.

\paragraph{Reduction to cake cutting.}
Due to the preprocessing step, we may assume that the functions $f_1,\dots, f_n$ satisfy $f_{ij}(x)\geq \tfrac{1}{2n}$ for all $i,j\in [n]$ and $x\in\Delta^{n-1}$. The valuation
of player $i$ for the $j$-th piece is defined to be $u_{ij}(x)=\max(0,
x_j-f_{ij}(x))$ for any point $x\in\Delta^{n-1}$. We now show that the agents are hungry. Suppose that $x\in\Delta^{n-1}$ with $x_j=0$ for some $j$. This implies that $u_{ij}(x)=0$ for all $j$. As $x_j=0<\tfrac{1}{2n}\leq f_{ij}(x)$ and $x,f_i(x)\in\Delta^{n-1}$, we get that there exists some $k$ with $f_{ik}(x)<x_k$. This implies that $u_{ik}(x)>0=u_{ij}(x)$. In particular, no agent prefers an empty piece.

If $x$ is an envy-free division and $\pi$ is the corresponding permutation, 
we obtain for any $i$ that
\begin{align}\label[ineq]{ineq:efFIXPhard}
\max\left(0,x_{\pi(i)}-f_{i,\pi(i)}(x)\right)\geq \max\left(0,x_j-f_{ij}(x)\right)\mbox{ for all }j,
\end{align}
because of the inequalities $u_{i,\pi(i)}(x)\geq u_{ij}(x)$ for all $i,j$. 
We claim that $x_{\pi(i)}\geq f_{i,\pi(i)}(x)$ for all $i$. Suppose that there exists some $i$ for which this is not the case, i.e., $x_{\pi(i)}< f_{i,\pi(i)}(x)$.
This implies that the left-hand side of \cref{ineq:efFIXPhard} is equal to $0$. However, we also have that 
\[x_{j},f_{ij}(x)\geq 0, \ \ \ \  x_{\pi(i)}<f_{i,\pi(i)}(x), \ \ \ \  \text{and}  \ \ \ \ 
\sum_{j=1}^{n}x_{j} = \sum_{j=1}^{n}f_{ij}(x) = 1,\]
which implies that
$x_{j}>f_{ij}(x)$ for some $j$. This implies that the right-hand side of \cref{ineq:efFIXPhard}
is strictly positive, leading to a contradiction. From this, we obtain the following.

\begin{proposition}\label{prop:efFIXPhard}
The envy-free cake cutting problem with very general valuations is $\FIXP$-hard.
\end{proposition}

\subsection{The {K-K-M} Lemma}\label{sec:kkm}
As we mentioned in the beginning of the section, our technique also has implications to the FIXP-membership of some known topological problems which are related to the cake cutting problem, in the sense that they have been, or can be used to prove the existence of an envy-free division. Here, we present those results. \medskip

\noindent The \emph{Knaster–Kuratowski–Mazurkiewicz ({K-K-M}) lemma}~\citep{FM:KnasterKM1929} is a basic result
concerning certain covers of the unit simplex by closed sets, related to Brouwer's fixed point theorem. First we present the definition of a \emph{K-K-M covering} and then the statement of the lemma.

\begin{definition}[K-K-M covering]
  Let $T_1,\dots,T_n \subseteq \RR^n$ be a collection of closed
  sets. We say that $T_1,\dots,T_n$ form a \emph{{K-K-M} covering} of
  $\Simplex_{n-1}$ if
  $conv(\{e_i \colon i \in S \}) \subseteq \cup_{i\in S} T_i $, for
  any set $S \subseteq [n]$.
\end{definition}

\begin{lemma}[K-K-M lemma \citep{FM:KnasterKM1929}]\label{lem:kkm}
Let $T_1,\dots,T_n$ be a {K-K-M} covering of $\Simplex_{n-1}$. Then $\cap_{i=1}^n T_i \neq \emptyset$.
\end{lemma}

\noindent We can derive a corresponding total search problem that applies to any
collection $T_1,\dots,T_n$ of closed sets as follows.
\begin{corollary}
  Let $T_1,\dots,T_n \subseteq \RR^n$ be a collection of closed
  sets. Then there exists $x \in \Simplex_{n-1}$ such that exactly one
  of the following conditions holds.
  \begin{enumerate}
  \item $x \in \cap_{i=1}^m T_i$
  \item $x \notin \cup_{i \in \support(x)} T_i$
  \end{enumerate}
  \label{COR:KKM-search-problem}
  In case the second condition never holds we say that the sets
  $T_1,\dots,T_n$ satisfy the K-K-M condition.
\end{corollary}
\begin{proof}
  Let us note that for any $x \in \Simplex_{n-1}$ we have that
  $x = \sum_{i \in \support(x)} x_i e_i \in \conv(\{e_i \colon i \in
  \support(x)\})$. Thus, if the second condition does not hold for any
  $x \in \Simplex_{n-1}$ it follows that $T_1,\dots,T_n$ form a
  {K-K-M} covering. By the {K-K-M} lemma we have
  $\cap_{i=1}^n T_i \neq \emptyset$, which means that there exist
  $x \in \Simplex_{n-1}$ satisfying the first condition.
\end{proof}  

\noindent \citet{FM:KnasterKM1929} gave a proof
of Brouwer's fixed point theorem from the K-K-M lemma, whereas
\citet{MS:Gale1955} conversely gave a proof of the K-K-M lemma
from Brouwer's fixed point theorem. To adapt these proofs to a
computational setting, we need to settle on the generality of closed
sets we would like to consider and how to represent those.
It is natural to restrict attention to \emph{closed semi-algebraic sets}
definable in the first order theory of the reals using polynomials
with integer coefficients.

\begin{definition}[Basic closed semi-algebraic set]
  A set $S \in \RR^n$ is a \emph{basic closed semi-algebraic set} if there
  exists polynomials $P_1,\dots,P_k$ such that
  $S=\{x \in \RR^n \mid \wedge_{i=1}^k P_i(x) \geq 0\}$.
\end{definition}

\noindent Any closed semialgebraic set $S$ is a finite union of basic closed
semialgebraic sets, see e.g.~\citep{BochnakCosteRoy1998}, and if $S$ is
definable using polynomials with integer coefficients, this holds also
for the basic closed semialgebraic sets whose union
is~$S$. Disregarding complexity concerns, all such sets can be
represented as an inverse image $F^{-1}(0)$ of a nonnegative function
$F$ computed by an algebraic circuit as shown in the following lemma.

\begin{lemma}
  Let $S \subseteq \RR^n$ be a closed semialgebraic set definable using
  polynomials with integer coefficients. Then there exists an
  algebraic circuit $C$ computing a nonnegative function $F$ such that
  $S = \{x \in \RR^n \mid F(x)=0\}$.
  \label{LEM:ClosedSemialgebraicSet-Circuit}
\end{lemma}
\begin{proof}
  We may write $S$ as a union $S = \cup_{i=1}^s S_i$ of basic closed
  semialgebraic sets
  $S_i=\{x \in \RR^n \mid \wedge_{j=1}^{k_i} P_{ij}(x) \geq 0\}$,
  where each polynomial $P_{ij}$ have integer coefficients.

  Note that
  $S_{ij}=\{x \in \RR^n \mid P_{ij}(x) \geq 0\} = \{x \in \RR^n \mid
  \max(0,-P_{ij}(x))=0 \}$, which means in particular that $S_{ij}$ may
  be represented by as the inverse image $F_{ij}^{-1}(0)$ of a
  nonnegative function $F_{ij}$ computed by an algebraic circuit. We
  complete the proof by showing that the collection of sets
  representable in this way is closed under intersection and union.

  Suppose $F_1$ and $F_2$ are nonnegative functions computed by
  algebraic circuits and let $S_1 = F_1^{-1}(0)$ and
  $S_2 = F_1^{-1}(0)$.  Then the two functions $\max(F_1(x),F_2(x))$ and
  $\min(F_1(x),F_2(x))$ are nonnegative functions computed by
  algebraic circuits that represents $S_1 \cap S_2$ and
  $S_1 \cup S_2$, respectively.
\end{proof}

\noindent We note that if $S$ is any closed set, then the distance $d(x,S)$ is a
nonnegative continuous function that assumes the value~$0$ precisely
in the set $S$. We may thus think of $F$ as a stand-in for the distance
$d(x,S)$ from $x$ to $S$. \\

\noindent We can now define a search problem based on
\cref{COR:KKM-search-problem} and the representation of
closed semialgebraic sets used in
\cref{LEM:ClosedSemialgebraicSet-Circuit}.
\begin{definition}[K-K-M problem]
  Given algebraic circuits $C_1,\dots,C_n$ computing non-negative
  functions $F_1,\dots,F_n : \RR^n \rightarrow \RR$, a point
  $x \in \Simplex_{n-1}$ is a solution of the associated K-K-M problem if one of the following conditions holds.
  \begin{enumerate}
  \item $\forall i: F_i(x)=0$.
  \item $\forall i: x_i>0 \Rightarrow F_i(x)>0$.
  \end{enumerate}
  In case the second condition never holds we say that the functions 
  $F_1,\dots,F_n$ satisfy the K-K-M condition.
  \label{DEF:KKM-Search-Problem}
\end{definition}

\noindent Let us note that while the condition that the functions $F_i$ are
non-negative is a semantic condition, it can be enforced syntactically
when computed by algebraic circuits by considering either $(F_i(x))^2$
or $\max(F_i(x),0)$ depending on whether we wish to represent the set
$F_i^{-1}(0)$ or the set $\{x \in \RR^n \mid F_i(x)\leq 0\}$.

We can now adapt the proof by \citet{MS:Gale1955} to show that the
K-K-M problem is in $\FIXP$, and the proof of \citet{FM:KnasterKM1929} to show 
that the K-K-M problem is also $\FIXP$-hard. We have the following theorem.

\begin{theorem}\label{thm:kkm-fixp-complete}
The K-K-M problem is $\FIXP$-complete.
\label{PROP:KKM-FIXP-complete}
\end{theorem}

\begin{proof}
  We first prove $\FIXP$-membership of the K-K-M problem. Define the
  function $G: \Simplex_{n-1} \rightarrow \Simplex_{n-1}$ by
  \[
    G(x)_i = \frac{x_i + F_i(x)}{1+\sum_{j=1}^n F_j(x)}
  \]
  Given algebraic circuits for all functions $F_i$, an algebraic
  circuit computing $G$ can clearly be constructed in polynomial
  time. Suppose now that $x \in \Simplex_{n-1}$ such that $G(x)=x$. It
  follows that $x_i \sum_{j=1}^n F_j(x) = F_i(x)$ for all $i$. If
  $F_i(x)=0$ for all~$i$, $x$ is a solution satisfying the first
  condition of the K-K-M problem. Otherwise, suppose there is~$i$ such
  that $F_i(x)>0$. Then $\sum_{j=1}^n F_j(x)>0$ as well and it follows
  that
    \[
      x_i = \frac{F_i(x)}{\sum_{j=1}^n F_j(x)}
  \]
  for all~$i$. Hence $x_i>0$ implies $F_i(x)>0$, for all~$i$,
  and $x$ is a solution satisfying the second condition of the K-K-M
  problem.

  To prove $\FIXP$-hardness we reduce from the basic $\FIXP$ problem
  with domain $\Simplex_{n-1}$. Suppose that
  $G: \Simplex_{n-1} \rightarrow \Simplex_{n-1}$ is a continuous
  function computed by an algebraic circuit.
  We may then in polynomial time construct algebraic circuits
  computing the functions $F_i(x) = \max(0,G(x)_i-x_i)$, for
  $i=1,\dots,n$. We claim that these functions satisfy the K-K-M
  property. Suppose for the contrary that $x \in \Simplex_{n-1}$ such
  that $F_i(x)>0$ whenever $x_i>0$. Letting $S = \support{x}$ we then have $G(x)_i>x_i$ for all $i \in S$,  leading to the contradiction
  \[
    1 = \sum_{i \in S}  x_i < \sum_{i \in S} G(x)_i
    \leq \sum_{i=1}^n G(x)_i = 1 \enspace .
  \]
  Thus if $x \in \Simplex_{n-1}$ is a solution of the K-K-M problem it
  satisfies that $F_i(x)=0$, for all $i$, which implies that
  $x_i \geq G(x)_i$, for all $i$. Since also $G(x) \in \Simplex_{n-1}$
  it follows that in fact $x_i = G(x)_i$, for all $i$, which means
  that $x$ is a fixed point of $G$.
\end{proof}

\subsubsection{The rainbow K-K-M lemma and Bapat's Brouwer fixed point generalization}

\noindent Several years after his proof for the K-K-M lemma in \citep{MS:Gale1955}, \citet{IJGT:Gale1984} also proved a generalization of the {K-K-M} lemma, commonly referred to as the \emph{rainbow {K-K-M} lemma}.\footnote{Sometimes in the literature the term ``colorful'' is used instead of ``rainbow''.} 

\begin{lemma}[Rainbow K-K-M lemma \citep{IJGT:Gale1984}]
  Let $T_{i,j=1}^n \subseteq \RR^n$ be a collection of closed sets
  such that for every $i$, the collection $T_{i,1},\dots,T_{i,n}$ form
  a {K-K-M} covering of $\Simplex_{n-1}$. Then there exists a permutation
  $\pi$ of $[n]$ such that $\cap_{i=1}^n T_{i,\pi(i)} \neq \emptyset$.
\end{lemma}  

\noindent \citet{MP:Bapat1989} provided a proof of the rainbow K-K-M lemma via 
the introduction of a generalization of Sperner's lemma. In the same paper, he also
provided his generalization of Brouwer's fixed point theorem (\cref{thm:bapat}), which we used in the proof of \cref{prop:efFIXPhard}. Our reduction in \cref{sec:efck-hardness} also proves the membership of the computational version of Bapat's Brouwer fixed point theorem in FIXP (henceforth \emph{Bapat's Brouwer fixed point problem}), since envy-free cake cutting is in FIXP by \cref{prop:efmembership}.

The reduction of \cref{PROP:KKM-FIXP-complete} from the K-K-M problem to the computational version of Brouwer's fixed point theorem generalizes immediately to the case of reducing from the rainbow K-K-M version (henceforth the \emph{rainbow K-K-M problem}) to Bapat's Brouwer fixed point problem, thus establishing the FIXP-membership of the rainbow K-K-M problem as well. In fact Bapat's proof uses the rainbow K-K-M lemma in the same way as \citet{FM:KnasterKM1929} do to prove Brouwer's
fixed point theorem from the K-K-M lemma. Therefore, we have the following theorem.

\begin{theorem}\label{thm:KKM-and-Bapat}
The rainbow K-K-M problem and Bapat's Brouwer fixed point problem are $\FIXP$-complete.
\end{theorem}

\section{Applications to Markets} \label{sec:markets}

In this section, we show how our main technique can be used to prove the FIXP-membership of computing equilibria in competitive markets. Our first result here is a rather general one, namely that the problem of computing a market equilibrium in Arrow-Debreu markets with concave utilities is in FIXP. This generalizes previously known results, on markets with specific utility functions \citep{garg2016dichotomies} but at the same time is conceptually easier, as long as our OPT-gate is used as a black box. We state the main theorem below.

\begin{theorem}\label{thm:Arrow-Debreu}
The problem of computing a market equilibrium in an Arrow-Debreu market with concave utilities is in FIXP.
\end{theorem}

\noindent Our second result regards the problem of computing an equilibrium in the pseudomarket mechanism of \citet{hylland1979efficient}, which was shown to be in FIXP quite recently by \citet{vazirani_et_al:LIPIcs.ITCS.2021.59}. We obtain the same result, via a conceptually simpler proof based on our general technique. 

\begin{theorem}\label{thm:HZ}
The problem of computing an equilibrium of the Hylland-Zeckhauser mechanism is in $\FIXP$. 
\end{theorem}

\noindent The crucial ingredient that allows us to obtain proofs which are conceptually simpler and often more general is in the utility-maximizing optimization program of the agents, which appears in the corresponding proofs of existence. Generally speaking, given a set of prices, an agent computes a utility-maximizing allocation given a set of constraints on consumption and endowment. Then, market clearing is ensured by the fixed point condition on the prices. Using our OPT-gate, we can directly ``simulate'' the utility-maximization program in these proofs, which makes our membership results look very similar to the existence proofs themselves, which rely on some fixed point correspondence theorem, most often Kakutani's fixed point theorem (\cref{thm:kakutani}). Note however that proving \FIXP-membership using the OPT-gate implicitly also yields a proof of existence based on Brouwer's fixed point theorem (\cref{thm:brouwer}).

\subsection{Arrow-Debreu Markets} \label{sec:markets-arrow-debreu}

The fundamental model of competitive markets was established in the pioneering work of \citet{arrow1954existence}, formalizing some ideas of \citet{walras1874elements}. Arrow and Debreu showed that under relatively mild assumptions, every market has an equilibrium, i.e., a set of prices and allocations such that supply equals demand, and every market participant is maximally satisfied with their assigned commodities at the given prices. We present the formal setting of the Arrow-Debreu market below, following closely their original paper. When $<$ and $\leq$ are used for vectors, they are applied componentwise.

\paragraph{Arrow-Debreu Market.} An Arrow-Debreu market consists of
\begin{itemize}[leftmargin=*]
    \item[-] A set $L$ of \emph{commodities} (i.e., divisible resources or items); let $\ell = |L|$.
    \item[-] A set $N$ of \emph{production units} or \emph{firms}; let $n = |N|$.\medskip
    
    For each firm $j \in N$, there is a set $Y_j$ of possible \emph{production plans}. An element $y_j \in Y_j$ is a vector in $\RR^\ell$, i.e., $y_j = (y_{j,1}, y_{j,2}, \ldots, y_{j,\ell})$, where $y_{j,h}$ is the output (i.e., the produced amount) of commodity $h$ according to plan $y_j$. The production can also be negative, where $y_{j,h}<0$ is interpreted as the commodity $h$ being an input to the production plan, rather than an output. Let $Y=\sum_{j=1}^n Y_j$.
    \medskip
    
    The production sets $Y_j$ satisfy the following assumptions:
    \begin{assumptions}[label={I.\alph*.},leftmargin=30pt]
        \item For all $j \in N$, $Y_j$ is a closed convex subset of $\RR^\ell$ containing $0$, \hfill \emph{(non-increasing returns to scale)}\label{prod:1}
        \item $Y \cap \RRnn^\ell  = \{0\}$, \hfill \emph{(no output without input)}\label{prod:2}
        \item $Y \cap (-Y) = \{0\}$. \hfill \emph{(irreversible production)}\label{prod:3}
    \end{assumptions}
    \item[-] A set $M$ of \emph{consumption units} or \emph{consumers}; let $|M| = m$. \medskip
    
    For each consumer $i \in M$, there is a set $X_i$ of possible \emph{consumption vectors}, indicating the set of resources that a consumer could possibly consume, if there were no budgetary constraints. We use $x_i = (x_{i,1}, x_{i,2}, \ldots, x_{i,\ell})$ to denote the consumption of agent $i$, where $x_{i,h}$ denotes the consumer's consumption of commodity $h$. The consumption can also be negative, where $x_{i,h} < 0$ is interpreted as a labor service provided by the consumer to the market. The consumption sets $X_i$ satisfy:
    \begin{assumptions}[label={II.},leftmargin=30pt]
        \item For all $i \in M$, $X_i$ is a closed, convex subset of $\RR^\ell$ which is bounded from below, i.e., there is a vector $\xi_i$ such that $\xi_i \leq x_i$ for all $x_i \in X_i$.\label{cons-lower}
    \end{assumptions}

    Each consumer has a \emph{utility function} $u_i \colon X_i \to \RR$, that satisfies the following properties:
    \begin{assumptions}[label={III.\alph*.},leftmargin=30pt]
        \item $u_i(x_i)$ is a continuous function on $X_i$, \hfill \emph{(continuity)}\label{util1}
        \item For any $x_i \in X_i$, there is an $x_i' \in X_i$ such that $u_i(x_i') > u_i (x_i)$, \hfill \emph{(non-satiation)}\label{util2}
        \item If $u_i(x_i) > u_i(x_i')$ and $0 < t < 1$, then $u_i\left(tx_i + (1-t)x_i'\right) > u_i(x_i')$.  \hfill \emph{(convexity of preferences)}\label{util3}
    \end{assumptions}
    
    Each consumer $i \in M$ is also endowed with a vector $\zeta_i = (\zeta_{1,i}, \zeta_{2,i}, \ldots, \zeta_{\ell,i}) \in \RR^\ell$ of initial holdings of the different commodities, which we will refer to as the \emph{endowment} of consumer $i$. The following assumption is made:
    \begin{assumptions}[label={IV.a.},leftmargin=30pt]
        \item For all $i \in M$, there exists some $x_i \in X_i$ with $x_i < \zeta_i$.\label{endowment}
    \end{assumptions}    
     
    Additionally, consumer $i$ has a \emph{share} $\alpha_{ij}$ of the profit of the $j$-th production unit, for each $j \in N$. Shares are non-negative and the profits are entirely shared among the consumers, i.e.,
    \begin{assumptions}[label={IV.b.},leftmargin=30pt]
        \item For all $i \in M$ and $j \in N$, $\alpha_{ij} \geq 0$; for all $j \in N$, $\sum_{i=1}^m \alpha_{ij}=1$.\label{shares}
    \end{assumptions}  
\end{itemize}

\noindent Before we proceed, we provide a brief discussion on some of the assumptions of the model above.\footnote{The reader is referred to the original paper by \citet{arrow1954existence} for an extensive discussion on the assumptions of the model.} For the production plans, \cref{prod:1} corresponds to non-increasing returns to scale, i.e., when all production variables are increased by an amount, the output is increased by an ``at-most-proportional'' amount. \cref{prod:2} states that it is not possible for the production to have output without having some input. Finally, \cref{prod:3} asserts that it is not possible to have two production vectors that exactly ``cancel'' each other, since some labor is necessary for production, and labor cannot be produced. For the consumers' utilities, \cref{util1} is a standard assumption in consumer and market theory. \cref{util2} asserts that there is no ``saturation point'', i.e., a consumption vector that the consumer prefers to all others. A related sufficient condition would be that the utility function is \emph{strictly monotone}, although non-satiation is a weaker condition which is still sufficient for the existence of a market equilibrium by \cref{thm:arrow-debreu-original}. \cref{util3} is a standard assumption on the convexity of the indifference curves. \cref{endowment} asserts that each consumer could feasibly consume from their endowment and still have leftover stock to trade in the market. Finally, \cref{shares} simply states that the profits are non-negative and are shared entirely among the consumers, i.e., there are no profits supplied to outside entities. \\

\noindent At the center of competitive markets is the concept of \emph{competitive} or \emph{market equilibrium}, i.e., a set of prices, production plans and consumption vectors such that supply equals demand and all agents maximize their individual utility at the given set of prices. We provide the formal definition below. 

\begin{definition}[Arrow-Debreu Market Equilibrium]
A tuple of vectors $(x_1^{*}, x_2^{*}, \ldots, x_m^{*}, y_1^{*}, y_2^{*}, \ldots y_n^{*}, p^{*})$ is a (Arrow-Debreu) market equilibrium if the following conditions are satisfied:
\begin{conditions}[label={\arabic*.},ref={\arabic*}]
    \item $y_j^{*} \in \arg\max_{y_j \in Y_j} p^{*} \cdot y_j$, for all $j \in N$, \hspace{3cm} \hfill \emph{(firm profit maximization)}\label{market_eq-item1}
    \item $x_i^{*} \in \arg\max_{x_i \in S_i} u_i(x_i)$, where \hfill \emph{(consumer utility maximization)}\\
    $S_i = \{x_i | x_i \in X_i, p^{*} \cdot x_i \leq p^{*} \cdot \zeta_i + \sum_{j=1}^n \alpha_{ij}\, p^{*} \cdot y_j^{*}\}$, for all $i \in M$,\label{market_eq-item2}
    \item $p^{*} \in P=\{p | p \in \RR^\ell, p \geq 0, \sum_{h=1}^\ell p_h =1\}$,  \hspace{3cm} \hfill \emph{(non-negative, normalized prices)}\label{market_eq-item3}
    \item $z^{*} \leq 0$ and $p^{*} \cdot z^{*} = 0$, where $z^{*}  = \sum_{i=1}^m x_i^{*} - \sum_{i=1}^n y_j^{*} - \sum_{i=1}^m \zeta_i$. \hfill \emph{(market clearance)}\label{market_eq-item4}
\end{conditions}
\end{definition}

\noindent Again, we briefly discuss the four conditions of the Arrow-Debreu market equilibrium definition above. \cref{market_eq-item1} requires that at the given set of prices, the firms maximize their profits within their production plans. \cref{market_eq-item2} requires that at the given set of prices, the consumers maximize their utility within their consumption vectors that satisfy their budget constraint. \cref{market_eq-item3} stipulates that the prices are non-negative, and can be normalized to sum to $1$ without loss of generality. Finally, \cref{market_eq-item4} is the market clearing, the ``supply equals demand'' condition. The condition states that (a) the total consumption of each commodity minus the sum of the total production and the consumers' endowment of that commodity has to be non-positive (i.e., consumers cannot consume more than what is available) and (b) this difference is actually zero (i.e., consumers consume exactly what is available), for commodities for which the price is non-zero. Commodities which are priced at zero are allowed to not be entirely consumed. We remark that in the related literature (e.g., see \citep[Chapter 10]{mas1995microeconomic}), \cref{market_eq-item4} is often replaced by the strong condition of ``supply equals demand'' for all commodities, regardless of the price, namely
\[
\sum_{i=1}^m x_{i,h}^{*} = \sum_{i=1}^n y_{j,h}^{*} + \sum_{i=1}^m \zeta_{i,h}, \ \ \ \ \text{for } h=1,\ldots,\ell.
\]
Then, in order to guarantee the existence of a market equilibrium as in \cref{thm:arrow-debreu-original}, the extra assumption of \emph{free disposability} needs to be added to the list of \cref{prod:1,prod:2,prod:3} for the production sets (e.g., see \citep[Chapter 5]{mas1995microeconomic}), namely:
\begin{assumptions}[label={I.d.},leftmargin=30pt]
        \item $Y - \RRnn^\ell \subseteq Y$. \hfill \emph{(free disposability)}\label{free-dispo}
\end{assumptions}  
The property asserts that it is always possible to absorb additional inputs without producing any extra output (i.e., the inputs can be ``freely disposed''). Our approach still works even if one requires this stronger clearing condition, see \cref{rem:market-clearing}. \\

\noindent \citet{arrow1954existence} proved that for any market as defined above, a market equilibrium always exists. 

\begin{theorem}[\citep{arrow1954existence}]\label{thm:arrow-debreu-original}
Every instance of the market above admits a market equilibrium. 
\end{theorem}

\noindent One important observation is that \cref{market_eq-item1,market_eq-item2} are essentially convex programs whereas \cref{market_eq-item3,market_eq-item4} are essentially sets of constraints. Using our OPT-gate from \cref{sec:kkm}, we can ``substitute'' these optimization programs by pseudogates, effectively transforming an existence proof into a FIXP-membership proof. Indeed, the proof of \cref{thm:arrow-debreu-original} as presented in \citep{arrow1954existence} uses these conditions essentially as optimization programs to construct a fixed point correspondence for the consumption, the production and the prices. Then the existence of an equilibrium is proven via a fixed point theorem due to \citet{debreu1952social}. Subsequent proofs used Kakutani's fixed point theorem (\cref{thm:kakutani}) to obtain the same existence result. Our FIXP-membership result essentially devises a proof via Brouwer's fixed point theorem instead, similar in that regard to the proof of \citet{geanakoplos2003nash}. 

Our result in \cref{thm:Arrow-Debreu} generalizes some membership results that were already presented in the literature, for markets with specific utility functions, production and consumption sets. \citet{SICOMP:EtessamiY10} proved the FIXP-membership of a setting where there are no explicit utilities, and the aggregate demand is a given function, rather than a correspondence which is typically the case in these markets.
\citet{garg2016dichotomies} provided a \FIXP-membership result for Arrow-Debreu markets with \emph{Piecewise Linear Concave} (PLC) utility functions, and PLC production sets. The authors consider straightforward consumption sets, namely that the consumption of an agent is non-negative. For the formal definition of PLC functions, we refer the reader to \cref{app:PDC}.

\begin{remark}\label{rem:CES}
Besides the \FIXP-membership results mentioned above, \citet{ChenPY17-non-monotone-markets} proved a \FIXP-membership result for markets with CES utilities, only non-negativity constraints on consumption and no production. While these are a special case of the Arrow-Debreu market that we consider, strictly speaking, our result does not generalize theirs. This is exclusively due to technical reasons, which we highlight below.

As we explain in \cref{sec:arrow-debreu-computation} below, we assume that access to the utility functions is given to us via the supergradients of those functions; this is clearly necessary when dealing with general concave functions without any additional structure. In the case of explicit utility functions however, like PLC utilities or CES utilities, ideally we would like to be able to compute the supergradients given the utility functions, rather than assume access to them. For PLC utilities (and actually much more general utility functions) we can do that, and we show how in \cref{app:PDC}. In that sense, our \FIXP-membership result is a strict generalization of the corresponding membership result proven in \citep{garg2016dichotomies}. On the other hand, CES utilities are not superdifferentiable when some coordinate is $0$, and therefore we cannot claim the same for the results of \citet{ChenPY17-non-monotone-markets}. We believe that it is possible to adapt our approach in \cref{sec:markets-arrow-debreu} to capture the case of CES utilities, but we leave that for future work.
\end{remark}

\subsubsection{The computational Arrow-Debreu market problem}\label{sec:arrow-debreu-computation}
In some of the applications that we have presented so far, deriving the corresponding computational problems from the existence theorems has been relatively straightforward, and thus not explicitly stated. In the market domain, due to the generality of the utility functions, as well as the production and consumption sets, some additional discussion is in order. 
\begin{itemize}[leftmargin=*]
    \item[-] For the utility functions $u_i$, we will assume that we are given a pseudogate computing the supergradients $\partial u_i$ (see \cref{sec:OPT-gate-pseudogate,sec:OPT-gate-convex}). This is in a sense necessary, since we are concerned with general concave utility functions, with no further particular structure. In the case of specific utility functions, such as the PLC utilities studied in  \citep{garg2016dichotomies}, we can easily construct pseudogates computing the supergradients. In \cref{app:PDC}, we explain how to compute the supergradients for a general class of utility functions that subsumes PLC utilities, which we refer to as \emph{Piecewise Differentiable Concave (PDC) functions}, rather than assuming that they are given to us as inputs. 
    \item[-] For the consumption and production sets, we will assume that they are given to us as sets of convex inequalities and linear equations that satisfy the (standard) Slater condition, see \cref{sec:OPT-gate-convex}. Note that here we do not need them to satisfy the explicit Slater condition of \cref{def:explicit-slater-convex}; this is because the consumption and production sets are part of the input to the problem, and therefore we can apply preprocessing to the corresponding constraints to eliminate any linear dependence. For the convex inequalities, we also assume that we are given pseudogates computing their subgradients, as explained in \cref{sec:OPT-gate-convex}. We further assume that the endowments $\zeta_i$, the shares $\alpha_{ij}$, and the lower bounds $\xi_i$ on the consumption are given in the input as rational numbers. 
    \item[-] For the production sets, the assumptions in the Arrow-Debreu model ensure that the set of all $(y_1,\dots,y_n) \in (\RR^\ell)^n$ that satisfy
	\begin{equation}\label{eq:market-allocation-bound}
	y_j \in Y_j,  \ \ \text{for } j = 1,\dots,n \ \ \ \ \text{ and } \ \ \ \ 
	\sum_{i=1}^m \xi_i - \sum_{j=1}^n y_j - \sum_{i=1}^m \zeta_i \leq 0 
	\end{equation}
is bounded. This is proved by \citet[pp. 276-277]{arrow1954existence} and, from a technical standpoint, one of the main purposes of the various assumptions on the production sets is to ensure that this property indeed holds. Note that any production plan $(y_1,\dots,y_n)$ that satisfies the market clearing condition must necessarily satisfy the inequality constraint above. Thus, in any market equilibrium, the production plan lies in this bounded set.

For the definition of the computational problem, we assume that the input to the problem contains an explicit bound $C$ on the set above (in the $\ell_\infty$-norm). This is needed to ensure that solutions can be bounded with some bound that has polynomial bit-complexity. We remark that in the case of the PLC production sets used in \citep{garg2016dichotomies}, the constraint $y \in Y_j$ corresponds to a set of linear equalities and inequalities. As a result, the bound $C$ does not need to be provided in the input in that case, since it can easily be obtained by solving an LP.
\end{itemize}

\subsubsection{Membership in FIXP -- the proof of \cref{thm:Arrow-Debreu}}

\paragraph*{Bounding the domain of allocations.}
Consider any $(x_1,\dots,x_m) \in (\RR^\ell)^m$ and $(y_1,\dots,y_n) \in (\RR^\ell)^n$ that satisfy $x_i \in X_i$ for $i=1,\dots,m$, $y_j \in Y_j$ for $j=1,\dots,n$, and
\[z := \sum_{i=1}^m x_i - \sum_{j=1}^n y_j - \sum_{i=1}^m \zeta_i \leq 0.\] 
By \cref{cons-lower}, we have $\xi_i \leq x_i$, and as a result \eqref{eq:market-allocation-bound} is satisfied. It follows that $\|y_j\| \leq C$ for $j=1,\dots,n$, where $\|\cdot\|$ denotes the $\ell_\infty$-norm. Now, $z \leq 0$ also implies that for any $i \in [m]$
\[x_i \leq \sum_{j=1}^n y_j + \sum_{k=1}^m \zeta_k - \sum_{k \neq i}x_k \leq \sum_{j=1}^n y_j + \sum_{k=1}^m \zeta_k - \sum_{k \neq i}\xi_k\]
where we used $\xi_k \leq x_k$ again. Together with $\xi_i \leq x_i$, it follows that $\|x_i\| \leq nC + m \max_k \|\zeta_k\| + m \max_k \|\xi_k\| =: C'$.

We let $K := C'+ 1$. From the above, we have that if $(x_1,\dots,x_m,y_1,\dots,y_n)$ satisfies $x_i \in X_i$, $y_j \in Y_j$, and $z \leq 0$, then $x_i, y_j \in (-K,K)^\ell$ for all $i,j$. In particular, this must be satisfied at any equilibrium. Note that $K$ can be computed in polynomial time from the inputs to our problem.

\paragraph*{Construction of the circuit.}
We construct a circuit $F: D \to D$, where $D = ([-K,K]^\ell)^m \times ([-K,K]^\ell)^n \times [0,1]^\ell$. We let $(x_1,\dots,x_m,y_1,\dots,y_n,p)$ denote the input to $F$, and $(\overline{x}_1,\dots,\overline{x}_m,\overline{y}_1,\dots,\overline{y}_n,\overline{p})$ denote the output of $F$.

We set $\overline{y}_j$ as the output of the OPT-gate for the convex optimization problem:
\begin{equation}\label{eq:market-CP-firm}
\begin{aligned}
\mbox{maximize}\quad & p \cdot v\\
\mbox{subject to}\quad & v \in Y_j\\
& v \in [-K,K]^\ell
\end{aligned}
\end{equation}

We set $\overline{x}_i$ as the output of the OPT-gate for the convex optimization problem:
\begin{equation}\label{eq:market-CP-consumer}
\begin{aligned}
\mbox{maximize}\quad & u_i(v)\\
\mbox{subject to}\quad & v \in X_i\\
& p \cdot v \leq p \cdot \zeta_i + \sum_{j=1}^m \alpha_{ij} p \cdot y_j\\
& v \in [-K,K]^\ell
\end{aligned}
\end{equation}
The inequality constraint $p \cdot v \leq p \cdot \zeta_i + \sum_{j=1}^m \alpha_{ij} p \cdot y_j$ is called the \emph{budget constraint}.

Finally, we set $\overline{p}$ as the output of the OPT-gate for the LP:
\begin{equation}\label{eq:market-LP-prices}
\begin{aligned}
\mbox{maximize}\quad & v \cdot z\\
\mbox{subject to}\quad & \sum_{h=1}^\ell v_h = 1\\
& v \in [0,1]^\ell
\end{aligned}
\end{equation}
where $z := \sum_{i=1}^m x_i - \sum_{j=1}^n y_j - \sum_{i=1}^m \zeta_i$.

\paragraph*{Fixed points.}
Consider any fixed point of $F$, i.e., $(x_1,\dots,x_m,y_1,\dots,y_n,p) \in D$ such that
\[(x_1,\dots,x_m,y_1,\dots,y_n,p) = (\overline{x}_1,\dots,\overline{x}_m,\overline{y}_1,\dots,\overline{y}_n,\overline{p}).\]
We begin by showing that the explicit Slater condition holds at $(x_1,\dots,x_m,y_1,\dots,y_n,p)$ for all three optimization problems, and thus they are solved correctly by the OPT-gate. Then, we show that $(x_1,\dots,x_m,y_1,\dots,y_n,p)$ is indeed a market equilibrium, as desired.

\paragraph*{Explicit Slater condition.}
Clearly, the constraints of \eqref{eq:market-LP-prices} satisfy the explicit Slater condition. As a result, the OPT-gate functions properly and $p = \overline{p}$ is indeed an optimal solution of this LP.

The constraints of the convex optimization problem \eqref{eq:market-CP-firm} also satisfy the explicit Slater condition. To see this, note that, by assumption, the constraints describing $Y_j$ satisfy the explicit Slater condition, i.e., there exists $v \in Y_j$ that satisfies the inequality constraints of $Y_j$ strictly. Since $0 \in Y_j \cap (-K,K)^\ell$ (\cref{prod:1}) and by convexity of the inequality constraints of $Y_j$, there must exist $v' \in Y_j \cap (-K,K)^\ell$ that satisfies the inequality constraints of $Y_j$ strictly as well. As a result, the OPT-gate functions properly and $y_j = \overline{y}_j$ is an optimal solution of \eqref{eq:market-CP-firm}. In particular, $y_j \in Y_j$.

By \cref{endowment}, there exists $v \in X_i$ with $v < \zeta_i$ (componentwise). Since $\xi_i \leq v$ (\cref{cons-lower}), it also follows that $v \in X_i \cap (-K,K)^\ell$, by construction of $K$. Since $y_j$ is an optimal solution of \eqref{eq:market-CP-firm}, we have $p \cdot y_j \geq 0$, because $0 \in Y_j \cap [-K,K]^\ell$ is a feasible point of \eqref{eq:market-CP-firm}. Thus, $\sum_{j=1}^m \alpha_{ij} p \cdot y_j \geq 0$, since $\alpha_{ij} \geq 0$. Now, since $v < \zeta_i$ and since there exists $h \in [\ell]$ with $p_h > 0$ (by optimality of $p$ for \eqref{eq:market-LP-prices}), it follows that $p \cdot v < p \cdot \zeta_i + \sum_{j=1}^m \alpha_{ij} p \cdot y_j$, i.e., $v \in X_i \cap (-K,K)^\ell$ strictly satisfies the budget constraint in \eqref{eq:market-CP-consumer}.

Recall that the constraints describing $X_i$ are assumed to satisfy the explicit Slater condition, i.e., there exists $v' \in X_i$ that satisfies the inequality constraints of $X_i$ strictly. By the convexity of the inequality constraints of $X_i$, it follows that there exists $v''$ on the segment $[v,v']$ such that $v'' \in X_i \cap (-K,K)^\ell$, $v''$ strictly satisfies the budget constraint in \eqref{eq:market-CP-consumer}, and $v''$ strictly satisfies the inequality constraints of $X_i$. This means that the constraints of \eqref{eq:market-CP-consumer} satisfy the explicit Slater condition and thus the OPT gate for \eqref{eq:market-CP-consumer} works correctly. As a result, $x_i = \overline{x}_i$ is an optimal solution of \eqref{eq:market-CP-consumer}. In particular, $x_i \in X_i$ and $x_i$ satisfies the budget constraint for consumer $i$.

\paragraph*{Market equilibrium.}
By summing up the budget constraints satisfied by each $x_i$, and using the fact that $\sum_{i=1}^m \alpha_{ij} = 1$ (\cref{shares}) we immediately obtain that $p \cdot z \leq 0$. Now, since $p$ is an optimal solution of \eqref{eq:market-LP-prices}, it must be that $z \leq 0$. But then, by construction of $K$, it follows that $x_i,y_j \in (-K,K)^\ell$ for all $i$ and $j$. As a result, $y_j$ is also an optimal solution of \eqref{eq:market-CP-firm} without the constraint $v \in [-K,K]^\ell$. Similarly, by \emph{convexity of preferences} (\cref{util3}), which in particular holds for concave $u_i$, $x_i$ is also an optimal solution of \eqref{eq:market-CP-consumer} without the constraint $v \in [-K,K]^\ell$. We have thus shown that \cref{market_eq-item1,market_eq-item2} hold.

Recall that we have $z \leq 0$ and $p \cdot z \leq 0$. In order to prove that \cref{market_eq-item4} holds, it remains now to show that $p \cdot z = 0$. Clearly, if the budget constraint for each consumer $i$ is tight, then indeed $p \cdot z = 0$. Assume that for some $i$, the budget constraint is not tight, i.e., $p \cdot x_i < p \cdot \zeta_i + \sum_{j=1}^m \alpha_{ij} p \cdot y_j$. By \emph{non-satiation} (\cref{util2}), there exists $x_i' \in X_i$ with $u_i(x_i') > u_i(x_i)$. By \emph{convexity of preferences} (\cref{util3}), there exists $x_i'' \in X_i \cap (-K,K)^\ell$ with $u_i(x_i'') > x_i$ and such that $x_i''$ also satisfies the budget constraint. This is a contradiction to the optimality of $x_i$ for \eqref{eq:market-CP-consumer}. Thus, \cref{market_eq-item4} also holds.

Finally, note that \cref{market_eq-item3} trivially holds, since $p$ is a feasible point for \eqref{eq:market-LP-prices}.
It follows that any fixed point $(x_1,\dots,x_m,y_1,\dots,y_n,p)$ of $F$ is indeed a market equilibrium.

\begin{remark}\label{rem:market-clearing}
If we also assume \emph{free disposability} (\cref{free-dispo}), then there exists a market equilbrium that also satisfies $z=0$, i.e., supply equals demand for all commodities, even those with zero price. We briefly sketch how $F$ can be modified to yield \FIXP-membership for this problem too.

Instead of outputting $(\overline{y}_1,\dots,\overline{y}_n)$, $F$ outputs $(\overline{y}_1',\dots,\overline{y}_n')$, which is set as the output of the OPT-gate for the convex feasibility problem:
\begin{equation*}
\begin{aligned}
\mbox{maximize}\quad & 0\\
\mbox{subject to}\quad & \sum_{j=1}^n v_j = \sum_{i=1}^m x_i - \sum_{i=1}^m \zeta_i\\
& p \cdot v_j = p \cdot \overline{y}_j, \forall j\\
& v_j \in Y_j, \forall j\\
& v_j \in [-K,K]^\ell, \forall j
\end{aligned}
\end{equation*}
Furthermore, in \eqref{eq:market-CP-consumer} and \eqref{eq:market-LP-prices}, we replace $y_j$ by $\overline{y}_j$, including in $z$.

Using the same arguments as above, it follows that $p=\overline{p}$ is an optimal solution of \eqref{eq:market-LP-prices}, $\overline{y}_j$ is an optimal solution of \eqref{eq:market-CP-firm}, and $x_i=\overline{x}_i$ is an optimal solution of \eqref{eq:market-CP-consumer}. As before, this implies that $z \leq 0$ and $p \cdot z = 0$.

Recall that we have redefined $z := \sum_{i=1}^m x_i - \sum_{j=1}^n \overline{y}_j - \sum_{i=1}^m \zeta_i$.
By \emph{free disposability} (\cref{free-dispo}), and since $z \leq 0$, it follows that $w := \sum_{j=1}^n \overline{y}_j + z \in Y$. This means that we can write $w = \sum_{j=1}^n w_j$ where $w_j \in Y_j$, and it holds that $z' := \sum_{i=1}^m x_i - \sum_{j=1}^n w_j - \sum_{i=1}^m \zeta_i = 0$. In particular, $w_j \in (-K,K)^\ell$ by construction of $K$. Furthermore,
\[\sum_{j=1}^n p \cdot w_j = p \cdot w = \sum_{j=1}^n p \cdot \overline{y}_j + p \cdot z = \sum_{j=1}^n p \cdot \overline{y}_j\]
since $p \cdot z = 0$. This implies that $p \cdot w_j = p \cdot \overline{y}_j$ for all $j$, because $\overline{y}_j$ is an optimal solution of \eqref{eq:market-CP-firm}, and $w_j$ is feasible for \eqref{eq:market-CP-firm}. As a result, $(w_1,\dots,w_n)$ is feasible for the feasibility problem above. In particular, the feasible region is nonempty, and thus the OPT-gate ensures that $(\overline{y}_1',\dots,\overline{y}_n')$ is indeed feasible for the feasibility problem.

Since $x_i$ satisfies the budget constraint with $\overline{y}_j$, it also satisfies it with $w_j$ instead. It follows that $(x_1,\dots,x_m,y_1,\dots,y_n,p)$ is a market equilibrium where supply equals demand for all commodities, even those with zero price.
\end{remark}

\subsection{The Hylland-Zeckhauser pseudomarket mechanism} \label{sec:markets-hz}

In 1979, \citet{hylland1979efficient} introduced the \emph{random assignment
problem} (also known as the randomized \emph{house allocation problem}, e.g., see \citep{bogomolnaia2001new,abdulkadirouglu1998random}). In this problem, the goal is to assign a set of 
indivisible goods $G=\{1,\dots, n\}$ 
to a set of agents $A=\{1,\dots, n\}$, who have preferences over the goods. 
These preferences are expressed via cardinal utilities values, i.e., every agent $i$ has a 
utility vector $u_i$, where $u_{ij}$ denotes the utility of 
agent $i$ for good $j$. 

Together with the introduction of the problem, \citet{hylland1979efficient} also proposed a mechanism
for coming up with an allocation that satisfies several desirable properties, namely \emph{ex-ante envy-freeness} and \emph{ex-ante Pareto efficiency}.
The mechanism works  by assigning every good a price $p_j\geq 0$ and every
agent $i$, endowed with one unit of artificial currency, buys \emph{probability shares} $x_i = (x_{i1},\dots, x_{in})$ 
of the goods in a manner that maximizes the agent's expected utility 
$\sum_{j=1}^{n}u_{ij}x_{ij}$ subject to a budget constraint $\sum_{j=1}^{n}p_j x_{ij}\leq 1$ and a feasibility allocation constraint $\sum_{j=1}^{n}x_{ij}= 1$. 
We require that every good is entirely assigned to the agents, that is 
$\sum_{i=1}^{n}x_{ij}=1$ for every $j$, and if this is not the 
case, the mechanism reacts by adjusting the prices. 
A set of prices and an allocation for which 
all the items are entirely allocated (and the artificial budgets are exhausted) is an equilibrium point of the mechanism, which we refer to as a \emph{HZ-equilibrium}.

\begin{definition}[HZ-equilibrium]
A pair $(p,x)$ of a price vector $p$ and an allocation matrix $x$ form a \emph{HZ-equilibrium} for the pseudomarket described above if:
\begin{enumerate}[label={\roman*.}]
\item For every good $j\in G\colon\sum_{i=1}^{n}x_{ij} = 1$.
\item For every agent $i$, the allocation $x_i$ maximizes 
$\sum_{j=1}^n u_{ij}x_{ij}$ subject to the constraints 
$\sum_{j=1}^{n}p_j x_{ij}\leq 1$ and $\sum_{j=1}^{n}x_{ij}= 1$. 
Also, if there are multiple allocations for $i$ that satisfy this, 
then $x_i$ must be the cheapest such allocation.
\end{enumerate}
\end{definition}

From the definition above, it is clear why 
the Hylland-Zeckhauser mechanism is often referred to as a \emph{pseudomarket}. The HZ-equilibrium clearly resembles the \emph{Competitive Equilibrium from Equal Incomes} \citep{foley1966resource,varian1973equity}, except for the allocation constraints which restrict a buyer's allocation to be a total of one unit. This constraint turns out to be rather crucial, as computing a CEEI can be done in polynomial time, whereas the complexity of computing exact HZ-equilibria is still an open problem. 

The reason for requiring that the $x_i$ are the cheapest 
utility-maximizing allocations is that this ensures that a 
HZ-equilibrium $(p,x)$ is ex-ante Pareto optimal. We note 
that $p_j\leq n$ for all $j$, because of the budget 
constraints. Using Kakutani's fixed point theorem, Hylland 
and Zeckhauser proved the following result.

\begin{theorem}[\citet{hylland1979efficient}]
Every instance of the pseudomarket above admits a HZ-equilibrium 
$(p,x)$. Also, $p$ may be chosen such that $p_j = 0$ for some $j$. 
\end{theorem}

\subsubsection{Membership in FIXP -- the proof of \cref{thm:HZ}}

\citet{vazirani_et_al:LIPIcs.ITCS.2021.59} recently provided a proof of 
membership in $\FIXP$ for the problem of computing a 
HZ-equilibrium. 
Their proof defines a map of prices and allocations $F=(F_p,F_1,\dots, F_n)$. Each of
the coordinate functions is defined by a straight-line-program consisting of various steps.
For instance, in one of the steps of $F_i$ the allocation $x_i$ is altered in the case
that agent $i$ has not exhausted its budget and $x_{ij}>0$ for some good $j$ with
$u_{ij}<\max_{k}u_{ik}$. Using a potential function argument, they argue that if $(p,x)$
is a fixed point, then none of the steps of the coordinate functions of $F$ will alter $(p,x)$.
With this in mind, they then argue that if $(p,x)$ is not an equilibrium, then there must
be some step of $F$ that causes a change, implying that $(p,x)$ cannot be a fixed point. 
Next, we present a different proof which makes use 
of our general technique and which we believe to be conceptually simpler.

The domain of our map $F$ is 
$D=[0,n]^n\times ([0,1]^n)^n$. 
The basic idea of the map is that, on input $(p,x)$, we first compute an allocation $x_i'$ for each agent $i$ that 
maximizes the expected utility (but not necessarily the cheapest among those) given prices $p$. 
Then, using $x_i'$, we compute an allocation $y_i$ on $n+1$ goods, which achieves the same utility as $x_i'$, but also minimizes the cost. The extra ``dummy'' good is carefully constructed in such a way that it is never chosen (i.e., $y_{i,n+1} = 0$) at a fixed point of $F$. The function $F$ outputs $(\overline{p}, \overline{x})$, where $\overline{x}_i$ is equal to $\pi(y_i)$, namely the projection of $y_i$ onto the first $n$ coordinates, and $\overline{p}$ is obtained as the solution to a simple LP which ensures that overallocated goods get maximum price, and underallocated goods get price $0$. We note that the input $x$ is not used in the computation of the circuit; the only reason why it is included in the domain of $F$ is so that a fixed point $(p,x)$ of $F$ also includes an optimal allocation.

\paragraph{Description of the map.}

On input $(p,x)$, $F: D \to D$ first computes an allocation $x_i'$ as a solution to the Linear Program:
\begin{align}
\mbox{maximize}\quad & \sum_{j=1}^{n}u_{ij}z_{ij}\label{HZ:LP2}\\
\mbox{subject to}\quad & \sum_{j=1}^{n}z_{ij}=1 \nonumber\\
& \sum_{j=1}^{n}p_j z_{ij}\leq 1 \nonumber\\
& z_{ij}\geq 0,\forall j \nonumber
\end{align}
Before describing the Linear Program computing $y_i$, we
introduce the dummy good that we mentioned earlier. 
For every agent $i$, we let $\delta_i$
denote the difference in utility between its most preferred and 
second most preferred good (if the agent prefers all goods
equally, we let $\delta_i = 1$). The utility of agent $i$ for
the dummy good is now defined as $u_{i,n+1}=u_{i\max}+\delta_i$. 
Furthermore, we let $p_{n+1}=3n$. Now, $y_i$ is computed as a solution to the Linear Program:
\begin{align}
\mbox{minimize}\quad &\sum_{j=1}^{n+1}p_j z_{ij}\label{HZ:LP3}\\
\mbox{subject to}\quad  &\sum_{j=1}^{n+1}z_{ij}=1 \nonumber\\
&\sum_{j=1}^{n+1}{u_{ij} z_{ij}}\geq\sum_{j=1}^{n}u_{ij}x_{ij}' \nonumber\\
& z_{ij}\geq 0,\forall j \nonumber
\end{align}
The second output of $F$ is obtained by setting $\overline{x}_i := \pi(y_i)$, where $\pi(\cdot)$ denotes projection onto the first $n$ coordinates. Finally, we compute a solution $p^{*}$ to
the following Linear Program:
\begin{align}
\mbox{maximize}\quad & \sum_{j=1}^{n}q_j\big(\sum_{i=1}^{n}y_{ij}-1\big)\label{HZ:LP1}\\
\mbox{subject to}\quad & 0\leq q_j\leq n,\,\forall j\nonumber
\end{align}
and we put $\bar{p}=p^{*}-(\min_{j}p_j^{*})\cdot e$ where 
$e=(1,\dots,1).$

\paragraph{Fixed points.}

Suppose that $(p,x)$ is a fixed point of $F$, i.e., $(p,x) = (\bar{p},\bar{x})$. Since by construction $\bar{p}=p^{*}-(\min_{j}p_j^{*})\cdot e$,
there exists some $j$ with $p_j=\bar{p}_j=0$. As a result, there exists a feasible solution to LP~\eqref{HZ:LP2}
that satisfies the inequality constraints strictly: let $z_{ik}=\eps$ for all $k\neq j$ and $z_{ij}=1-(n-1)\eps$ 
for a sufficiently small $\eps > 0$. Also, the equality 
constraints are clearly linearly independent. This implies that the explicit Slater condition is satisfied for LP~\eqref{HZ:LP2}, and therefore $x_i'$ is indeed an optimal solution to this LP. 
Similarly, using that $u_{i,n+1}>u_{ij}$ for all $j\leq n$
we may show that there exists a feasible solution to LP~\eqref{HZ:LP3}
that satisfies all the inequalities strictly. Again, this means that the explicit Slater condition is satisfied for this LP and thus
$y_i$ is an optimal solution.

\begin{lemma}
For every $i$ we have $y_{i,n+1}=0$.
\end{lemma}
\begin{proof}
Suppose towards a contradiction that $y_{i,n+1}>0$. Because $y_i$
is a cost-minimizing allocation and $p_{n+1}>p_{j}$ for all 
$j\leq n$, there exists some $j$ with $y_{ij}>0$ and
$u_{ij}<u_{i\max}$, where $u_{i\max}=\max_{j}u_{ij}$. Otherwise, 
one could achieve a better allocation that consists of only 
one of the goods that has maximal utility for agent $i$.

Let $k$ denote a good of maximal utility for agent $i$, and 
pick $0<\eps<y_{ij},y_{i,n+1}$. Now, define a new allocation
$y_{i}'$ by $y_{i\ell}' = y_{i\ell}$ for $\ell\neq j,k,n+1$, $y_{ik}'=y_{ik}+2\eps$,
$y_{ij}'=y_{ij}-\eps$, and $y_{i,n+1}'=y_{i,n+1}-\eps$. We 
claim that $y_{i}'$ is a strictly better allocation. Clearly, 
$y_i'$ is still a stochastic vector. Also, using the definition 
of $\delta_i$, we have that $u_{ij}\leq u_{i\max}-\delta_i$, so 
we may bound the change in utility as
\begin{align*}
\Delta u = 2\eps u_{i\max}-\eps u_{ij}-\eps u_{i,n+1}=\eps (2u_{i\max}-u_{ij}-(u_{i\max}+\delta_{i}))\geq 0
\end{align*}
We conclude that $y_i '$ also satisfies the utility constraint. 
Using that $p_j,p_k\in [0,n]$ and $p_{n+1}=3n$, we may bound 
the change in price of the allocation as
\begin{align*}
\Delta P = 2\eps p_k-\eps p_j-\eps p_{n+1}\leq 2\eps n-3\eps n<0
\end{align*}
We conclude that $y_i'$ is a strictly better solution, contradicting
the optimality of $y_i$. Therefore we have that $y_{i,n+1}=0$. 
\end{proof}

By the lemma above, we obtain that $\pi(y_i)$ is a cheapest 
utility-maximizing allocation for agent $i$ at prices $p$. In particular,
it satisfies the budget constraint. What
remains to show is that every good $j\leq n$ is allocated fully. 
If that was not the case, we would have that 
there is some $j$ with $\sum_{i=1}^{n}y_{ij}<1$. This means
that a solution $p^{*}$ to LP~\eqref{HZ:LP1} must have $p_j^{*}=0$.
There must also exist some $k$ with $\sum_{i=1}^{n}y_{ik}>1$ 
which implies that $p_k^{*}=n$.  As $p_j^{*}=0$, we obtain 
that $p_k = \overline{p}_k = n$. However, $p_k=n$ together with
$\sum_{i=1}^{n}y_{ik}>1$ implies that $\sum_{i=1}^n p_k y_{ik}>n$
which contradicts the budget constraints of the agents. 
Therefore $\sum_{i=1}^{n}y_{ij}=1$ for every good $j$. 
We conclude that $(p,x) = (\overline{p},\pi(y))$ forms a HZ-equilibrium, therefore proving \cref{thm:HZ}.

\section{Conclusion and Future Work}

In this paper, we introduced the OPT-gate, a powerful tool that can be used as a black box for \FIXP-membership proofs, essentially substituting the various Linear Programs and more general convex optimization programs that often appear in the corresponding existence proofs. We demonstrated the strength of our technique via a set of different applications on quite important and fundamental equilibrium computation problems in game theory, fair division and competitive markets. 

We believe that our technique can be used even more broadly in the future, to enable clean and simple \FIXP-membership proofs of other interesting problems. For example, one could study the equilibria of the various generalizations of the \citeauthor{hylland1979efficient} mechanism, e.g., due to \citet{he2018pseudo} and \citet{echenique2021constrained}; we believe that our OPT-gate can be used to show membership results for those problems as well, but the details need to be worked out. Independently of our technique but related to our results, some important open problems are whether one can show a \FIXP-hardness result for the equilibrium computation problem in Hylland and Zeckhauser pseudomarkets, or in Arrow-Debreu markets. For Hylland and Zeckhauser pseudomarkets, it was recently proved that computing a (weak) approximate equilibrium is PPAD-complete \citep{ChenCPY22-HZ-PPAD}, but the complexity of computing exact equilibria remains open. For Arrow-Debreu markets, as we explained in \cref{sec:related}, the result of \citet{garg2017settling} does not quite yield the \FIXP-completeness result for the market model as presented in \citep{arrow1954existence}. On the other hand, our \FIXP-membership result enables future work to consider rather general markets (with general concave utilities and convex consumption and production sets), in the quest of establishing the desired \FIXP-completeness of the problem.

\appendix

\section{Piecewise Differentiable Concave Functions}\label{app:PDC}

A \emph{piecewise differentiable concave} function is a function that is obtained by taking the lower envelope of differentiable concave functions. Formally, a piecewise differentiable concave function $f \colon \RR^n \to \RR$ is given by
\[f(x) = \min_{j \in [m]} g_j(x)\]
where for each $j \in [m]$, $g_j \colon \RR^n \to \RR$ is differentiable and concave. Note that, as a result, $f$ is also concave and, in particular, admits a superdifferential $\partial f$.

For computational purposes $f$ is represented as follows. For each $j \in [m]$, we are given:
\begin{itemize}
    \item an algebraic circuit computing $g_j$, and
    \item an algebraic circuit computing $\nabla g_j$, the gradient of $g_j$.
\end{itemize}
Clearly, given the circuits for $g_j$ we can easily construct an algebraic circuit that computes $f$.

\begin{lemma}\label{lem:PDC}
Given $f$ represented as above, we can in polynomial time construct a pseudogate computing the superdifferential $\partial f$.
\end{lemma}

Before proving this Lemma, we provide some notable special cases of piecewise differentiable concave function.
\begin{itemize}
    \item \textbf{Piecewise Linear Concave (PLC):} This corresponds to the special case where the functions $g_j$ are linear affine, i.e., $g_j(x) = a_j \cdot x + b_j$. PLC functions are usually represented by the rationals $a_j$, $b_j$ for $j=1,\dots,m$. Circuits for $g_j$ and $\nabla g_j$ can easily be constructed from these. This class of functions also contains Leontief utilities.
    \item \textbf{Piecewise Polynomial Concave:} This corresponds to the special case where the functions $g_j$ are concave polynomials. The polynomials are represented explicitly as a list of monomials. Again, circuits for $g_j$ and $\nabla g_j$ can easily be constructed from these.
\end{itemize}

\begin{proof}[Proof of \cref{lem:PDC}]
We show how to construct a circuit $G_{\partial f} \colon \RR^n \times [0,1]^\ell \to \RR^n \times [0,1]^\ell$ such that $\fix_\ell[G_{\partial f}](x) \subseteq \partial f(x)$ for all $x \in \RR^n$.

On input $x \in \RR^n$, $G_{\partial f}$ first computes $g_j(x)$ and $\nabla g_j(x)$ for $j=1,\dots,m$, using the circuits for $g_j$ and $\nabla g_j$. Then using the OPT-gate, it computes $w \in \RR^n$ as an optimal solution to the following LP:
\begin{equation*}
\begin{aligned}
\mbox{minimize}\quad & \sum_{j=1}^m v_j g_j(x)\\
\mbox{subject to}\quad & \sum_{j=1}^m v_j = 1\\
& v_j \geq 0, \forall j
\end{aligned}
\end{equation*}
Finally, $G_{\partial f}$ outputs $\sum_{j_1}^m w_j \nabla g_j(x)$. Note that the $\ell$ auxiliary inputs/outputs of $G_{\partial f}$ are used to implement the OPT-gate.

Clearly, the LP above satisfies the explicit Slater condition (\cref{def:explicit-slater-LP}), and as a result $w$ is indeed an optimal solution of the LP. In other words, any $z \in \fix_\ell[G_{\partial f}](x)$ can be written as $z = \sum_{j_1}^m w_j \nabla g_j(x)$, where $w$ is an optimal solution to the LP above. By construction, optimality for the LP ensures that $z \in \conv\{\nabla g_j(x) \colon g_j(x) = \min_k g_k(x)\} \subseteq \partial f(x)$. The last containment follows from standard properties of the superdifferential, see, e.g., \citep{Rockafellar70-convex-analysis}.
\end{proof}

\subsubsection*{Acknowledgments}
We thank the anonymous reviewers for comments and suggestions that helped improve the presentation of the paper. We thank Kousha Etessami for discussions about the Kakutani fixed point theorem.
Kristoffer Arnsfelt Hansen and Kasper Høgh were supported by the Independent Research Fund Denmark under grant no. 9040-00433B. Alexandros Hollender was supported by an EPSRC doctoral studentship (Reference 1892947).

\bibliographystyle{plainnat}
\bibliography{FIXP-OPT}

\end{document}